\documentclass[floatfix,twocolumn,aps,superscriptaddress,nofootinbib,pra,10pt]{revtex4-2}

\usepackage[american]{babel}
\usepackage{microtype}
\usepackage{amsmath,amsthm,amssymb,amsfonts}
\usepackage{xfrac}
\usepackage{stmaryrd}
\usepackage[colorlinks=true,linktocpage=true]{hyperref}
\hypersetup{
    allcolors  = {blue},
}

\usepackage{xcolor}
\usepackage[all]{xy}

\usepackage{graphicx}

\usepackage{braket}
\usepackage{float}

\allowdisplaybreaks

\def\emph{\textbf}

\theoremstyle{plain}
\newtheorem{theorem}             {Theorem}

\newtheorem{lemma}      [theorem]{Lemma}

\newtheorem*{theorem*}    {Theorem}
\newtheorem*{proposition*}{Proposition}
\newtheorem*{lemma*}      {Lemma}
\newtheorem*{corollary*}  {Corollary}
\newtheorem*{conjecture*} {Conjecture}

\theoremstyle{definition}
\newtheorem{definition}[theorem]{Definition}

\newtheorem*{definition*}{Definition}
\newtheorem*{example*}   {Example}
\newtheorem*{question*}   {Question}
\newtheorem*{idea*}  {Idea}

\theoremstyle{remark}

\usepackage{proof}

\usepackage{colonequals}
\newcommand{\defeq}{\colonequals} 
 
\usepackage{multirow}  

\usepackage{array}     

\def\implies{\Rightarrow}

\newcommand{\Tr}{\mathrm{Tr}}

\usepackage{graphicx}
\usepackage{multirow}
\usepackage{url}
\usepackage{tcolorbox}

\usepackage{bbold}

\usepackage[normalem]{ulem}

\usepackage{trimclip}
\makeatletter
\DeclareRobustCommand{\circbullet}{\mathbin{\vphantom{\circ}\text{\circbullet@}}}
\newcommand{\circbullet@}{%
  \check@mathfonts
  \m@th\ooalign{%
    \clipbox{0 0 0 {\dimexpr\height-\fontdimen22\textfont2}}{$\bullet$}\cr
    $\circ$\cr
  }%
}
\makeatother
\newcommand{\mixedstate}[1]{#1^\bullet}
\newcommand{\purestate}[1]{#1^\circ}
\newcommand{\anystate}[1]{#1^{\circbullet}}
\newcommand{\wdim}[2]{{#1}^{(#2)}}
\newcommand{\wdimd}[1]{\wdim{#1}{d}}
\newcommand{\cQ}{\mathcal{Q}}
\newcommand{\cC}{\mathcal{C}}
\newcommand{\cP}{\mathcal{P}}
\newcommand{\cN}{\mathcal{N}}
\newcommand{\cR}{\mathcal{R}}
\newcommand{\cI}{\mathcal{I}}
\newcommand{\cZ}{\mathcal{Z}}
\newcommand{\cX}{\mathcal{X}}
\newcommand{\mQt}{\mixedstate{\cQ}_3}
\newcommand{\mCt}{\mixedstate{\cC}_3}
\newcommand{\mPt}{\mixedstate{\cP}_3}
\newcommand{\mNt}{\mixedstate{\cN}_3}
\newcommand{\mRt}{\mixedstate{\cR}_3}
\newcommand{\mIt}{\mixedstate{\cI}_3}
\newcommand{\mZt}{\mixedstate{\cZ}_3}
\newcommand{\mXt}{\mixedstate{\cX}_3}
\newcommand{\pQt}{\purestate{\cQ}_3}
\newcommand{\pCt}{\purestate{\cC}_3}
\newcommand{\pPt}{\purestate{\cP}_3}
\newcommand{\pNt}{\purestate{\cN}_3}
\newcommand{\pRt}{\purestate{\cR}_3}
\newcommand{\pIt}{\purestate{\cI}_3}
\newcommand{\pZt}{\purestate{\cZ}_3}
\newcommand{\pXt}{\purestate{\cX}_3}
\newcommand{\aQt}{\anystate{\cQ}_3}

\newcommand{\aPt}{\anystate{\cP}_3}
\newcommand{\aNt}{\anystate{\cN}_3}
\newcommand{\aRt}{\anystate{\cR}_3}
\newcommand{\aIt}{\anystate{\cI}_3}
\newcommand{\aZt}{\anystate{\cZ}_3}
\newcommand{\aXt}{\anystate{\cX}_3}

\usepackage{mathtools}
\usepackage{enumitem}

\newcommand{\stkout}[1]{\ifmmode\text{\sout{\ensuremath{#1}}}\else\sout{#1}\fi}

\makeatletter
\def\p@subsection{}
\def\p@subsubsection{}
\makeatother

\renewcommand{\emph}{\textit}

\begin{document}

\title{Quantum circuits for measuring weak values, \texorpdfstring{\\}{} Kirkwood--Dirac quasiprobability distributions, and state spectra}

\newcommand{\bina}{Faculty of Engineering and the Institute of Nanotechnology and Advanced Materials, Bar Ilan University, Ramat Gan 5290002, Israel}
\newcommand{\binashort}{Faculty of Engineering and the Institute of Nanotechnology and Advanced Materials, Bar Ilan University, Ramat Gan, Israel}
\newcommand{\inl}{INL -- International Iberian Nanotechnology Laboratory, Av. Mestre Jos\'{e} Veiga s/n, 4715-330 Braga, Portugal}
\newcommand{\inlshort}{INL -- International Iberian Nanotechnology Laboratory, Braga, Portugal}
\newcommand{\uff}{Instituto de F\'{i}sica, Universidade Federal Fluminense, Av. Gal. Milton Tavares de Souza s/n, Niter\'{o}i -- RJ, 24210-340, Brazil}
\newcommand{\uffshort}{Instituto de F\'{i}sica, Universidade Federal Fluminense, Niter\'{o}i -- RJ, Brazil}
\newcommand{\cfum}{Centro de F\'{i}sica, Universidade do Minho, Campus de Gualtar, 4710-057 Braga, Portugal}
\newcommand{\cfumshort}{Centro de F\'{i}sica, Universidade do Minho, Braga, Portugal}
\newcommand{\huji}{School of Computer Science and Engineering, Hebrew University, Jerusalem 9190401, Israel}
\newcommand{\hujishort}{School of Computer Science and Engineering, Hebrew University, Jerusalem, Israel}
\newcommand{\uob}{H. H. Wills Physics Laboratory, University of Bristol, Tyndall Avenue, Bristol BS8 1TL, United Kingdom}
\newcommand{\uobshort}{H. H. Wills Physics Laboratory, University of Bristol, Bristol, United Kingdom}
\newcommand{\diumshort}{Departamento de Inform\'atica, Universidade do Minho, Braga, Portugal}

\author{Rafael Wagner}
\email{rafael.wagner@inl.int}
\affiliation{\inlshort}
\affiliation{\cfumshort}

\author{Zohar Schwartzman-Nowik}
\affiliation{\hujishort}
\affiliation{\binashort}

\author{Ismael L. Paiva}
\affiliation{\uobshort}

\author{Amit Te'eni}
\affiliation{\binashort}

\author{Antonio Ruiz-Molero}
\affiliation{\inlshort}
\affiliation{\diumshort}

\author{Rui Soares Barbosa}
\affiliation{\inlshort}

\author{Eliahu Cohen}
\affiliation{\binashort}

\author{Ernesto F. Galv\~ao}
\email{ernesto.galvao@inl.int}
\affiliation{\inlshort}
\affiliation{\uffshort}
\date{\today}

\begin{abstract}
Weak values and Kirkwood--Dirac (KD) quasiprobability distributions have been independently associated with both foundational issues in quantum  theory and advantages in quantum metrology.
We propose simple quantum circuits to measure weak values, KD distributions, and spectra of density matrices without the need for post-selection.
This is achieved by measuring unitary-invariant, relational properties of quantum states, which
are functions of Bargmann invariants, the concept that underpins our unified perspective.
Our circuits also enable experimental implementation of various functions of KD distributions, such as out-of-time-ordered correlators (OTOCs) and the quantum Fisher information in post-selected parameter estimation, among others.
An upshot is a unified view of nonclassicality in all those tasks. In particular, we discuss how negativity and imaginarity of Bargmann invariants relate to set coherence.
\end{abstract}

\maketitle

\tableofcontents

\section{Introduction}\label{sec: intro}

Two concepts have profoundly impacted quantum foundations,  metrology, and thermodynamics: weak values, introduced in the seminal work of Aharonov, Albert, and Vaidman~\cite{aharonov1988result}, and the quasiprobability distribution introduced by Kirkwood and Dirac~\cite{kirkwood1933quantum, dirac1945analogy}.
These were recently connected with one another \cite{yunger_Halpern2018quasiprobability}, and both can be experimentally measured using weak measurement schemes involving pre- and post-selection of carefully chosen observables~\cite{lostaglio2022kirkwood}.
Despite such great impact, investigations of measurement schemes that can be implemented using currently relevant quantum circuit architectures without post-selection and without weak coupling are only now beginning to appear~\cite{cohen2018determination,lostaglio2022kirkwood}.

Another important carrier of nonclassical information about a quantum state $\rho$, seemingly unrelated to those discussed above, is the spectrum $\text{Spec}(\rho)$.
Spectral properties are well known to capture nonclassical features of states such as entanglement~\cite{horodecki2009quantum} and basis-dependent coherence~\cite{streltsov2017colloquium}.
Various nonclassicality properties can be learned from the spectrum of a quantum state, or from non-linear functions thereof, such as univariate traces of the form $\text{Tr}(\rho^n)$ for some integer $n\geq1$.
In quantum information, knowledge of such univariate traces has recently been used to witness non-stabilizerness (i.e., ``magic'') of quantum states~\cite{tirrito2023quantifying}, or in subroutines for variational quantum eigensolvers~\cite{wang2021variational}.

In this article, we describe a unified framework that enables expressing Kirkwood--Dirac (KD) quasiprobability functions, weak values, and univariate traces $\text{Tr}(\rho^n)$ in terms of more fundamental quantities, known as \emph{Bargmann invariants}.
These invariants fully characterize the relational properties of a set of quantum states, that is, all the properties that remain invariant under the application of the same unitary to all states in the set. The simplest non-trivial Bargmann invariant is the overlap $\text{Tr}(\rho\sigma)$ between two states $\rho$ and $\sigma$, and it is also the easiest to probe experimentally~\cite{giordani2021witnesses}. 

The starting point for our conceptual unification is the observation that KD distributions, weak values, univariate traces, and many other constructions of interest can be written as functions of Bargmann invariants. From this simple yet powerful observation, we draw foundational and practical implications.
On the one hand, the connection with Bargmann invariants provides a unified view for studying nonclassicality of sets of states for all use cases considered in this work.
On the other, it opens the door for the use of a family of circuits that measure Bargmann invariants to estimate \textit{any} function of Bargmann invariants, including all those mentioned above.

A crucial notion of nonclassicality we consider here is that of \emph{set coherence},
a basis-independent notion of coherence proposed in Ref.~\cite{designolle2021set}, whereby a set of quantum states is said to be coherent if and only if it is not pairwise commuting. We show how this form of nonclassicality connects with the theory of Bargmann invariants. 

KD distributions can exhibit a different kind of nonclassicality by taking on negative or even  non-real complex values. Similarly, the weak value of an observable is considered nonclassical when it does not lie within the observable's spectrum.
As it turns out, the nonclassicality of KD distributions and weak values, which underpins their relevance as quantum information resources, is a relational property described by Bargmann invariants.
We show that, in general, learning the values of higher-order invariants (beyond overlaps) is \textit{necessary} to assess nonclassicality. We also show how further assumptions can be used to ascertain nonclassicality using overlaps only.

More pragmatically, the nonclassicality of weak values and KD distributions has been linked to quantum advantage in metrology~\cite{arvidsson_Shukur2020quantum} and to quantification of quantum information scrambling~\cite{gonzalez2019out,alonso2022diagnosing}.
Our framework describing those quantities via Bargmann invariants allows us to use recently proposed cycle test circuits \cite{oszmaniec2021measuring, quek2022multivariate} to directly measure weak values and KD distributions.
This enables quantum circuit measurements of  {functions} associated with multiple applications;
see Fig.~\ref{fig:concept} for a conceptual scheme describing our main contributions.
We compare the performance of these circuits with the usual strategies for performing weak measurements and quantum state tomography. We also show how out-of-time-ordered correlators (OTOCs), used to quantify information scrambling in quantum dynamics, can be measured using the same type of circuit.
We show that this is also true for the quantum Fisher information obtained in post-selected parameter estimation.

We recall how these circuits can also estimate the spectrum of a given $d$-dimensional quantum state $\rho$ via estimation of $\Tr(\rho^n)$ for $n=1,\hdots,d$. This circuit architecture is well known~\cite{yirka2021qubitefficient}, as is the fact that one can learn spectral properties by classical post-processing using the Faddeev--LeVerrier algorithm~\cite{escofier2012galois}. Besides our unified perspective that such quantities correspond to specific Bargmann invariants, we provide new sample complexity arguments and numerical simulations suggesting that learning the spectrum via this technique is useful for near-term applications, and is significantly simpler when compared with other efficient techniques.

We expect our contributions to open the path to the quantification of nonclassicality and many other applications, for example in the characterization of coherence, entanglement, and quantum computational advantage.

\begin{figure*}
    \includegraphics[width=\textwidth]{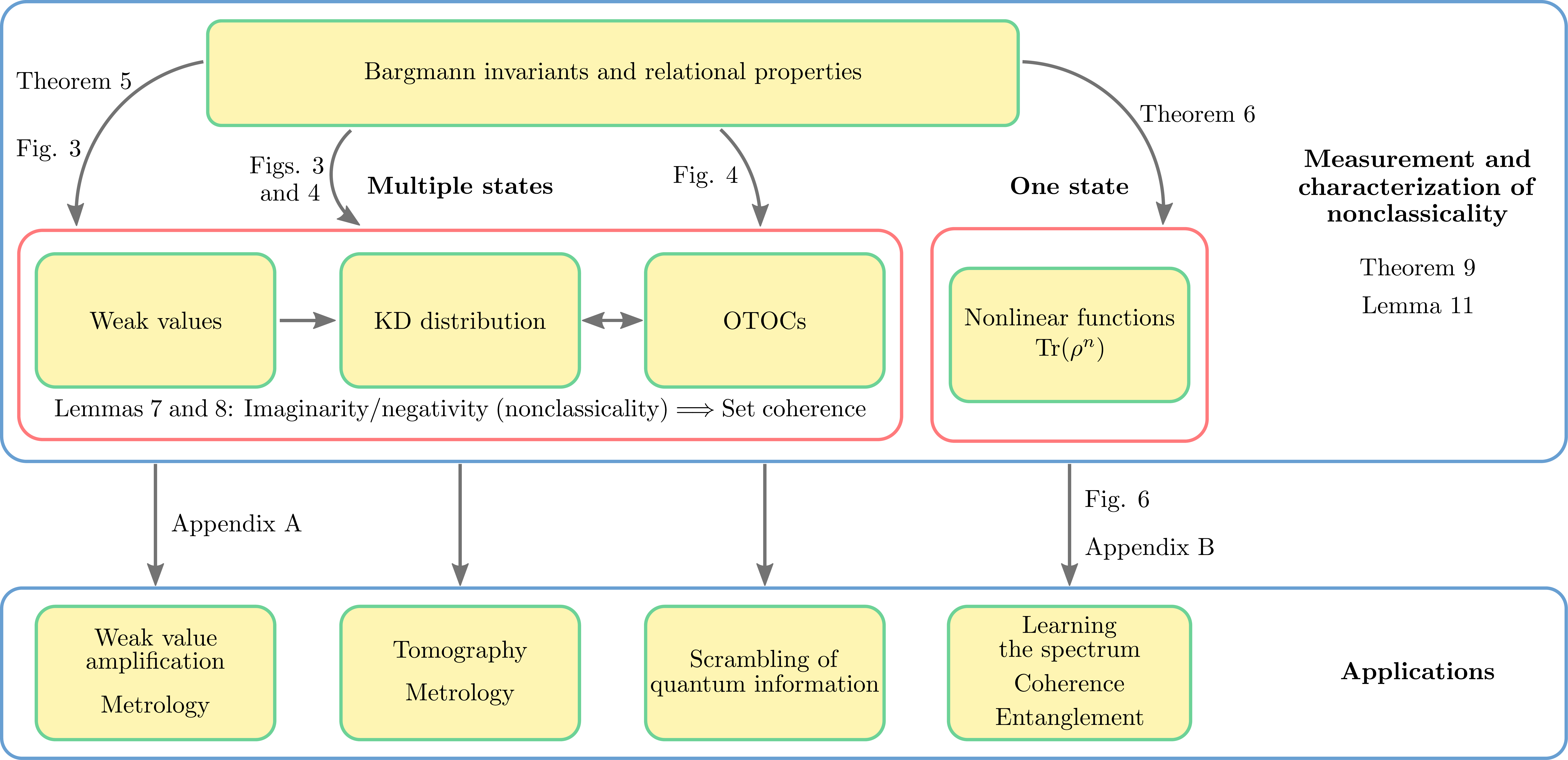}
    \caption{\textbf{Conceptual scheme describing the content and main contributions and their inter-relations.} We show how weak values, the Kirkwood--Dirac (KD) quasiprobability distribution, and out-of-time-ordered correlators (OTOCs) can all be expressed as unitary-invariant quantities known as Bargmann invariants in Sec.~\ref{sec: unifying}. Circuits to measure them are discussed in Sec.~\ref{sec: circuits}. Moreover, results on the nonclassicality of Bargmann invariants, such as those obtained in Sec.~\ref{sec: nonclassicality}, enable a better understanding and quantification of quantum advantage, where we show that negativity and imaginarity of Bargmann invariants require set coherence.}
    \label{fig:concept}
\end{figure*}

\textbf{Outline.} This paper is organized as follows. In Sec.~\ref{sec: preliminaries} we introduce the basic definitions and properties of KD distributions, weak values arising from weak measurements, extended KD distributions, and Bargmann invariants. In Sec.~\ref{sec: unifying} we then propose Bargmann invariants as the unifying concept connecting all the constructions just described.
We proceed to discuss the pragmatic quantum circuit measurement schemes in Sec.~\ref{sec: circuits}, and the foundational implications for analyzing nonclassicality in Sec.~\ref{sec: nonclassicality}. Specifically, in Sec.~\ref{sec: circuits}, we present our main results, describing how one can experimentally access weak values, KD distributions, and state spectra via the  measurement of Bargmann invariants. We compare the efficiency and experimental particularities of estimating weak values using the standard weak measurement approach against our proposed quantum circuits. We also compare our scheme with other approaches to estimate the KD distribution. Moreover, we study how the measurement of higher-order Bargmann invariants enables the estimation of the quantum Fisher information obtained in post-selected parameter amplification, OTOCs, and the spectrum of a density matrix. In Sec.~\ref{sec: nonclassicality}, we start by formally connecting Bargmann nonclassicality with set coherence, hence connecting KD and weak value nonclassicality with various other notions and results in the literature. We then study how nonclassical properties of invariants help characterize nonclassical quasiprobability distributions. We finish this section by analyzing minimal conditions for characterizing the nonclassicality of KD distributions and related quantities.
We conclude in Sec.~\ref{sec: conclusion} with an outlook on possible future work.

\section{Preliminaries}\label{sec: preliminaries}

In this section, we recall the necessary background on weak values, (extended) KD distributions, and Bargmann invariants.

\subsection{Weak values}\label{sec: weak values}

Consider a system prepared (pre-selected) in state $\vert\psi\rangle$, on which one performs a weak measurement of an observable $A$, i.e., a measurement associated with a small coupling strength (when compared to the standard deviation of the measuring pointer). Although this measurement has the apparent downside of not generating a significant average shift of the pointer compared to its standard deviation, it has the benefit of causing little disturbance to the system of interest. Finally, let the system be post-selected in the state $|\phi\rangle$. It turns out that for this post-selection, the average shift of the pointer used for the intermediate weak measurement is proportional to
\begin{equation}
    A_w \defeq \frac{\langle\phi| A |\psi\rangle}{\langle\phi|\psi\rangle},
    \label{eq:def-wv}
\end{equation}
which is known as the weak value of $A$ \cite{aharonov1988result}. $A_w$ arises from a first-order approximation of a certain Taylor expansion \cite{dressel2014colloquium}. This quantity may lie outside the spectrum of the measured operator, either by having a non-zero imaginary part or by having a real part outside the spectrum of $A$, in which case the weak value is said to be \emph{anomalous}.
In fact, with appropriate choices of pre- and post-selections, anomalous weak values can have arbitrarily large real or imaginary parts.  For future reference, we state this as a definition.

\begin{definition}[Weak value nonclassicality] \label{def: nonclassical weak values}
Let $A: \mathcal{H}\to \mathcal{H}$ be a
Hermitian operator and $\ket{\psi}, \ket{\phi} \in \mathcal{H}$ two non-orthogonal vectors in a finite-dimensional Hilbert space $\mathcal{H}$.
We say that the weak value $A_w = \langle \phi | A | \psi \rangle / \langle \phi | \psi \rangle$ is \emph{nonclassical} or \emph{anomalous} if $A_w \notin \text{Spec}(A)$, where $\text{Spec}(A)$ denotes the spectrum of $A$.
\end{definition}

Due to the appearance of such anomalous values and the usual way in which these quantities are experimentally obtained, via weak measurements, the quantum nature of weak values has been questioned since their introduction; see Ref.~\cite{vaidman2017weak} for an overview.
More recently, this situation seems to have changed, with the nonclassicality of certain weak values being more well established. Important theoretical results contributing to this understanding state that anomalous weak values obtained from weak measurements constitute proofs of generalized contextuality~\cite{pusey2014anomalous, kunjwal2019anomalous}.

Since weak measurements cause little disturbance to the measured system, they also do not extract much information about it. As a result, protocols involving weak measurements typically require large ensembles in order to decrease the variance associated with them; see, however, the recent demonstration in Ref.~\cite{rebufello2021anomalous}.
In spite of this, weak values have found various practical applications. In particular, through a method known as \emph{weak value amplification}, large weak values are used for realizing extremely sensitive measurements~\cite{dixon2009ultrasensitive, turner11, susa2012optimal, dressel2013strengthening, jordan2014technical, pang2014entanglement, alves2015weak, harris2017weak, pfender19, cujia19, fang2021weak, huang2021amplification, paiva2022geometric}. This technique is especially helpful in the presence of technical noise or detector saturation~\cite{jordan2014technical, harris2017weak}.

It is also noteworthy that weak values can be used to estimate wave functions, i.e., a state's representation in position coordinates $\Psi(x) \coloneqq \langle x | \psi \rangle$, or specific density matrix terms~\cite{thekkadath2016direct, lundeen2012procedure, lundeen2011direct, shi2015scanfree, malik2014direct}.
Observe, however, that the present work focuses exclusively on the finite-dimensional setting.

Considering the spectral decomposition of the observable $A$,
$A = \sum_{a \in { \text{Spec}}(A)}a\vert a\rangle\langle a\vert$, with $\text{Spec}(A)$ denoting the spectrum of $A$, we see that
\begin{equation}
    A_w = \sum_{a\in{ \text{Spec}}(A)}\! a \frac{\langle \phi \vert a \rangle \langle a \vert \psi \rangle}{\langle \phi \vert \psi \rangle } =  \sum_{a\in{ \text{Spec}}(A)}\! a \frac{\langle \phi \vert a \rangle \langle a \vert \psi \rangle \langle \psi \vert \phi \rangle }{\vert \langle \phi \vert \psi \rangle\vert^2 }
    \label{eq:wv-expanded}
\end{equation}
since, by construction, $\langle \psi \vert \phi \rangle \neq 0$.
Other generalizations of weak values, such as sequential weak values or joint weak values, provide valuable information regarding the behavior of a given state with respect to two---possibly incompatible---observables~\cite{lundeen2005practical, lundeen2009experimental, martinezBecerril2021theoryexperiment, kim2018direct}.
It so happens that the numerators in the expression for the weak value above, or those of joint weak values, correspond to values of the KD quasiprobability distribution at the phase-space point determined by the pre- and post-selected states~\cite{johansen2004nonclassicality}.
These have their own relevance in quantum information theory~\cite{yunger_Halpern2018quasiprobability}.
We now proceed to review some basic facts on the KD quasiprobability distribution.

\subsection{Kirkwood--Dirac quasiprobability distribution}\label{subsec: preliminaries KD}

Consider a finite discrete phase space $I \times F${,} associated with quantum states $\{\ket{i}\}_{i\in I}$ and $\{\ket{f}\}_{f\in F}$ that form two orthonormal bases of a $d$-dimensional Hilbert space $\mathcal{H}$.
The \emph{KD quasiprobability distribution} for a given state $\rho\in\mathcal{D}(\mathcal{H})$\footnote{We denote by $\mathcal{D}(\mathcal{H})$ the set of all quantum states in the Hilbert space $\mathcal{H}$.} {is the function $\xi(\rho\vert\cdot)\colon I \times F \to \mathbb{C}$} given by
\begin{equation}\label{eq: KD-definition}
    \xi(\rho \vert i,f) \defeq \langle f \vert i \rangle \langle i \vert \rho \vert f \rangle    
\end{equation}
at each phase-space point $(i,f) \in I \times F$.
Under the assumption that $\langle f \vert i \rangle \neq 0$ for all $i,f$, the KD distribution provides complete information about the state $\rho$.
This is because, in this case, the values of the KD distribution at phase space points are the decomposition coefficients with respect to the orthonormal basis $\left\{\vert i \rangle \langle f \vert/\langle f\vert i \rangle \right\}_{(i,f)\in I\times F}$ of $\mathcal{B}(\mathcal{H})$, the space of bounded operators on $\mathcal{H}$~\cite{johansen2007quantum,johansen2008quantum}. The KD distribution has been experimentally measured in Refs.~\cite{bamber2014observing, thekkadath2017determining, hofmann2012how, salvail2013full, hernandezgomez2023projective}.

The KD distribution can be extended to be well defined for a larger number of bases, or even general projection-valued measures (PVMs).
For a family of PVMs $M_i = \{\Pi^{i}_k\}_{k \in K_i}$, with $i=1,\dots,n$, the \emph{extended KD distribution} for a state $\rho$ at a phase space point $(k_1,\dots,k_n) \in K_1 \times\dots\times K_n$ reads 
 \begin{equation}\label{eq: KD extended}
     \xi(\rho \vert k_1,\dots,k_n) = \text{Tr}(\Pi^1_{k_1}\Pi^2_{k_2}\dots\Pi^n_{k_n}\rho).
 \end{equation}

The following definition formalizes what is understood by nonclassicality in the context of KD distributions. 

\begin{definition}[KD nonclassicality]\label{def: KD nonclassicality}
Fix $\{\vert i \rangle\}_{i\in I}$, $\{\vert f \rangle\}_{f\in F}$ two reference orthonormal bases of a finite-dimensional Hilbert space $\mathcal{H}$.
We say that the KD quasiprobability distribution $\xi(\rho\vert\cdot)$ associated with a state $\rho \in \mathcal{D}(\mathcal{H})$ is \emph{nonclassical}
when there exists a phase-space point $(i,f) \in I\times F$ such that $\xi(\rho \vert i,f) \notin [0,1]$.     
\end{definition}

For extended KD quasiprobability distributions, we similarly define KD nonclassicality to mean taking on values outside the unit interval.

\subsubsection{Post-selected quantum Fisher information} \label{subsec: prelminaries post-selected QFI}

The KD distribution and its extended variants are key constructions for witnessing many nonclassical properties of quantum dynamics~\cite{yunger_Halpern2018quasiprobability, lostaglio2022kirkwood, budiyono2023quantifying}.
Recent developments have shown that the KD distribution is deeply connected to the quantum Fisher information~\cite{lostaglio2022kirkwood}.
The latter provides the optimal rate with which one can learn some (set of) parameter(s) $\theta$ encoded in quantum states $\rho_\theta$ via the Cram\'er--Rao bound~\cite{cramer1999mathematical, rao1945information}. For good introductory works on Fisher information theory and quantum or classical estimation theory, see, e.g., Refs.~\cite{ly2017tutorial, liu2019quantum, paris2009quantum}. It is therefore interesting from a foundational perspective to seek ways of estimating the quantum Fisher information, an important task for the noisy intermediate-scale quantum (NISQ) era of technological capabilities~\cite{Preskill2018, Bharti2022nisq, bharti2021fisher}. Indeed, this is currently an active research topic~\cite{meyer2021fisher, zhang2022direct}.

Let us focus on a recently found connection with quantum metrology. In Ref.~\cite{arvidsson_Shukur2020quantum},  nonclassicality of the KD distribution {was connected} with metrological advantage in parameter estimation.
In a prepare-and-measure metrological scenario {on a $d$-dimensional Hilbert space $\mathcal{H}$}, an initial pure state {$\rho = \vert \psi \rangle \langle \psi \vert$} is prepared, information about a parameter $\theta$ is encoded through the action of a unitary $U = e^{-i\theta I}$ with $I = \sum_i \lambda_i \vert i \rangle \langle i \vert$, evolving to a final state $\rho_ \theta^{\text{ps}}$, which is successfully post-selected with respect to a projector $F = \sum_{f}\vert f \rangle \langle f \vert $ with probability $p_\theta^{\text{ps}} = \text{Tr}(F\rho_\theta)$.
The quantum Fisher information $\mathcal{I}^{\text{ps}}$ captures the optimal rate one can learn about the parameter $\theta$ in this setup.
The following relationship was established in Ref.~\cite{arvidsson_Shukur2020quantum} between the extended KD distribution and the quantum Fisher information of post-selected states:
\begin{equation}\label{eq: ps QFI}
    \mathcal{I}^{\text{ps}}= 4 \sum_{i,i',f}\lambda_i\lambda_{i'}\frac{\xi( \rho_\theta \vert i,i',f)}{p_{\theta}^{\text{ps}}}-4\left\vert\sum_{i,i',f}\lambda_i\frac{\xi( \rho_\theta \vert i,i',f)}{p_{\theta}^{\text{ps}}} \right\vert^2.
\end{equation}
The extended KD distribution for this scenario takes the form
\begin{equation}\label{eq: KD-PPA}
    \xi(\rho_\theta\vert i,i',f) = \langle i \vert \rho_\theta \vert i' \rangle \langle i' \vert f \rangle \langle f \vert i \rangle .
\end{equation}
{For a generalization of Eq.~\eqref{eq: ps QFI} to the case of multi-parameter estimation and non-ideal observables $F$, see Refs.~\cite{salvati2023compression,jenne2022unbounded}.}
Importantly, Eq.~\eqref{eq: ps QFI} holds only for pure states {$\rho_\theta = \vert \psi_\theta \rangle \langle \psi_\theta \vert$}, and was later generalized in Ref.~\cite{lupu_Gladstein2022negative}. Moreover, in Ref.~\cite{lupu_Gladstein2022negative}, the relationship {displayed} in Eq.~\eqref{eq: ps QFI} { was used to unravel} metrological advantages of post-selected parameter amplification for the estimation of a small parameter $\theta$. Also, in Ref.~\cite{das2023saturating}, the efficiency of the {post-selection} advantage was bounded, formally relating the growth of the post-selected version of Fisher information with the factor $p_\theta^{\text{ps}}$.

As a remark, the KD distribution also appears as part of the quantum Fisher information without the need for post-selection in the so-called linear response regime, as discussed in Ref.~\cite{lostaglio2022kirkwood}.\newline

\subsubsection{Out-of-time-ordered correlators} \label{subsec: prelminaries OTOC}

In Ref.~\cite{yunger_Halpern2018quasiprobability} it was shown that the extended KD distribution underlies the out-of-time-ordered correlators (OTOCs), commonly used to witness the scrambling of quantum information. These are given by the expression
\begin{equation}
    \text{OTOC}(t) \defeq \text{Tr}(W^\dagger(t)V^\dagger W(t)V\rho).
\end{equation}
Intuitively, OTOCs witness scrambling of information in the following situation. Consider a many-body system $\rho$ and two observables {$W=W(0)$} and $V$ acting over { distant} regions of the system; the canonical example considers {local} observables acting on the initial and final spins in a one-dimensional lattice of size $N\gg1$. The function $\text{OTOC}(t)$ witnesses the noncommutativity of the observable $W(t) = U(t)^\dagger W U(t)$,
the unitary evolution {of $W$} in the Heisenberg picture by {the one-parameter group} $U(t)$, with respect to $V$.
It thus signals the delocalization of quantum information.
The OTOC can be written as \cite{yunger_Halpern2018quasiprobability}
\begin{equation}\label{eq: OTOC as invariants}
    \sum_{v_1,w_2,v_2,w_3} v_1w_2v_2^*w_3^*\text{Tr}(\Pi^{W(t)}_{w_3} \Pi^{V}_{v_2} \Pi^{W(t)}_{w_2} \Pi^{V}_{v_1} \rho ),
\end{equation}
where $v$ and $w$ range over the eigenvalues of $V = \!\!\sum_{v \in { \text{Spec}}(V)} v \Pi_v^V$ and $W = \!\!\sum_{w \in { \text{Spec}}(W)} w \Pi_w^{W}$ associated to eigenprojectors $\Pi_v^V$ and $\Pi_w^{W}$, respectively.
We also have that $\Pi_w^{W(t)} = U(t)^\dagger \Pi_w^W U(t)$.
The choice of labels made in Eq.~\eqref{eq: OTOC as invariants} comes from the interpretation of Ref.~\cite{yunger_Halpern2018quasiprobability}, where the OTOC is described by a quasiprobability naturally arising from comparing {the} forward and backward evolution of the system under a unitary action.

The coarse-graining hides the possible degeneracy in the spectral decomposition of the observables. In fact, if we take such degeneracy into account, we have that each observable is described by
\begin{equation}
    \Pi_v^V = \sum_{\lambda_v}\vert v, \lambda_v \rangle \langle v, \lambda_v \vert,
\end{equation}
where the eigenspace associated with eigenvalue $v$ is described by the complete set of vectors $\vert v, \lambda_v \rangle$ with $\lambda_v$ ranging over the degeneracy parameters. We find that the OTOC can be described in a more fine-grained way by
\begin{equation}
    \text{OTOC}(t) = \sum_{\substack{(v_1,\lambda_{v_1}),(v_2,\lambda_{v_2}), \\ (w_2,\lambda_{w_2}),(w_3,\lambda_{w_3})}} v_1 w_2 v_2^*w_3^*\tilde{A}_\rho,
\end{equation}
where $\tilde{A}_\rho = \tilde{A}_\rho(v_1,\lambda_{v_1};w_2,\lambda_{w_2};v_2,\lambda_{v_2};w_3,\lambda_{w_3})$ is the quasiprobability behind the OTOC~\cite{yunger_Halpern2018quasiprobability}. This is an extended KD distribution given at each phase-space point by
\begin{equation}
    \begin{aligned}
        \tilde{A}_\rho \,=\, &\langle w_3, \lambda_{w_3} \vert U \vert v_2,\lambda_{v_2}\rangle \langle v_2, \lambda_{v_2}\vert U^\dagger \vert w_2, \lambda_{w_2} \rangle \\
        &\langle w_2, \lambda_{w_2} \vert U \vert v_1, \lambda_{v_1}\rangle \langle v_1, \lambda_{v_1} \vert \rho U^\dagger \vert w_3, \lambda_{w_3}\rangle,
    \end{aligned}
    \label{eq: OTOC fine-grained}
\end{equation}
where each $\vert v_l, \lambda_{v_l}\rangle \langle v_l, \lambda_{v_l}\vert $ {or $\vert w_l, \lambda_{w_l}\rangle \langle w_l, \lambda_{w_l}\vert $}  corresponds to the state of the system after a measurement (related to observable $V$ {or $W$}) has been performed on it and an outcome $(v_l, \lambda_{v_l})$ {or $(w_l, \lambda_{w_l})$ has been} obtained.

As we will see shortly, the various fundamental concepts discussed so far---namely{,} the KD distribution in Eq.~\eqref{eq: KD-definition}, its extended version in Eq.~\eqref{eq: KD extended}, the quasiprobability behind the OTOC in Eq.~\eqref{eq: OTOC fine-grained}, {the post-selected quantum Fisher information in Eq.~\eqref{eq: ps QFI},} and the expression for weak values in Eq.~\eqref{eq:wv-expanded}---are all written in terms of Bargmann invariants, which we review next.\\

{\subsection{Bargmann invariants and how to measure them}\label{sec: Bargmann invariants}

\begin{figure*}
    \centering
    \includegraphics[width=0.9\textwidth]{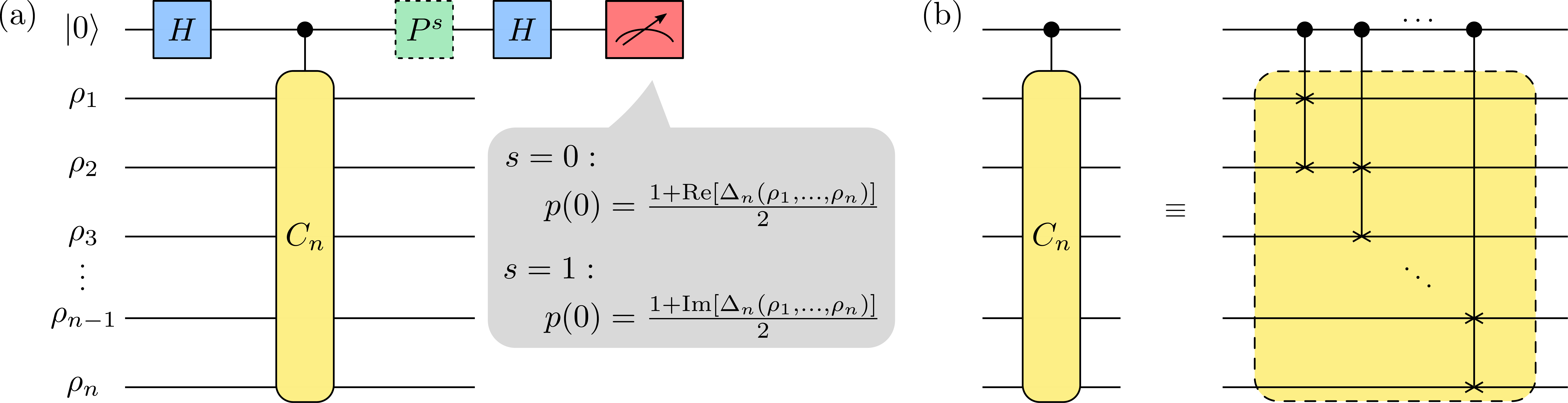}
    \caption{{\textbf{Hadamard test circuit to measure Bargmann invariants.} (a) This Hadamard test circuit can be used to estimate the real and imaginary part of the $n$th order Bargmann invariant $\Delta_n(\rho_1,\dots,\rho_n) = \text{Tr}(\rho_1 \dots \rho_n)$. The controlled-unitary (in yellow) is a controlled cycle permutation, and the $P^s$ gate equals identity for $s=0$, and $\text{diag}(1,i)$ for $s=1$. (b) Explicit decomposition of the controlled-cycle gate in terms of 3-system {controlled-swap} gates, also known  {in the qubit case} as Fredkin gates.}}
    \label{fig:cycle-test preliminaries}
\end{figure*}

General multivariate traces of quantum states can be directly measured in an efficient manner using the cycle test scheme~\cite{oszmaniec2021measuring}  or constant-depth circuit variations thereof~\cite{quek2022multivariate}. They are also known as $n$-th order Bargmann invariants~\cite{bargmann1964note}:
{given an $n$-tuple of states $(\rho_1, \ldots, \rho_n)$, its Bargmann invariant is}
\begin{equation}\label{eq: bargmann invariant definition}
  \Delta_n(\rho_1,\dots,\rho_n) = \text{Tr}(\rho_1 \dots \rho_n).
\end{equation}
These quantities are invariant under the simultaneous conjugation $\rho_i \mapsto Z\rho_iZ^{-1}$ by any invertible matrix $Z$.\footnote{{Throughout the text, when dealing with pure states $\psi$, we adopt the terser notation $\Delta_n(\ldots, \psi, \ldots)$ for $\Delta_n(\ldots, \ket{\psi}\!\bra{\psi}, \ldots)$.}}

{In general, we are interested in studying a set of states $\{\rho_i\}_{i}$ through all the Bargmann invariants among them:
the quantities $\Delta_n(\rho_{i_1}, \ldots, \rho_{i_n})$ for tuples of (not necessarily distinct) labels $(i_1,\dots,i_n)$.}

Bargmann invariants play an important role in linear optics~\cite{jones2022distinguishability}, characterizing multiphoton interference. In the context of invariant theory, multivariate traces of the form $\text{Tr}(X_1 \dots X_n)$ have been studied in depth since they completely generate the ring of invariants over matrix tuples~\cite{wigderson2019mathematics}.
In Ref.~\cite{oszmaniec2021measuring}, it was pointed out that this feature can be used to completely characterize the equivalence class{es} of  unitarily equivalent tuples of states, i.e., tuples $(\rho_i)_{i=1}^m$ that 
are mapped to one another by { the action of} a unitary map { $U$ simultaneously on every state in the tuple:} $\rho_i \mapsto U \rho_i U^\dagger$.
We refer to properties of sets of states that are invariant under {such} global unitary symmetry as \emph{relational properties}.
This is the most fundamental property of Bargmann invariants: every relational property of a set (or even a tuple) of states must be expressible as a function of the Bargmann invariants among those states.%
\footnote{\label{footnotecomment}{To avert potential misunderstanding, we reserve the letter $n$ to indicate the order of Bargmann invariants. When studying the relational invariant properties of a tuple (or set) of states, we use the letter $m$ to denote its size if its finiteness is relevant, as in $(\rho_i)_{i=1}^m$. Such properties are characterized by Bargmann invariants $\Delta_n(\rho_{i_1}, \ldots, \rho_{i_n})$ where $i_1, \ldots, i_n \in \{1, \ldots, m\}$ are $n$ labels.}}
As we will see later, in many non-trivial situations, one requires non-linear functions of invariants, as is the case for weak values.
Moreover, recognizing that a given function on a set of states can be written solely in terms of Bargmann invariants reveals that such a function captures relational information {(see Def.~\ref{def: relational nonclassicality} below)}.
{ We refer to Ref.~\cite[Sec. 13.9.2]{wigderson2019mathematics}, and references therein, for more information on polynomials over sets of matrices that are unitarily invariant.}

If we let $C_n$ be the unitary representation of the cyclic permutation of $n$ elements
\begin{equation}
    \label{eq:Cn}
    (\, 1 \; 2 \; 3 \; \cdots \; n \,) \; = \;  (\, 1 \; 2 \,) (\, 2 \; 3 \,) \cdots (\, n\!-\!1 \;\, n\,)
\end{equation}
that acts as
\begin{equation}
    (a_1, a_2, \ldots, a_{n-1}, a_n) \stackrel{C_n}{\longmapsto} (a_n, a_1, a_2, \ldots, a_{n-1}),
\end{equation}
it is possible to write any Bargmann invariant as 
\begin{equation}
    \Delta_n(\rho_1,\dots,\rho_n) = \text{Tr}(C_n \; \rho_1 \otimes \dots \otimes \rho_n). \label{eq:bitrc}
\end{equation} 

Because $C_n$ is unitary, we can use a quantum circuit known as the Hadamard test, shown in Fig.~\ref{fig:cycle-test preliminaries}(a),  to estimate the Bargmann invariant from Eq.~\eqref{eq:bitrc}. This procedure was detailed in Ref.~\cite{oszmaniec2021measuring}.
To estimate a generic $n$-order Bargmann invariant $\Delta_n(\rho_1,\dots,\rho_n)$ we input, in parallel, an auxiliary qubit in the state $\vert 0\rangle $ and the quantum states $\rho_1,\dots,\rho_n$.
We apply a Hadamard gate to put the auxiliary qubit in superposition, followed by a controlled unitary operation of the unitary $C_n$. Depending on whether we want to estimate the real ($s=0$) or imaginary ($s=1$) part of the invariant, we apply the gate $P^s = \text{diag}(1,i^s)$ followed by another Hadamard gate. We conclude the protocol by measuring the auxiliary qubit in the $Z$ basis. 

There are different circuit constructions of $C_n$~\cite{oszmaniec2021measuring,quek2022multivariate}, the simplest of which consists of the applying successive SWAP operators, thus translating into circuit form the decomposition of the cyclic permutation into transpositions from Eq.~\eqref{eq:Cn}:
\begin{equation}
    C_n \;=\; \text{SWAP}_{1,2}\circ \text{SWAP}_{2,3} \circ \dots \circ \text{SWAP}_{n-1,n},
\end{equation}
as depicted in Fig~\ref{fig:cycle-test preliminaries}(b).
The circuit must be run multiple times to estimate the invariant from the relative frequency estimation of a final computational basis measurement of the auxiliary qubit. For the specifics of implementing a qudit SWAP operation, see Ref.~\cite{garciaescartin2013aSWAPgate}.

In Sec.~\ref{sec: nonclassicality}, we consider how this Hadamard test circuit can be used to estimate weak values, KD distributions, and state spectra.

Among the most commonly considered higher-order Bargmann invariants in the literature are the \emph{univariate traces} of a single quantum state, i.e., quantities of the form $\Delta_n(\rho{, \ldots, \rho}) = \text{Tr}(\rho^n)$.
It is known that learning sufficiently many of these can be used to estimate the spectrum $\text{Spec}(\rho)$ of a given mixed state $\rho$~\cite{horodecki2003fromlimits}, an experimental task relevant for, e.g., quantifying  coherence~\cite{streltsov2017colloquium} or entanglement~{\cite{vedral1997quantifying, bovino2005direct, horodecki2009quantum, daley2012measuring, subacsi2019entanglement, quek2022multivariate}} using resource theory monotones.

{Procedures for measuring the quantities $\{\text{Tr}(\rho^n)\}_{n=1}^d$}}
have been proposed using visibility-based quantum algorithms~\cite{ekert2002direct,alves2003direct,horodecki2003fromlimits,horodecki2002method}, via similar implementations of the cycle operator~\cite{tanaka2014determining, quek2022multivariate}, and using random measurements~\cite{van2012measuring}.

{More generally, one may be interested in characterizing all relational (i.e., unitary-invariant) properties of a set of $m$ states. For pure states, there are known constructions of complete sets of invariants using $O(m^2)$ invariants of at most order $m$~\cite{oszmaniec2021measuring, chien2016characterization}.
Less is known for mixed states, but the minimum number of invariants in a complete set is known to depend also on the system dimension~\cite{oszmaniec2021measuring}.}

\subsection{Nonclassicality of Bargmann invariants}\label{sec: nonclassicality of Bargmann invariants preliminaries}

A particularly important relational property of a set of states, which can be completely characterized by Bargmann invariants, is pairwise commutativity. Given any set of states {$\{\rho_i\}_{i}$}, this is equivalent to the existence of some unitary $U$ such that $U\rho_i U^\dagger$ is diagonal for all $i$.
In Ref.~\cite{designolle2021set} the term \emph{set coherence} was used to describe the property of a set of states for which no such unitary exists. { For future reference, we write this notion as a definition.

\begin{definition}[Set coherence~\cite{designolle2021set}]\label{def: set-coherence}
    Let $\{\rho_i\}_{i}$ be a set of quantum states in a finite-dimensional Hilbert space $\mathcal{H}$.
    This set is said to be \emph{set incoherent} if all the states in it pairwise commute, or equivalently, if all these states are simultaneously diagonalizable, i.e.~are represented by diagonal density matrices with respect to the same reference basis.
    Otherwise, we say that the set $\{\rho_i\}_i$ is \emph{set coherent}.
\end{definition}
}
Refs.~\cite{galvao2020quantum,wagner2022inequalities,wagner2022coherence} studied this property by analyzing the simplest of Bargmann invariants, two-state overlaps $\Delta_2(\rho_i,\rho_j) = \text{Tr}(\rho_i\rho_j)$.\footnote{A note on terminology: in the literature, the term overlap is sometimes used to denote the inner product between two pure states, rather than its modulus squared, as we do here. Moreover, for pure states, the overlap becomes equivalent to the fidelity between two states.}

It is sensible to discuss the different \emph{realizations} of an invariant, or of a tuple of invariants.
We say that a value for $\Delta_n$ is realisable if there exists a tuple of quantum states $(\rho_1, \ldots, \rho_n)${,for} some Hilbert space $\mathcal{H}${,} such that $\Delta_n(\rho_1, \ldots, \rho_n)$ has that value.
Typically, there are infinitely many different realizations of a given value; for example, every tuple 
of the form $(\psi,\psi)$ for a pure state $\ket{\psi}$ realises $\Delta_2(\psi,\psi) = 1$. 

More generally, we consider tuples of Bargmann invariants, for example, the three overlaps among three states 
\begin{equation}\label{eq:triple_overlap_C2}
(\Delta_2(\rho_1,\rho_2), \Delta_2(\rho_1,\rho_3), \Delta_2(\rho_2,\rho_3)) \in [0,1]^3 .
\end{equation}
If only two-state overlaps appear, we also refer to the tuple as an overlap tuple.
It is possible to obtain linear constraints bounding the tuples of Bargmann invariants {realizable with sets} of incoherent states, as shown in  Refs. \cite{galvao2020quantum, wagner2022inequalities}.
We briefly sketch the main ideas of this approach. First, it is shown that some $0/1$-valued assignments to the overlaps in a tuple are not realizable by set-incoherent states;
in our running example from Eq.~\eqref{eq:triple_overlap_C2}, {this is the case, for instance, for} the triple of values $(1,0,1)$.
We then characterize the set of realizable, coherence-free tuples as the convex hull of all $0/1$-valued assignments for the overlap tuple that are realizable by set-incoherent states~\cite{wagner2022inequalities}.

{As the convex hull of a finite number of points, this forms a polytope, which can also be characterized via a finite system of facet-defining linear inequalities.}
For bounding triples of overlaps among three set-incoherent states $\rho_1,\rho_2$, and $\rho_3$, the simplest such inequality is
\begin{equation}\label{eq: preliminaries 3-cycle inequalities}
    \Delta_2(\rho_1,\rho_2) + \Delta_2(\rho_1,\rho_3)-\Delta_2(\rho_2,\rho_3) \leq 1 .
\end{equation}
This discussion will be relevant in Sec.~\ref{sec: nonclassicality} when we relate nonclassicality of quasiprobability distributions at some phase-space points to violations of these inequalities.

Ref.~\cite{galvao2020quantum} studied the polytope {$\mCt$} of overlap triples realizable by three set-incoherent states. This polytope is fully characterized  by the inequality of Eq.~\eqref{eq: preliminaries 3-cycle inequalities}, its index relabellings, plus the trivial inequalities $0 \leq \Delta_2 \leq 1$.
We may {also consider} general quantum realizations, yielding the set
\begin{equation}\label{eq:Q3}
    {\mQt} = \left\{\left(\begin{matrix}
        \Delta_2(\rho_1,\rho_2)\\\Delta_2(\rho_1,\rho_3)\\ \Delta_2(\rho_2,\rho_3) 
    \end{matrix}\right) \in [0,1]^3 \colon (\rho_i)_{i=1}^3 \in \mathcal{D}(\mathcal{H})^3\right\}.
\end{equation}
This set was shown numerically to be convex in Ref.~\cite{galvao2020quantum}. It is sometimes referred to as the body of quantum correlations, and parts of it correspond to the elliptope of Ref.~\cite{le2023quantum}.
In Refs.~\cite{galvao2020quantum,le2023quantum} it was shown that points in {$\mQt$} satisfy the inequality
\begin{equation}
    \begin{aligned}
        &\Delta_2(\rho_1,\rho_2)+\Delta_2(\rho_1,\rho_3)+\Delta_2(\rho_2,\rho_3)\\
        &\hspace{1.0cm}-2\sqrt{\Delta_2(\rho_1,\rho_2)\Delta_2(\rho_1,\rho_3)\Delta_2(\rho_2,\rho_3)} \leq 1.
    \end{aligned}
    \label{eq: convex body inequality}
\end{equation}
It holds that {$\mCt \subset \mQt$}, and any point in $\mQt \setminus \mCt$ serves as a witness of set coherence for the triplets realizing the tuple of invariants.

For higher-order invariants, there are also constraints on the values realizable using only set-incoherent states. It was argued in Ref.~\cite{oszmaniec2021measuring} that negativity or imaginarity of Bargmann invariants witness this form of basis-independent coherence for a set of states, namely set coherence.
In Sec.~\ref{sec: nonclassicality} we revisit this idea. 

Set coherence in particular (see Def.~\ref{def: set-coherence}), and noncommutativity of observables in general, is a unitary-invariant property, i.e., it is independent of a choice of reference basis.
As such, deciding whether a set of observables pairwise commute can be framed solely in terms of Bargmann invariants, as discussed in Ref.~\cite{oszmaniec2021measuring}.
This, in turn, demands the characterization of sufficiently many Bargmann invariants, forming a so-called \emph{complete set of invariants}, which solves the problem entirely: knowing the values of all Bargmann invariants in that set suffices for deciding the property of interest.
For the case of noncommutativity, the complete sets of Bargmann invariants and their associated bounds are not known in general.

A significantly simpler task is to select a single invariant, or an incomplete family of invariants, or functions thereof, and use their values to merely \textit{witness} nonclassicality in the form of set coherence.
Each of these witnesses signals a different manifestation of noncommutativity through the possible values attained by them. Because different witnesses usually capture significantly different phenomena, they each deserve individual analysis.
Still, as we discuss in this work, all such nonclassicality witnesses can be framed within the following common terminology.

\begin{definition}[Relational nonclassicality]\label{def: relational nonclassicality}
Consider sets $\{\rho_i\}_{i}$ of states in a finite-dimensional Hilbert space $\mathcal{H}$.
A property of such sets is said to be \emph{relational} when it is captured by a function $f(\{\rho_i\}_i)$ which is \emph{unitarily invariant} in that
$f(\{\rho_i\}_i) = f(\{U\rho_i U^\dagger\}_i)$ for all unitaries $U$.
We term \emph{relational nonclassicality} the realization of values of such a relational property $f$ that are not realizable with set-incoherent states.
\end{definition}
Clearly, Bargmann invariants provide a specific case of function $f$ in the above definition. Since this is a particularly relevant case for the analysis in this work, we may refer to the appearance of nonclassical values as \emph{Bargmann nonclassicality}.

Next, we start presenting our results, building on the following three aspects of Bargmann invariants.
First, we recognize functions of Bargmann invariants that correspond to different concepts we have reviewed, e.g., OTOCs, weak values, and KD distributions. Following that, we point out that cycle test circuits can be used to estimate these quantities without the need for post-selection. Then, we show that any nonclassical property of the functions of invariants we consider can be linked to the nonclassicality of the underlying invariants themselves.

\section{Bargmann invariants as a unifying concept}\label{sec: unifying}

In this section, we argue that Bargmann invariants serve as a unifying concept for
the discussion of nonclassicality in KD distributions and weak values.
We show how both are functions of third-order invariants.
The same holds for extended KD distributions. A practical by-product of this unified description is the realization that cycle test circuits can be used to measure weak values, KD distributions, and the spectrum of quantum states, as we discuss in Sec.~\ref{sec: circuits}.

Starting with the definition of weak value in Eq.~\eqref{eq:wv-expanded}, we recognize that
\begin{equation}\label{eq: wk-value as bi}
    \langle A \rangle_w = \sum_{a \in { \text{Spec}}(A)} a \frac{\Delta_3(\phi,a,\psi)}{\Delta_2(\phi,\psi)}. 
\end{equation}
Thus, we regard weak values as functions of second- and third-order invariants, encoding the relational information between the basis associated with observable $A$, and the two reference states $\ket{\phi}$ and $\ket{\psi}$.
While weak values can also be viewed as averages of a conditional KD distribution~\cite{lostaglio2022kirkwood}, this requires the introduction of a second complete basis of reference into the weak value setup. Therefore, Eq.~(\ref{eq: wk-value as bi}) expresses the weak value in terms of an economical function of only the strictly necessary unitary invariants, dispensing with the irrelevant, arbitrary choice of a second reference basis.

As weak values are functions of Bargmann invariants, they can be estimated with no need for post-selection, up to any desired accuracy, using the cycle test circuits of \cite{oszmaniec2021measuring}. This shows that weak values need not be defined with respect to weak measurement schemes, as they were historically introduced. Eq. (\ref{eq: wk-value as bi}) provides a natural description of the weak value as an average of the eigenvalues of observable $A$, weighted by a quasiprobability distribution $\Delta_3(\phi,a,\psi)/\Delta_2(\phi,\psi)$ in terms of states $\ket{\phi},\ket{\psi}$. Eq. (\ref{eq: wk-value as bi}) also shows it is not strictly necessary to view the weak value as merely the first term in an infinite series describing the shift of a classical pointer in a weak measurement scheme \cite{dressel2014colloquium}. Together with our analysis presented in Sec.~\ref{sec: circuits}, this strengthens the case for viewing weak values as interesting constructions of their own, as argued, e.g., in Ref. \cite{cohen2018determination}.

Similarly, the KD quasiprobability distribution can be {also} shown to be a relational property of the states involved, and thus a function of Bargmann invariants. The definition of the KD distribution at a single phase-space point, given by Eq.~\eqref{eq: KD-definition}, shows explicitly that $\xi(\rho \vert i,f) = \Delta_3(i,\rho,f)$, i.e., that its value is a third-order Bargmann invariant. The  KD distribution encodes the relational information about the quantum state $\rho$ we want to describe, together with two complete Hilbert space bases $\{\ket{i}\}_{i\in I}$ and $\{\ket{f}\}_{f\in F}$.

Consequently, the question of characterizing the nonclassicality of KD distributions (see Def.~\ref{def: KD nonclassicality}) boils down to characterizing the nonclassicality of individual third-order invariants, which we address in Sec.~\ref{sec: nonclassicality}.
As pointed out there, this novel perspective may prove useful in the future for connecting KD-nonclassicality with contextuality, an important open problem in quantum foundations.
From a practical point of view, this simple understanding allows us to propose in Sec.~\ref{sec: circuits} variations of cycle test circuits to measure KD distributions with no need for post-selection, in contrast to previous proposals based on weak measurements.

The observation above carries over to extended KD distributions relative to $n$ bases, which are expressible as $(n+1)$th order Bargmann invariants.
More generally still, one can express an extended KD distribution of the form given by Eq.~\eqref{eq: KD extended}, which is relative to arbitrary PVMs (not necessarily composed only of rank-one projectors, i.e., bases).
That additionally involves a factor that accounts for multiplicity.
Setting $\rho^i_k \coloneqq \frac{\Pi^i_k}{\text{Tr}(\Pi^i_k)}$, one has\footnote{This expression was suggested to us by one of the anonymous referees.}
\begin{equation}\label{eq: KD extended as Bargmann Referee}
    \xi(\rho \vert k_1,\dots,k_n)= \Delta_{n+1}(\rho,\rho^1_{k_1},\dots,\rho^n_{k_n})\prod_{i=1}^n \text{Tr}(\Pi^i_{k_i}).
\end{equation}
This expression shows that at each phase-space point $(k_1, \ldots, k_n)$ the extended KD distribution describes relational properties between the probed state $\rho$ together with all the states $\{\rho^i_{k_i}\}_i$. Any nonclassical property of the quasiprobability distribution must be due to the nonclassical properties of such higher-order Bargmann invariants.

A number of constructions that are derived from KD distributions are, as a result, also functions of Bargmann invariants. The extended KD distribution used in Eq.~\eqref{eq: ps QFI} is simply the collection of all invariants of the form 
\begin{equation}
    \xi(\rho_\theta \vert i,i',f) = \Delta_4(i,\rho_\theta,i',f).  
\end{equation}
This leads to the representation of the post-selected quantum Fisher information $\mathcal{I}^{\text{ps}}$ purely in terms of Bargmann invariants:
\begin{equation}\label{eq: ps QFI as Bargmann}
 \mathcal{I}^{\text{ps}} = 4\sum_{i,i',f}\lambda_i\lambda_{i'}\frac{\Delta_4(i,\rho_\theta,i',f)}{\sum_f \Delta_2(\rho_\theta,f)}-4\left \vert \sum_{i,i',f}\lambda_i \frac{\Delta_4(i,\rho_\theta,i',f)}{\sum_f \Delta_2(\rho_\theta,f)}\right \vert^2   .
\end{equation}
Similarly, the quasiprobability distribution behind the OTOC is an extended KD distribution:
\begin{equation*}
    \tilde{A}_\rho = \Delta_5(\lambda_{w_3}, \lambda_{v_2}^{(t)}, \lambda_{w_2},\lambda_{v_1}^{(t)},\rho(t)),
\end{equation*}
where we use the simplified notation $\ket{\lambda_{w_l}} \defeq \ket{w_l,\lambda_{w_l}}$,
$\ket{\lambda_{v_l}^{(t)}} \defeq U(t) \ket{v_l,\lambda_{v_l}}$, and $\rho(t) \defeq U(t)\rho U(t)^\dagger$.

Finally, it is important to remember the reasoning behind invariant theory: every relational property between states can be expressed in terms of a set {of} Bargmann invariants. This holds true for the spectrum of a quantum state, with the peculiarity that here we must consider invariants describing a (multi)set with $n$ copies of a single quantum state.
Estimating quantities of the form 
\begin{equation}
    \Delta_n(\rho,\dots,\rho) = \text{Tr}(\rho^n)
\end{equation}
for sufficiently large $n$, one can learn several interesting properties of a single state, in particular its spectrum, as we recall later in Sec.~\ref{sec: spectrum}.

There are two upshots from the unified perspective in terms of Bargmann invariants established in the present section:
\begin{itemize}[nosep,leftmargin=*]
\item First, the fact that KD distributions, OTOCs, and weak values are functions of Bargmann invariants allows for a unified discussion of nonclassicality in terms of the values taken by them, in the sense of Def.~\ref{def: relational nonclassicality}. We use this insight in Sec.~\ref{sec: nonclassicality}, where we establish connections between the resource theory of set coherence (discussed above, see Def.~\ref{def: set-coherence}) and known results on the nonclassicality of quasiprobability distributions.
\item The fact that all the constructions presented can be written as functions of Bargmann invariants means that cycle test circuits can be employed to estimate all of them, as we discuss in the next section. 
\end{itemize}

\section{Quantum circuits for measuring weak values and Kirkwood--Dirac quasiprobability distributions}\label{sec: circuits}

\begin{figure*}
    \includegraphics[width=0.9\textwidth]{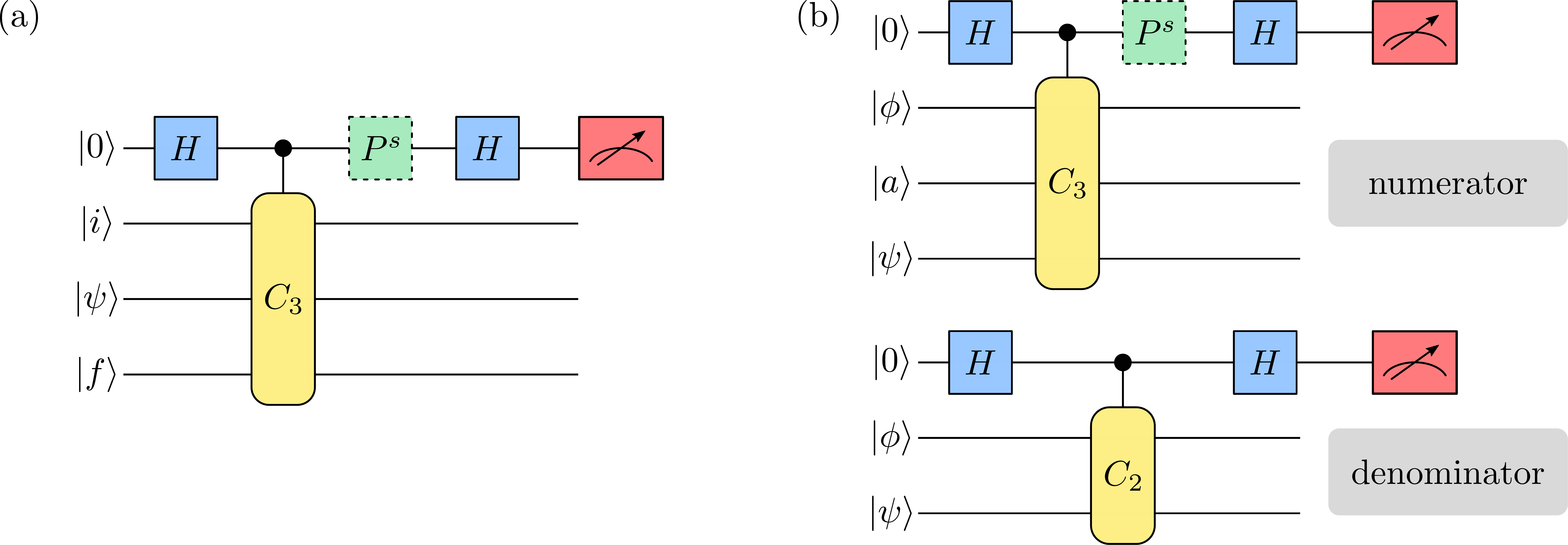}
    \caption{\textbf{Quantum circuits for measuring the third-order Bargmann invariants characterizing (a) Kirkwood--Dirac quasiprobability distribution and (b) weak values.} $P^s$ corresponds either to applying a  phase-gate $P = \text{diag}(1,i)$ in case we want to estimate the imaginary part of the invariant ($s=1$), or an identity $P^0 = \mathbb{1} =\text{diag}(1,1)$ in case we want to estimate the real part ($s=0$). Circuit (a) assumes that one can prepare basis states from the two bases $\{\vert i \rangle \}_i, \{\vert f \rangle \}_f$ defining the KD distribution, as well as the state $\vert \psi \rangle$ to be characterized. Two variations of the circuit estimate the real and imaginary parts of the value of the KD distribution at a chosen phase-space point. For the circuit in (b), we require input states that are eigenvectors $\{\vert a \rangle \}_a$ of an observable $A$, as well as the pre- and post-selection states used in the standard weak measurement protocol. On top of that, anomalous weak values result from precisely selecting the overlaps $\vert \langle \phi \vert \psi \rangle \vert^2$, which can be directly measured using the two-state SWAP test.}
    \label{fig:third_order_invariant_circuit}
\end{figure*}

\subsection{Kirkwood--Dirac quasiprobability}\label{subsec: KD distribution}

Fig.~\ref{fig:third_order_invariant_circuit}(a) presents a circuit for measuring the value of the KD quasiprobability associated with pure state $\ket{\psi}$ at a given phase-space point $(i,f)$: $\xi( \psi \vert i,f) = \langle f \vert i \rangle \langle i \vert \psi \rangle \langle \psi \vert f \rangle = \Delta_3(i,\psi,f)$.
The circuit, based on a cycle test \cite{oszmaniec2021measuring}, can be straightforwardly generalized to measure any value $\text{Tr}(\Pi_{k_1}^1\dots\Pi_{k_n}^n\rho)$ of extended KD distributions; cf.~Secs.~\ref{subsec: post-selected QFI} and~\ref{subsec: OTOC}.
While we here focus on the case of pure states, the same architecture is capable of characterizing $d$-dimensional mixed input states.

The circuit shown in Fig.~\ref{fig:third_order_invariant_circuit}(a) has four wires or systems: the first is an auxiliary qubit, while the remaining are $d$-dimensional systems described by the Hilbert space $\mathcal{H}$ under study.
Given a fixed choice of $i$ and $f$, the circuit initializes by preparing the product state $\vert 0 \rangle \otimes \vert i \rangle \otimes \vert \psi \rangle \otimes \vert f \rangle$.
The specific order is relevant as the value of the Bargmann invariant changes for different choices of orders, due to noncommutativity of the rank-one projectors involved.
The circuit then proceeds to put the auxiliary qubit in a superposition and perform a cascade of controlled-SWAP operations between the remaining states, as mentioned in Fig.~\ref{fig:cycle-test preliminaries}.

Assuming that one can prepare the states $\ket{\psi}$, $\ket{i}$, and $\ket{f}$, the circuit in Fig.~\ref{fig:third_order_invariant_circuit}(a)
measures the (real and imaginary part{s} of) the third-order Bargmann invariant defining $\xi( \psi \vert i,f)$, to precision $\varepsilon$, with high probability, using $O(1/\varepsilon^2)$ samples of the triplet of states~\cite{quek2022multivariate},
and therefore constant order of samples of the state $\ket{\psi}$.
Consequently, the protocol can estimate the entire KD distribution of a state in any finite dimension $d$ to precision $\varepsilon$ with high probability using a total of $\tilde{O}(d^2/\varepsilon^2)$ samples, where the tilde hides $\log(d)$ terms.

Assuming that $\langle i \vert f \rangle \neq 0$,  learning the full KD distribution of a state can be used to perform full tomography~\cite{johansen2007quantum}.
Tomography via KD distribution is neither better nor worse in terms of sample and measurement complexity than textbook quantum state tomography using Pauli measurements. The latter requires $O(d^4/\varepsilon^2)$ samples~\cite{nielsen2002quantum,flammia2012quantum} to achieve a precision $\varepsilon$ relative to the trace distance with high probability.
Now, if we learn the value of the KD distribution of a state $\rho$ at every phase-space point, we can completely reconstruct the state  by
\begin{equation}\label{eq:rho from its KD}
    \rho = \sum_{(i,f)\in I \times F} \frac{\vert i \rangle \langle f \vert}{\langle i \vert f \rangle } \xi(\rho \vert i,f)
\end{equation}
(in fact, this equation holds for any { bounded operator, not just density matrices}~\cite{yunger_Halpern2018quasiprobability}).
We count the number of samples and measurements required to perform tomography in this fashion, where a measurement consists of one run of the cycle test circuit to estimate some $\xi(\rho \vert i,f)$.
In order to estimate the value at each phase-space point up to precision $\varepsilon$ with probability $1-\delta$,\footnote{Whenever we say that our learning task is successful with high probability, we mean that we are considering $1-\delta$ with $\delta>0$ a fixed small number.} we need $O(\ln(2/\delta)/\varepsilon^2)$ samples, a bound provided by the Hoeffding concentration inequality~\cite{hoeffding1994probability,shalev2014understanding}.
Assuming that we want to estimate at every phase-space point with the same precision $\varepsilon$ and probability $1-\delta$, the number of samples is of the order of $\tilde{O}(d^2/\varepsilon^2)$, where we hide the dependence on $\delta$,
as it is commonly assumed to be fixed.

To compare this scaling in samples with standard tomography, we must quantify the success of the procedure in terms of the distance induced by the $1$-norm.
Write $\hat{\xi}$ for the estimated KD distribution from the cycle test and $\hat{\rho}$ for the corresponding estimate of the state $\rho$, obtained via Eq.~\eqref{eq:rho from its KD}.
Assuming that the values of KD are estimated with precision $\varepsilon_1$ at every phase-space point, i.e.,
that $\left\vert \xi(\rho \vert i,f) - \hat{\xi}(\rho \vert i,f)\right\vert\leq \varepsilon_1$
for all $(i,f)\in I \times F$, we then have
\begin{equation}
    \begin{aligned}
        \Vert \rho - \hat{\rho}\Vert_1 &= \left\Vert \sum_{i, f}\frac{\vert i \rangle \langle f \vert}{\langle i \vert f \rangle }\left(\xi(\rho\vert i,f) - \hat{\xi}(\rho\vert i,f)\right)\right\Vert_1\\
        &\leq \varepsilon_1 \left\Vert \sum_{i, f}{ 
        \frac{\vert i \rangle \langle f \vert}{\langle i \vert f \rangle }
        }\right\Vert_1 = \varepsilon_1 \left\Vert \mathbb{1}_{d \times d}
        \right\Vert_1 = \varepsilon_1 d,
    \end{aligned}
\end{equation}
using that the rank-one operators $\vert i \rangle \langle  f\vert / \langle i \vert f \rangle$ form an  orthonormal basis of $\mathcal{B}(\mathcal{H})$.
Therefore, in total, one needs $\tilde{O}(d^4/\varepsilon^2)$ samples to perform full KD state tomography $\varepsilon$ close in the trace distance.

However, from the above calculation, it is clear that performing full tomography is \textit{not necessary} for achieving arbitrarily good precision in learning the KD distribution of a given state $\rho$. Therefore, in scenarios where one is interested in learning solely the KD distribution $\{\xi(\rho \vert i,f)\}_{(i,f) \in I \times F}$, 
the required number of $\Tilde{O}(d^2/\varepsilon^2)$ samples constitutes a polynomial advantage in sample complexity with respect to performing Pauli-based full tomography, which requires $O(d^4/\varepsilon^2)$ samples, or even efficient incoherent full tomography~\cite{chen2022tight}, which requires $O(d^3/\varepsilon^2)$ samples.

\subsection{Weak values}\label{subsec: weak values}

Fig.~\ref{fig:third_order_invariant_circuit}(b) presents a circuit for estimating the weak value $A_w$ of an observable $A$ with respect to pre- and post-selection states $\psi$ and $\phi$.
It works similarly to the cycle test described in Sec.~\ref{subsec: KD distribution}.
The main difference in the estimation of the weak value stems from the fact that it requires running two different circuits.
The first, used to estimate the numerator of the expression for the weak value in Eq.~\eqref{eq: wk-value as bi}, works just like the one shown in Fig.~\ref{fig:third_order_invariant_circuit}(a), while the second corresponds to a simple SWAP-test for estimating the denominator of Eq.~\eqref{eq: wk-value as bi}.
Such a process estimates a single summand $\Delta_3(\phi,a,\psi)/\Delta_2(\phi,\psi)$.
To obtain the entire weak value, the procedure must be repeated for each eigenvector, and the estimated quantities then summed after being weighted by their corresponding eigenvalue.

The whole procedure assumes that we are able to not only prepare the selected states $\ket{\phi}$ and $\ket{\psi}$, but also each eigenvector $\ket{a}$ of $A$.
This could be feasible given the mathematical description of the spectral decomposition of the operator $A$. But even in a more operational setup, it could be achieved, for example, if assuming the ability to perform a (strong) measurement of $A$ on the totally mixed state.

The nonclassicality of weak values, witnessed by anomalous negative or nonreal values, can also be directly estimated using the above circuit and analyzed through the lens of Bargmann invariant nonclassicality.
This result dissociates once more weak values from weak measurements, similarly to the discussions in Refs.~\cite{cohen2018determination,johansen2007reconstructing,vallone2016strong}.

The lack of post-selection might suggest that the sample complexity for measuring weak values using the cycle test improves over the standard weak measurement protocol. The following theorem shows that this is not the case.
In fact, if the overlap between the pre- and post-selected states is small, our bounds show that the cycle test performs worse.

\begin{theorem}\label{theorem: sample complexity}
    For estimating the weak value $A_w$ with precision $\varepsilon$ and high probability, one needs $N^{(weak)} = O(|A_w|^2/\varepsilon^2 \Delta_2(\phi,\psi))$ samples when using the standard weak measurement scheme, with $\Delta_2(\phi,\psi) = \vert \langle \phi \vert \psi \rangle \vert^2$ being the probability of successful post-selection.
    Moreover, one needs $N^{(cycle-test)} = O(|A_w|^2/\varepsilon^2\Delta_2(\phi,\psi)^2)$ samples when using {the procedure based on} cycle test{s}.
\end{theorem}

\begin{figure*}
    \includegraphics[width=1\textwidth]{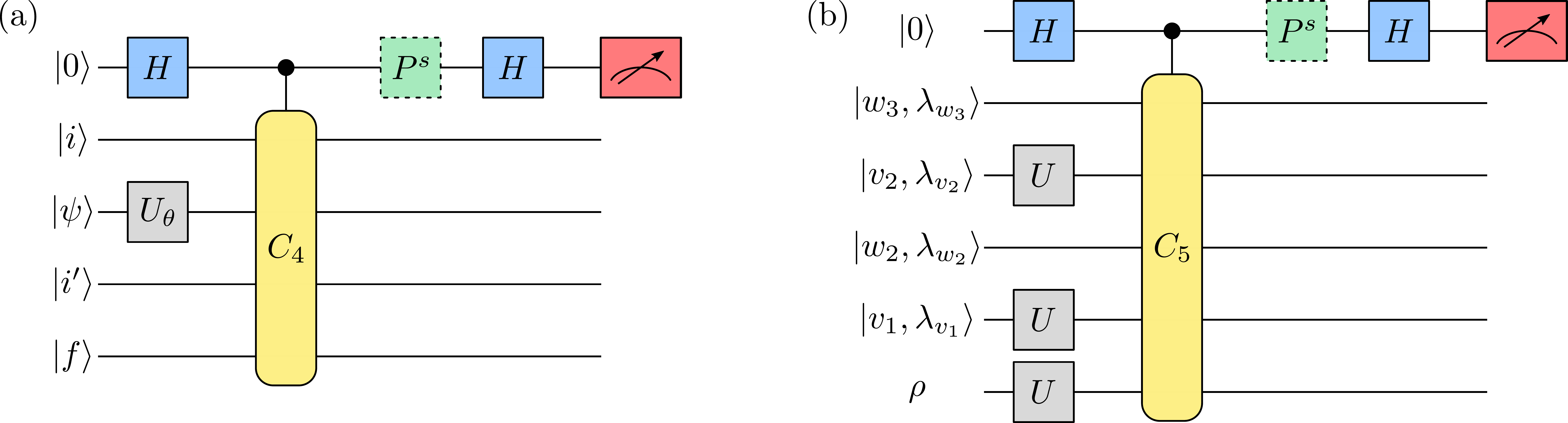}
    \caption{\textbf{Bargmann invariant circuits for measuring the extended Kirkwood--Dirac quasiprobability distributions behind (a) the post-selected quantum Fisher information and (b) the out-of-time-ordered correlator.} $P^s$ is as described in Fig.~\ref{fig:third_order_invariant_circuit}. $C_4$ and $C_5$ correspond to the circuit implementation of the unitary discussed in Sec.~\ref{sec: Bargmann invariants}. Part (a) shows how one can measure the post-selected quantum Fisher information using the cycle test and applying a unitary operator $U_\theta = e^{-i\theta I}$ over the state $\vert\psi\rangle$ considered in the protocol.  It is assumed to be possible to prepare, in parallel, states $\vert f \rangle$ associated with the post-selection and the eigenstates related to the operator $I$. Part (b) shows the protocol for estimating the out-of-time-ordered correlator from the related Bargmann invariants. A possibly scrambling unitary $U$ is applied to $\rho$ and to the projectors related to the eigenspaces of $v$ and $w$ defined by $\Pi_v^V$ and $\Pi_w^W$, respectively. It is possible to avoid {an effective} backward{s}-in-time evolution to estimate the OTOC, for any pair of observables $V$ and $W$.\label{fig: PPA and OTOC}}
\end{figure*}

We formally prove this result in Appendix \ref{appendix: cycle vs weak}.
Intuitively, the difference in the number of required samples comes from the fact that the estimate using cycle tests is constructed from the estimation of two quantities---a second- and a third-order invariants---while the standard weak measurement scheme directly estimates the weak value from the pointer's position.
Since to observe anomalous weak values in most relevant situations---in particular in weak-value amplification schemes---the overlap $\Delta_2(\phi,\psi)$ needs to be close to zero,\footnote{However, note that weak values with a small imaginary part do not require small overlaps between the pre- and post-selection states, and are also anomalous weak values. A simple example is the weak value of $\sigma_x$ for a system pre-selected in $|z-\rangle$ and post-selected in $|y+\rangle$, which is $(\sigma_x)_w = -i$. A similar argument can be made when the real part of a weak value lies slightly outside the interval determined by the spectrum of the operator.} this difference in sample complexity makes the protocol suboptimal.
Therefore, one needs to look for specific cases of interest for which the cycle test may be more relevant.
A clear example is provided by experimental situations in which one utilizes purely \textit{imaginary} weak values~\cite{brunner2010measuring, li2011ultrasensitive, dressel2012significance}.
In these cases, simply witnessing imaginarity boils down to observing imaginary third-order invariants related to the numerator of the weak value $A_w$, a task that (with high probability and to precision $\varepsilon$) requires only $O(d/\varepsilon^2)$ samples.

\subsection{Post-selected quantum Fisher information}\label{subsec: post-selected QFI}

We have seen that KD distributions can be extended { to involve a larger number of bases (or PVMs)} as described by Eq.~\eqref{eq: KD extended}, and that such constructions have appeared in the literature expressing important quantities behind OTOCs or quantum advantage in post-selected metrology{, for example}.
Their description effectively requires Bargmann invariants of order higher than $3$; see Eq.~\eqref{eq: KD extended as Bargmann Referee}.
In Fig.~\ref{fig: PPA and OTOC}, we describe circuits specifically targeted to measuring such quantities.

The circuit in Fig.~\ref{fig: PPA and OTOC}(a) measures the fourth-order Bargmann invariant in the expression for the post-selected quantum Fisher information; see Eq.~\eqref{eq: ps QFI as Bargmann}.
This quantity is of interest, for example, in proof-of-principle experiments such as that performed in Ref.~\cite{lupu_Gladstein2022negative}.
We discuss in detail how to make such a measurement using this simple quantum circuit.

Recall the metrological scenario for parameter estimation considered in Sec.~\ref{subsec: prelminaries post-selected QFI}, and the notation therein.
The circuit in Fig.~\ref{fig: PPA and OTOC}(a) consists of one (auxiliary) qubit and four systems described by the $d$-dimensional Hilbert space of interest.
It is initialized by preparing the product state
\begin{equation}
    \label{eq: prep post QFI circ}
    \ket{0} \otimes \ket{i} \otimes \ket{\psi} \otimes \ket{i'} \otimes \ket{f},
\end{equation}
where $\ket{i}$, $\ket{i'}$ are elements of an eigenbasis of (the Hamiltonian) $I$ and $\ket{f}$ an element of a basis for the range of the (post-selecting) projector $F$.
The order in which those states are considered matters, being related to the order in the KD distribution presented in Eq.~\eqref{eq: KD-PPA}.
The circuit then proceeds by applying a local unitary $U_\theta = e^{-i\theta I}$ on the state $\ket{\psi}$, which encodes information about the value of $\theta$.

A cycle test is then performed over the four states, similarly to the procedures discussed in Sec.~\ref{subsec: KD distribution} or Sec.~\ref{subsec: weak values}.
The auxiliary qubit is put in a superposition and three controlled SWAP operations are applied: first between $\ket{i}$ and $\ket{\psi_\theta}$, then between $\ket{\psi_\theta}$ and $\ket{i'}$, and finally between $\ket{i'}$ and $\ket{f}$.
As before, to estimate the real part of the invariant from Eq.~\eqref{eq: KD-PPA} one simply measures the auxiliary qubit in the $X$ basis, while to estimate the imaginary part of the invariant one first applies a phase gate $P = \text{diag}(1,i)$ and then measures the auxiliary qubit in the $X$ basis.

The description above corresponds to a single run of the protocol, which estimates a specific value of the extended KD distribution of $\ket{\psi_\theta}$ for a given phase-space point.
To estimate the entire function $\mathcal{I}^{\text{ps}}$, one must generate statistics for each preparation of the form in Eq.~\eqref{eq: prep post QFI circ} for all the labels $i$, $i'$, and $f$.
The basis $\{\vert i\rangle\}_i$ has $d$ elements, corresponding to the Hilbert space dimension, while $\{\vert f \rangle\}_f$ has fewer labels than $d$.

A note on the measurement and sample complexity related to estimating the post-selected quantum Fisher information in particular, and the extended KD distribution more generally, is necessary.
Estimating $\mathcal{I}^{\text{ps}}$ by this method requires $\tilde{O}(d^3/\varepsilon^2)$ samples and measurements, as all the information from the entire extended KD distribution is relevant.
This does not represent an information complexity gain to experimentally probe this quantity as there are other approaches that require fewer samples.
For instance, as we will recall in Sec.~\ref{sec: spectrum}, some techniques to estimate the spectrum (which  are also {used} to perform full tomography) solve the problem using $O(d^2/\varepsilon^2)$ samples.
Yet, due to the intricate aspects of performing ideal protocols such as these, the circuit from Fig.~\ref{fig: PPA and OTOC}(a) has the potential advantage of providing a simple and fixed structure for the measurement to be performed.

It is worth mentioning that, as in the experiment performed in Ref.~\cite{lupu_Gladstein2022negative}, one may be interested in estimating a single phase-space point of the extended KD distribution, in order to monitor the presence of negative values.
In such a case, provided we have access to some triplet of states $\{\vert i\rangle, \vert i'\rangle, \vert f \rangle\}$ with $i\neq i'$, estimating a single phase-space point has a simple scaling in terms of complexity relative to the one described in Sec.~\ref{subsec: KD distribution}.
{To estimate } any extended KD phase-space point value $\xi(\rho \vert \Pi_{k_1},\dots, \Pi_{k_n})$. one needs $\tilde{O}(1/\varepsilon^2)$ samples, while {a number of samples of} order $\tilde{O}((n+1)d^2/\varepsilon^2)$ {is required} to perform full tomography of all the states and projectors involved in the calculation.

As discussed in Ref.~\cite[Supp. Note 2]{arvidsson_Shukur2020quantum}, imaginary values of the post-selected quantum Fisher information do not contribute to metrological advantages since they cannot increase the second term of $\mathcal{I}^{\text{ps}}$ as described by Eq.~\eqref{eq: ps QFI}. Hence, if one is only interested in witnessing the presence of the resource, namely negativity, one can use the circuit version of Fig.~\ref{fig: PPA and OTOC}(a) with $s=0$ (i.e., $P^s=\mathbb{1}$) to estimate the real part of the invariant. 

{To conclude this section, we remark that our circuits can be applied, after straightforward adaptations, to estimate more general definitions of the post-selected quantum Fisher information \textit{matrix}, relevant for multi-parameter estimation, e.g., from Ref.~\cite[Eq.~(37)]{jenne2022unbounded}.}

\subsection{Out-of-time-ordered correlators}\label{subsec: OTOC}

The circuit in Fig.~\ref{fig: PPA and OTOC}(b) applies the cycle test to measure OTOCs. From Eq.~\eqref{eq: OTOC as invariants}, we may use the fact that $U$ is unitary and that the equivalence between the Heisenberg and Schr\"odinger pictures to rewrite the Bargmann invariant defining the OTOC as
\begin{equation}
    \begin{aligned}
        \text{Tr}(&\Pi_{w_3}^{W(t)}\Pi_{v_2}^V\Pi_{w_2}^{W(t)}\Pi_{v_1}^V\rho)\\
        &=\text{Tr}(U^\dagger\Pi_{w_3}^{W}U\Pi_{v_2}^VU^\dagger\Pi_{w_2}^{W}U\Pi_{v_1}^V\rho)\\
        &= \text{Tr}(\Pi_{w_3}^{W}U\Pi_{v_2}^VU^\dagger\Pi_{w_2}^{W}U\Pi_{v_1}^VU^\dagger U\rho U^\dagger)\\
        &= \text{Tr}(\Pi_{w_3}^{W}\mathcal{U}(\Pi_{v_2}^V)\Pi_{w_2}^{W}\mathcal{U}(\Pi_{v_1}^V)\mathcal{U}(\rho)),
    \end{aligned}
\end{equation}
where $\mathcal{U}(\cdot) \defeq U (\cdot)U^\dagger$ is the dynamical evolution in the Schr\"odinger picture. This allows us to have the cycle test that evaluates the OTOC without the need for applying { the reversed} evolution $\mathcal{U}^\dagger$. As compared with the protocols for estimating OTOCs reported in Ref.~\cite{yunger_Halpern2018quasiprobability}, not only does the cycle test avoid {an effective} backward{s}-in-time evolution, { it} also {does away with the need for} the auxiliary qubit to remain coherent during the scrambling dynamics.
The protocol estimates the extended KD distribution behind the OTOC in either its coarse-grained version or the fine-grained description of Eq.~\eqref{eq: OTOC fine-grained}, which carries nonclassical information that brings further relevant tools to the analysis of the scrambling dynamics~\cite{gonzalez2019out,alonso2022diagnosing}.

These observations place the cycle test as an interesting new paradigm for estimating OTOCs, in comparison with all protocols reported in Ref.~\cite[Table I, pg. 12]{yunger_Halpern2018quasiprobability}.
Two important drawbacks of this protocol are:
the fact that it requires quantum information processing of possibly many degrees of freedom in parallel,
and the assumption that one is capable of preparing the states associated with the eigenprojectors {$\Pi^W_w$ and $\Pi^V_v$}.
{Note that the latter assumption can be reduced to being able to measure the observables $W$ and $V$ (on any state)}

A similar research direction was pursued in Ref.~\cite{mitarai2019methodology}, where the authors analyzed when indirect metrological inference can be replaced by direct measurement processes, focusing on the Hadamard test. Our approach allows us to make claims that are similar in nature. For instance, Bargmann invariants can \textit{also} be used to directly measure OTOCs.

\subsection{Comparison with other methods for estimating the Kirkwood--Dirac quasiprobability distribution} \label{sec: comparison}

We have seen how to use cycle test circuits to estimate the KD distribution in Sec.~\ref{subsec: KD distribution}, weak values in Sec.~\ref{subsec: weak values}, and also two specific examples of interest in Secs.~\ref{subsec: post-selected QFI} and~\ref{subsec: OTOC}.
We proceed to situate our measurement scheme among those reviewed in Refs.~\cite{lostaglio2022kirkwood,yunger_Halpern2018quasiprobability}.
{
Table~\ref{tab: comparison} compares our proposal with several protocols found in the literature.}
The first column presents the references that introduced the measurement scheme reviewed, which in some cases were not directly linked to estimating the KD distribution but {were} later shown to serve for this purpose in Ref.~\cite{lostaglio2022kirkwood}.
The second column shows the output of the protocol after gathering statistics, which in most cases is later used to infer the KD distribution.
The third and fourth columns show the number of measurements and samples of the state $\rho \in \mathcal{D}(\mathcal{H})$ \textit{per trial} of the experiment, where we consider weakly coupling to a pointer as another measurement.
The fifth column indicates whether weak measurements are involved in the scheme, meaning that weak couplings are necessary. Finally, the last column indicates the {required} number of auxiliary quantum systems.

\begin{table*}[t]
    \centering
    \renewcommand{\arraystretch}{1.5}
    \begin{tabular}{c|c|c|c|c|c}
        \hline\hline
        Protocol & Output & \# Meas. & \# Samples of $\rho \in \mathcal{D}(\mathcal{H})$ & Weak Meas. & \# Aux. Systems \\
         \hline
         1) Fig.~\eqref{fig:third_order_invariant_circuit}\,\, & $\xi(\rho \vert i,f)$ & 1 & 1 & No &1 qubit + 2 systems $\mathcal{H}$\\
         \hline
         2) Ref.~\cite{johansen2007quantum}  & $p_{if}^{\text{TPM}},p_{f}^{\text{END}},p_{if}^{\text{WTPM}}$   & 5 & 3 & No & 0\\
         \hline
          3) Ref.~\cite{mazzola2013measuring}  & $\chi(u,v)$   & 1 & 1 & No & $1$ environment  \\
          \hline
          4) Ref.~\cite{buscemi2013direct}  & $p_{\pm,if}$   & 3 & 1 & No  & 1 system $\mathcal{H}$ \\
          \hline
          5) Ref.~\cite{resch2004extracting}  & Pointer shift  & 3 & 1 & Yes & 0 \\
          \hline
          6) Ref.~\cite{mitchison2007sequential}  & Pointer shift  & 3 & 1 & Yes & 0 \\
          \hline
          7) Ref.~\cite{yunger_Halpern2017jarzynski}  & $\langle i|\rho | f \rangle, \langle f|i\rangle $ & 6 & 1 & No & 2 qubits \\
          \hline
          8) Ref.~\cite{yunger_Halpern2017jarzynski}  & $\mathcal{P}_{\text{weak}}$ & 3 & 1 & Yes & 0 \\
          \hline
          9) Ref.~\cite{rall2020quantum}  & $\xi(\rho \vert i,f)$ & 1 & 1 & No & 3 qubits\\
          \hline\hline
    \end{tabular}
    \caption{{\textbf{Comparison between our approach and other methods for estimating the Kirkwood--Dirac quasiprobability distribution.} We compare previously proposed protocols (rows 2--9) with the one studied here (row 1). The first column shows the general output of the protocol; in the main text, we explain how one can recover the KD distribution from such outputs. {Other columns indicate} the number of measurements, samples, and auxiliary systems necessary for each single trial of the protocol{, and whether it requires weakly coupling to a pointer, i.e., weak measurements}.
    These single trials are solely for the estimation of the real part of the KD-value $\xi(\rho \vert i,f)$ for fixed $\rho$, $i$, and $f$.
    All protocols are {also} capable of estimating  the imaginary part with similar schemes.
    The protocol in row 6 estimates inner products, hence both real and imaginary parts, which makes a fair comparison more subtle.   {\textit{ Abbreviations:}} Meas.:  measurement, Aux.: auxiliary.} }
    \label{tab: comparison}
\end{table*}

As these protocols are quite different {in nature}, it is challenging to make fair comparisons. In what follows, we detail the specifics of each protocol and explain the reasoning for the information in Table~\ref{tab: comparison}.
We choose to present the number of measurements and samples per trial as most proposals lack a formal complexity analysis (the exception is the protocol presented in  Ref.~\cite{rall2020quantum}).
Making a complexity analysis for each protocol would be {beyond} the scope of this work, and could potentially hide relevant constant factors in the number of measurements or samples needed.

Table~\ref{tab: comparison} {analyzes the resources} for estimating just {the real part} $\text{Re}[\xi(\rho \vert i,f)]$ with $i,f$ fixed. Estimating the imaginary part constitutes a different experiment in all cases except for the interferometric protocol from Ref.~\cite{yunger_Halpern2017jarzynski}, which estimates inner products;
in this case, {our analysis considers the estimation of } the value $\xi(\rho \vert i,f)$.{ In the other cases, to also take into account} the estimation of $\text{Im}[\xi(\rho \vert i,f)]$ one simply has to multiply all measurements and samples by 2.
Whenever there is more than one output (e.g., second row), parallel estimations need to take place to reconstruct the KD distribution phase-space point $\xi(\rho \vert i,f)$, increasing the number of measurements, samples, and auxiliary systems per trial.
We now proceed to discuss each protocol in more detail.

The protocol first introduced in Ref.~\cite{johansen2007quantum} was recently investigated in Ref.~\cite{lostaglio2022kirkwood} and implemented in Ref.~\cite{hernandezgomez2023projective}. It demands the parallel implementation of { three} different schemes, each of which uses one sample of $\rho$ per trial{:} a two-point measurement (TPM)~\cite{esposito2009nonequilibrium} to learn $p_{if}^{\text{TPM}}$, another similar scheme { with a} non-selective {measurement}\footnote{This procedure was termed weak TPM {in} Ref.~\cite{lostaglio2022kirkwood}. While no weak measurement is needed, the fact that a non-selective (unsharp) measurement is used resembles the statistics that can be obtained from a weak measurement.} to estimate $p_{if}^{\text{WTPM}}${,} and a final experiment with a single measurement to estimate $p_{f}^{\text{END}}$.
The {real part of the} KD distribution is then recovered by {computing} $\text{Re}[\xi(\rho \vert i,f)] = p_{ij}^{\text{TPM}} + \frac{1}{2}(p_f^{\text{END}} - p_{if}^{\text{WTPM}})${, while} the imaginary part {can also be estimated} with a similar setup.
It is unclear how this protocol scales in terms of sample and measurement complexity. Our protocol demands {fewer} measurements and samples at the cost of quantum memory for the states $\vert i \rangle$ and $\vert f \rangle$.

The protocol from Ref.~\cite{mazzola2013measuring} estimates the characteristic function $\chi(u,v) = \sum_{if}\xi(\rho \vert i,f) e^{i \lambda_f v + i \lambda_i u}$ of the KD distribution with an interferometric scheme.
It needs to couple the system of interest to an auxiliary system (called an environment) for the estimation of the real and imaginary parts of $\chi(u,v)$. Importantly, to recover the KD distribution, one needs to post-process the final result {by} applying an inverse Fourier transform over $\chi$.
This might impact on the scalability, { in case} the classical processing  demand{s} more precision in the estimation, {as does} the classical post-processing for estimating the spectrum, which we will see in Sec.~\ref{sec: spectrum}.

While the protocols discussed above were studied with a focus on quantum foundations and quantum thermodynamics, the one from Ref.~\cite{buscemi2013direct} {provides} a general way of estimating $\text{Tr}(\rho AB)$ for two observables $A$ and $B$. It demands an auxiliary system {described by the $d$-dimensional Hilbert space} $\mathcal{H}$ prepared in a maximally mixed state.
It then has one joint projective measurement over the target and auxiliary systems $M_{\pm} \coloneqq \{\Pi^+, \Pi^-\}$ with $\Pi^{\pm} \coloneqq (\mathbb{1} {\pm} \text{SWAP})/2$, together with local projective measurements in each system, $\{\vert i\rangle\langle i \vert \}_i$ and $\{\vert f \rangle \langle f \vert\}_f $. The probability of obtaining $\pm$ from the first measurement $M_{\pm}$ followed by $i$ and $f$ from the local measurements is denoted by $p_{\pm,if}$. One then recovers the real part of $\xi(\rho \vert i,f)$ by
\begin{equation}
    \text{Re}[\xi(\rho \vert i,f)] = \frac{d+1}{2}p_{+,if}-\frac{d-1}{2}p_{-,if}.
\end{equation}
Note that this needs more measurements per trial than our protocol, but, due to the probabilistic nature of quantum measurements, one cannot target specific phase-space points unless outcomes from the local measurements are post-selected.

The protocol from Ref.~\cite{resch2004extracting} is based on joint weak measurements together with the estimation of higher-order weak values, which implies that the effective output to be measured is the statistics of an average of a pair of pointer shifts. It is important to point out that Ref.~\cite{resch2004extracting} does not directly link their findings to estimating the KD distribution, while we (and Ref.~\cite{lostaglio2022kirkwood}) argue that, {with} a direct re-reading of the protocol, one can connect the average shift of the pointer to KD values. The number of measurements in total is two weak measurements involving interactions with two different pointers in sequence, followed by final projective measurements. Since the protocol is based on standard weak measurement techniques, there is the need { for} post-selection.

While the protocol just mentioned is capable of estimating the KD distribution using joint weak measurements ({requiring} at least two systems) a similar approach using sequential weak measurements from Ref.~\cite{mitchison2007sequential}, and experimentally implemented in Refs.~\cite{thekkadath2016direct,piacentini2016measuring}, is also capable of estimating the KD distribution using a weak measurement scheme.

The interference scheme from Ref.~\cite[Appendix B]{yunger_Halpern2017jarzynski} consists of estimating each inner-product $\langle f | i \rangle$ and $\langle i|\rho |f \rangle$ separately. The general way proposed to estimate   $\langle i|\rho |f \rangle$ is using quantum state tomography and for estimating $\langle f | i \rangle$ the interferometer needs to estimate the real and imaginary parts {separately} by performing two parallel experiments: a SWAP test and another more involved one.
The overall protocol becomes fairly demanding in terms of measurements and samples, especially due {to} the need to reconstruct each inner product and the estimation of $\langle i|\rho |f \rangle$.
The same reference proposes a better scheme to estimate the extended KD distribution relevant for OTOC measurements that remains valid for general extended KD distributions (see Table~\ref{tab: comparison}).
In each trial of this weak-measurement-based scheme, for a general state $\rho$, one { performs} one weak measurement and two strong measurements in the pre/post-processing parts.\footnote{In fact, {this} proposal { requires} three weak measurements since they are essentially estimating a higher-order invariant appearing in the OTOC.
We use fewer weak measurements since we  consider only the usual KD distribution.
While this makes it more difficult to compare with our approach, this protocol is still an interesting measurement scheme to be considered for estimating the KD distribution.} For more details, see Ref.~\cite[Appendix A, Eq. (A5)]{yunger_Halpern2017jarzynski} and Ref.~\cite{yunger_Halpern2018quasiprobability}.

Finally, we comment on the more recent protocol { introduced in} Ref.~\cite{rall2020quantum}. The ideas present there greatly resemble our scheme{ since} something similar to a Hadamard test is performed (and a careful analysis of sample complexity is presented).
The main difference is that instead of { using} auxiliary systems {to prepare} states $\vert i \rangle$ and $\vert f \rangle$ entering a generic cycle test, the method uses  auxiliary qubits and performs block-encoding unitaries $U_i$ and $U_f$ that act {on the joint} system {composed of} $\vert \psi \rangle $ and the auxiliary qubits.
{As far as we know, t}his protocol is the only one that {requires a} similar number of samples and measurements to our approach without the need {for} major classical post-processing.
The main difference between the two is that we avoid the need {for} different block-encoding unitaries $U_i$ and $U_f$ for each phase-space instance, at the cost of having access to a quantum memory with prepared states $\vert i \rangle$ and $\vert f \rangle$.
{Sidestepping such encoding issues} clarifies the relational properties of the states involved, and may simplify the classical pre-processing for larger phase spaces, favoring our protocol{. Meanwhile,} for near-term devices, the lack of a large quantum memory may be more { crucial}, favoring the protocol { introduced in} Ref.~\cite{rall2020quantum}.

In contrast to the { other methods} listed in Table~\ref{tab: comparison}, our protocol works even when no classical description { of} the states being compared {is available}.
In other words, the cycle tests we use can receive as inputs unknown states provided by third parties, which may suit applications involving, for example, state comparisons in networks.
In that, our protocol is akin to the programmable quantum gate arrays, first proposed { in} Ref.~\cite{NielsenC97}, where a fixed circuit has different functionalities depending on the input program states.

{ We sum up with} a more detailed comparison of our proposal with those in Table~\ref{tab: comparison}.
Our protocol requires fewer measurements or samples than those from Refs.~\cite{buscemi2013direct,yunger_Halpern2017jarzynski,johansen2007quantum,resch2004extracting}, does not require post-selection or weak measurements as { that} from Ref.~\cite{resch2004extracting}, does not require classical post-processing of data such as { that} from Ref.~\cite{mazzola2013measuring} ({ which applies} the inverse Fourier transform), and it does not require a different Hadamard test, nor complete information on the states $\ket{i}$ and $\ket{f}$, { for each} KD phase-space point as that from Ref.~\cite{rall2020quantum}.
Similarly to the protocol from Ref.~\cite{rall2020quantum}, our protocol allows the estimation of the distribution at a single phase-space point.

\subsection{Estimating the spectrum of a quantum state} \label{sec: spectrum}

So far, all quantities considered have depicted some relational aspect among different states in a set.
Interestingly, cycle test circuits can also be used to estimate quantities associated with a single quantum state.
In this section, we start  by discussing quantum circuits that can estimate univariate trace polynomials $\text{Tr}(\rho^n)$, reminiscent of known approaches~\cite{alves2003direct}.
We carefully address the numerics of estimating the spectrum using these quantities due to the post-processing of classical data required by the Faddeev--LeVerrier algorithm.
Finally, we provide a careful analysis in terms of the sample complexity needed to perform the task, together with a comparison with the best protocol available{,} known as the empirical Young diagram algorithm. We revisit the task of learning the spectrum, known from the early 2000s~\cite{horodecki2003fromlimits}, in light of the results from Ref.~\cite{odonnel2015quantumspectrum} that provided optimal bounds for sample complexity of learning the spectrum, arguing that with our circuits it is possible to achieve the same efficiency with a simplified construction.

In Fig.~\ref{fig:learning spectrum}, we describe the cycle test circuit that estimates $\text{Tr}(\rho^n)$ for any {given} $n\in \mathbb{N}$.
It uses a single auxiliary qubit and $n$ controlled-SWAP gates over $n$ copies of the state $\rho$.
This proposal is not optimal in the number of gates, however. In Refs.~\cite{quek2022multivariate, oszmaniec2021measuring}, families of circuits were proposed with either logarithmic or constant depth.
Yet, we observe that such advantage in depth is gained at the cost of an increase in space, as the proposals require the preparation of auxiliary multi-party Greenberger--Horne--Zeilinger (GHZ) states~\cite{greenberger1989going}.

The Faddeev--LeVerrier algorithm, which uses the Newton identities~\cite[Chapter 3]{escofier2012galois}, allows us to compute the characteristic polynomial of the density matrix $\rho$ from the quantities $\{\text{Tr}(\rho^n)\}_{n=1}^d$.
The eigenvalues are then obtained by finding the roots of the characteristic equation.
We provide code for this computation, as well as supporting information, in the repository~\cite{molero2022github}, where we present these ideas with an experimentally-friendly approach. 

This algorithm has some limitations, associated with finding the roots of the characteristic polynomial. As the dimension of the system increases, the coefficients of lower-degree monomials decrease, making them prone to numerical inaccuracies. We find that, for dimensions greater than 9, the predicted roots have an imaginary part larger than $10^{-8}$.
This imaginary part appears because of numerical inaccuracies, and it is independent of statistical or experimental noise. Therefore, it does not affect the error in the real part of the predicted eigenvalues, and we may simply discard the imaginary part of the results. Appendix \ref{appendix: samples spectrum} presents a thorough study of this behavior.

\begin{figure}
    \centering
    \includegraphics[width=\columnwidth]{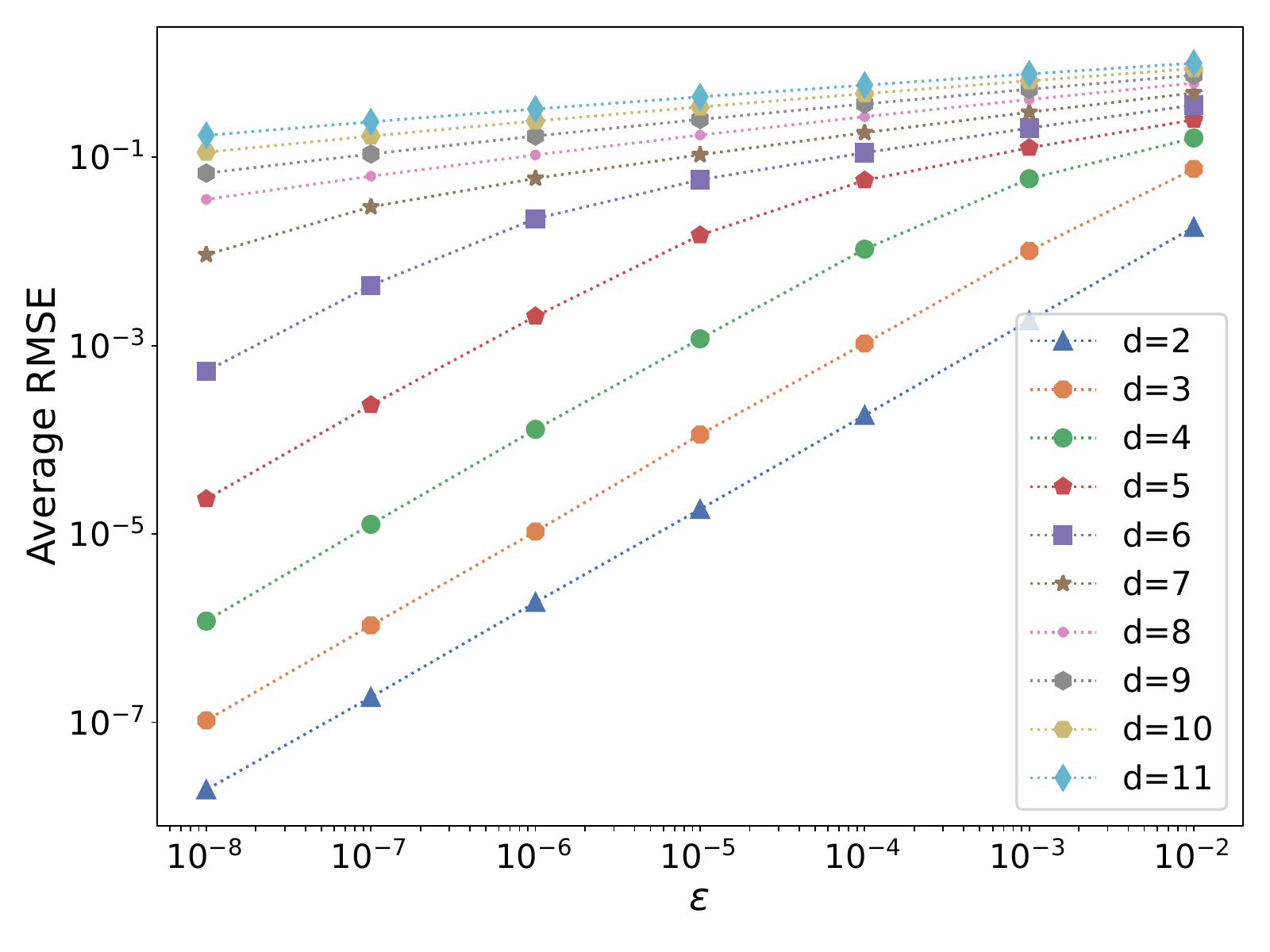}
    \caption{\textbf{Average root-mean-squared error (RMSE) of the estimate for the real part of the eigenvalues of random mixed states, under Gaussian noise $\epsilon$.} We start from a data set of exact values of $\Delta_n \defeq \text{Tr}(\rho^n)$,  $\{\widehat{\Delta_n}\}_{n=2}^d$, introduce Gaussian noise with standard deviation $\varepsilon$, and plot the average root-mean-squared error (RMSE) of the estimated eigenvalues under noise. The spectrum reconstruction uses  the Faddeev--LeVerrier algorithm based on Newton's identities, and the average RMSE is over 5000 samples, each used to generate 1000 noisy samples. We use the Ginibre random ensemble of mixed states of rank larger than one, employing the algorithm introduced in Ref.~\cite{bruzda2009random}.}
    \label{fig:spectrum MSE}
\end{figure}

In real experiments, noisy data will distort the predictions of the algorithm. In Fig.~\ref{fig:spectrum MSE}, we study the aggregated root-mean-squared error (RMSE) of the real part of the eigenvalues of random mixed density matrices considering different amounts of noise in the experimental data.
For a given set of traces computed from a random state, we can simulate the results from experiments by adding Gaussian noise with standard deviation $\varepsilon$.
This noisy set of traces represents the average results from experiments accurate up to an error $\varepsilon$.  The RMSE is computed {by} comparing the predictions of the algorithm using the noisy traces with the noiseless ones. We repeat the procedure and average with $1 \times 10^3$ noisy versions for each state, for $5 \times 10^4$ random states using the Ginibre ensemble~\cite{bruzda2009random}. These numbers are not related to the experimental samples required to estimate the traces, but to give enough statistics to see the response of the algorithm to noisy data.
Due to the estimation of the entire spectrum, it is expected that higher dimensions incur a high RMSE even for small $\varepsilon$.
Typical error ranges in experimental realizations, e.g.,  $\varepsilon \sim 10^{-3}$, would allow us to estimate the entire spectrum for dimensions 3 or 4 with good accuracy.
A lower error $\varepsilon$ in the experimental data could be attainable by increasing the number of samples. We remark that we only take statistical errors into account, as more details about the specific experimental implementation would be needed for further error analysis.

A simpler task involves learning the largest eigenvalue of $\rho$ with high precision.
This task should give better results as the higher-degree coefficients of the characteristic polynomial have a greater size than those of lower degrees, and are thus less affected by noise.
We study the scaling of the error in the estimation of the largest eigenvalue in Appendix~\ref{appendix: samples spectrum}.
As expected, higher dimensional systems can be studied. For instance, for an experimental error of $\varepsilon = 10^{-4}$, the average RMSE is around $5 \times 10^{-3}$ for dimension 6.
However, measuring the complete set of traces for estimating the largest eigenvalue is not an optimal procedure.
While the coefficients of higher-degree monomials of the polynomial require traces of lower orders in powers of the density matrix, coefficients of lower degrees require traces of higher orders. As these traces require a more complex circuit, they will also be noisier.
A sensible alternative for computing the largest eigenvalue is to truncate the polynomial to a certain lower degree $k$~\cite{subacsi2019entanglement},
that is, to compute the higher-degree coefficients of the polynomial using traces of lower powers of $\rho$, $\{ \text{Tr}(\rho^2), \hdots, \text{Tr}(\rho^{d-k})\}$.
In this approximation, the traces that are more difficult to measure are not taken into account, as their contribution is smaller.

Oddly enough, estimating the spectrum of $\rho$ from $\{\text{Tr}(\rho^n)\}_{n=1}^d$ for $d > 2$ has not, to the best of our knowledge, received any attention from experimental investigations yet, even though one can show that the required number of samples is optimal.

\begin{theorem}\label{theorem: sample complexity learning spectrum}
    The number of samples needed to estimate $\{\text{Tr}(\rho^n)\}_{n=1}^d$ up to precision $\varepsilon$ in all quantities with high probability is $N = \tilde{O}(d^2/\varepsilon^2)$ with the cycle test. The number of measurements needed over auxiliary qubits is $\tilde{O}(d/\varepsilon^2)$.
\end{theorem}

We prove this result in Appendix~\ref{appendix: samples spectrum}. The protocol of Ref.~\cite{quek2022multivariate} improves this sample complexity by logarithmic factors, obtaining that order $O(d^2/\varepsilon^2)$ samples are needed with high probability.
One can learn the spectrum of a given state $\rho$ using an optimal number of samples and measurements using the cycle test circuits described in this work or those described in Ref.~\cite{quek2022multivariate}.
The best bounds on the number of samples needed to estimate the spectrum ${ \text{Spec}}(\rho)$ for $d$-dimensional mixed states $\rho$ were first rigorously obtained in Refs.~\cite{keyl2001estimating,alicki1988symmetry} with the use of the empirical Young diagram (EYD) algorithm, later connected with complexity-theoretic arguments~\cite{wright2016learn, odonnel2015quantumspectrum}, and formally used to show that the same sample complexity is optimal for full state tomography~\cite{O_Donnell2016, Haah2016, chen2022tight, yuen2023improved}.

\begin{figure}
    \centering
    \includegraphics[width=0.9\columnwidth]{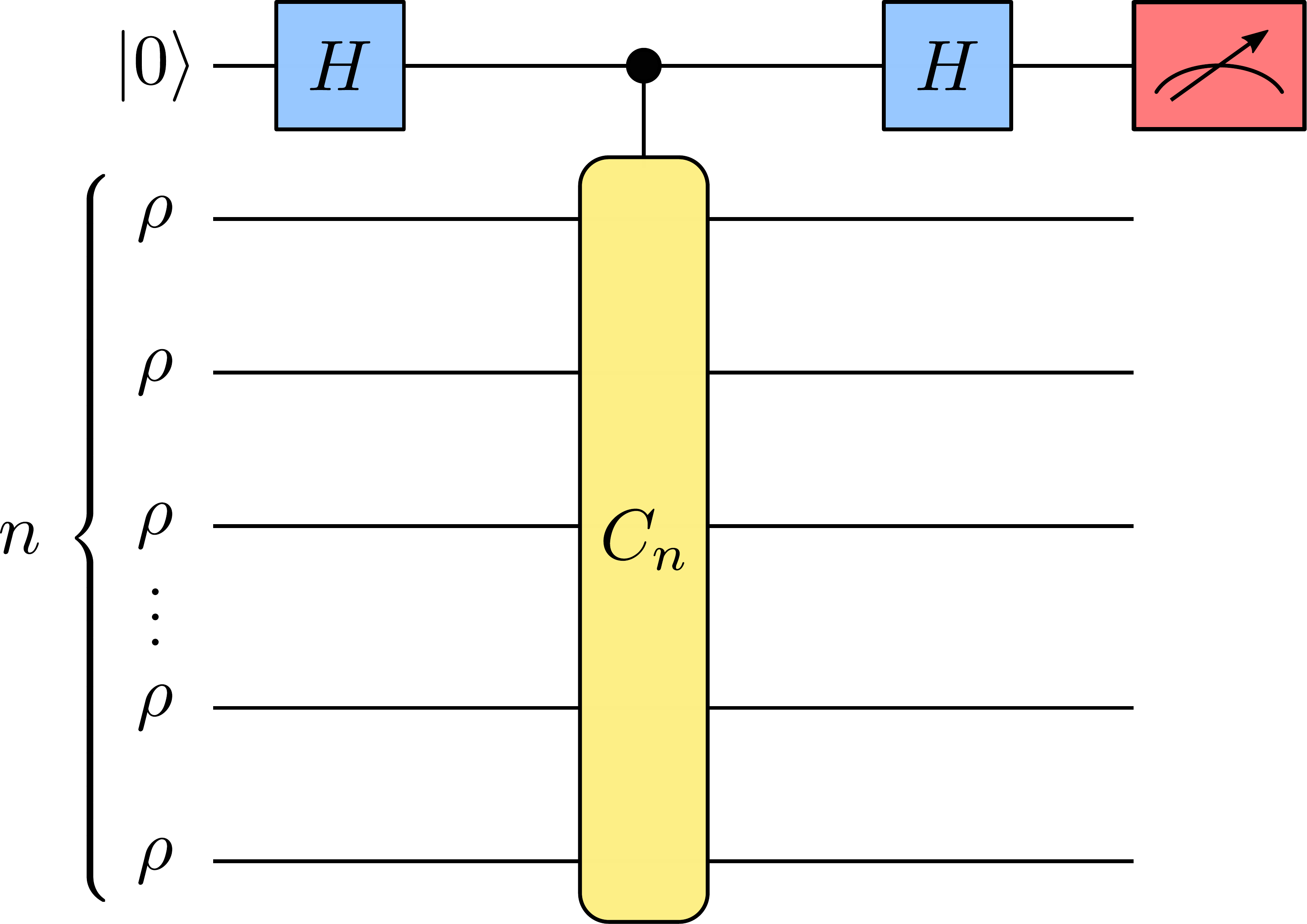}
    \caption{\textbf{Circuit for measuring} $\text{Tr}(\rho^n)$\textbf{.} The input consists of $n$ copies of a $d$-dimensional state $\rho$. Measurements of $\text{Tr}(\rho^n)$ for $n=1,2, \hdots, d$ gives us the spectrum of $\rho$, via the Newton identities.}
    \label{fig:learning spectrum}
\end{figure}

We provide a brief review of the EYD algorithm. Recall that by Schur--Weyl duality, the $n$-tensor product of $ \mathbb{C}^d $ decomposes as a direct sum of irreducible representations of $ \mathrm{U} \left( d \right) $:
\begin{equation}\label{Schur_Weyl}
    \left( \mathbb{C}^d \right)^{\otimes n} \cong \bigoplus_\lambda \left( V_\lambda \right)^{\otimes m_\lambda}
\end{equation}
where the direct sum goes over all partitions $\lambda$ of $n$ as a sum of no more than $d$ positive integers.
The direct sum contains $m_\lambda$ copies of the irreducible representation $V_\lambda$, where $m_\lambda$ is the number of standard Young tableaux corresponding to the partition $\lambda$.
The irreducible representations $V_\lambda$ (also known as \textit{Weyl modules} in this context) can be constructed from $(\mathbb{C}^d)^{\otimes n} $ as the images of certain projections.
Essentially, for each standard Young tableau, the corresponding projection acts on $(\mathbb{C}^d)^{\otimes n}$ by symmetrizing the tensor indices $1, \ldots, n$ that correspond to each row of the tableau, and anti-symmetrizing the indices corresponding to each column.
Now, the EYD algorithm is performed by fixing a large value of $n$, and then measuring $\rho^{\otimes n}$ with respect to the projections $\{ \Pi_\lambda\}$, where $\Pi_\lambda$ projects onto the subspace indexed by $\lambda$ in Eq.~\eqref{Schur_Weyl}.
The measurement outcome $\lambda$ is a set of positive integers that sum to $n$. After dividing each one by $n$, they comprise an estimate for the nonzero eigenvalues of $\rho$.

Some important differences between the two methods for learning the spectrum---the EYD algorithm and the protocol using the cycle test---deserve to be mentioned. To implement the EYD algorithm, one needs to measure $\rho^{\otimes n}$ in the highly entangled Schur basis (an orthonormal basis comprised of eigenvectors of the projections $\Pi_\lambda$), and to obtain the spectrum with precision $\varepsilon$, one needs $n_{\text{EYD}} = O \left( d^2 / \varepsilon^2 \right)$. In terms of sample complexity, this is optimal: one requires $N = n_{\text{EYD}}$ copies of $ \rho$, which is the optimal sample complexity---and the same as in our protocol, as stated before.

At first glance, the EYD algorithm may seem to be simpler, as it requires implementing only a single quantum circuit. Indeed, one should choose a large enough value of $n_{\text{EYD}}$, and then implement the Schur transform (mapping the Schur basis to the computational one) and measure the final state only for that fixed $n_{\text{EYD}}$. In contrast, our approach requires estimating $\Tr \left( \rho^n \right)$ for $n = 1, \ldots, d$, where every increment of $n$ requires a slightly larger circuit. However, this ``partition'' of the algorithm into $d$ quantum circuits holds a huge practical advantage. As is evident from the previous paragraph, the EYD algorithm requires a quantum computer that can reliably sustain a highly entangled state in a Hilbert space of dimension $ d^{n_{\text{EYD}}} $. This dimension grows larger for smaller desired precision $\varepsilon$.
In contrast, in our approach, the largest circuit acts on dimension $d^d$, which still grows large for higher values of $d$ but does not depend on $\varepsilon$.
Moreover, the gate complexity of the Schur transform is polynomial in both $d$ and $n_{\text{EYD}}$~\cite{bacon2006efficient}, while the complexity of the circuit for measuring $\Tr \left( \rho^n \right)$ only depends on $n$.\footnote{Or even better, achieving $\log(n)$, as shown in Ref.~\cite{oszmaniec2021measuring}, or constant depth, as shown in Ref.~\cite{quek2022multivariate}.}
Thus, the gate complexity of our approach does not depend on $\varepsilon$ and is generally much smaller when compared with EYD.

It should be noted that the authors of Ref.~\cite{Beverland2018} overcome the issue of gate complexity by considering a specific experimental setup where the symmetry of the interaction Hamiltonian is utilized to implement the Schur transform efficiently.\footnote{While acknowledging that some of the experimental aspects needed to implement the EYD algorithm ``seem like a daunting task in practice''~\cite[pg. 3]{Beverland2018}.}

Moreover, other methods have been proposed for estimating $\text{Tr}(\rho^n)$. One such test, known as Brun's multicopy method~\cite{brun2004measuring}, uses the fact that any polynomial  $f(\rho)$ on the matrix elements of $\rho$, with respect to some fixed basis, can be expressed as the expectation of some operator acting over $\rho^{\otimes m}$ for some natural number $m$.
Observables acting over such tensor products are sometimes called \textit{multicopy observables} and have been proposed to access nonclassicality of quantum states in optics~\cite{arnhem2022multicopy}.
For a formal comparison between our scheme and Brun's, see Ref.~\cite[Sec. V]{brun2004measuring}.
It is also worth mentioning that the cycle test was used in earlier works~\cite{ekert2002direct}, focusing on the estimation of $\text{Tr}(\rho^n)$ only, without considering more general Bargmann invariants.
Another recently introduced methodology for general estimation of Bargmann invariants is presented in Ref.~\cite{quek2022multivariate}.

\section{Minimal experimental conditions for witnessing nonclassicality}\label{sec: nonclassicality}

Most of the constructions discussed in the last section capture the community's interest because they present features that are puzzling to be described classically---{notably, values that are negative or have nonzero imaginary part}.
Given that the { protocols} discussed in this work are capable of estimating these quantities,
we are now interested in studying minimal experimental setups {for which it is possible} to witness nonclassicality {in various guises}.

{We} first discuss how the negativity or imaginarity of {(third- and higher-order)} Bargmann invariants{, and in particular that of weak values or KD distributions,} can be understood as {witnessing} quantum {(set)} coherence.

{However, following Ref.~\cite{wagner2022inequalities} (see Sec.~\ref{sec: nonclassicality of Bargmann invariants preliminaries}),
set coherence can also be witnessed from two-state overlaps.
So, afterwards, we consider to what extent these two ways of witnessing coherence are interrelated.
In particular, we consider the question of whether anything can be inferred about the negativity or imaginarity of third-order invariants from knowledge of pairwise overlaps alone.
{
In the simplest scenario, with three states, it turns out that, in general, overlaps alone are not enough to decide about negativity or imaginarity of third-order invariants, but if one restricts to real amplitude states, then overlap nonclassicality and negativity of the third-order invariant are mutually exclusive (complementary) witnesses of coherence.
}}

{\subsection{Negativity, imaginarity, and set coherence}}

{\subsubsection{Negativity and imaginarity as coherence witnesses}\label{subsubsec: negativity imaginarity witness coherence}}

{
As reviewed in Sec.~\ref{sec: nonclassicality of Bargmann invariants preliminaries},
the assumption of set incoherence for a set $\{\rho_i\}_i$ of quantum states, described in Def.~\ref{def: set-coherence}, imposes restrictions on tuples of Bargmann invariants among states drawn from that set.
This was extensively studied in Refs.~\cite{galvao2020quantum, wagner2022inequalities} for the case of overlaps, and briefly mentioned in an appendix of Ref.~\cite{oszmaniec2021measuring} for more general Bargmann invariants.
One such restriction is that all the invariants must necessarily be real and nonnegative (in fact, must take values in the interval $[0,1]$), as we now see.
}

{
Consider a set of states $\{\rho_i\}_i$ in a finite-dimensional Hilbert space $\mathcal{H}$,
and suppose that it is set incoherent.
This means that its elements are simultaneously diagonalizable, i.e., there is a basis with respect to which all those states $\rho_i$ are represented by a diagonal matrix, say $\text{diag}(p_i^{(1)}, \dots, p_i^{(d)})$, where $d=\dim(\mathcal{H})$.
Any Bargmann invariant drawn from this set then takes the form
\begin{equation}\label{eq: classical Bargmann}
    \Delta_n(\rho_{i_1}, \ldots, \rho_{i_n}) = \sum_{k=1}^d p_{i_1}^{(k)}\dots p_{i_n}^{(k)}.
\end{equation}
This value can be understood as the probability that, upon measuring the $n$ states $\rho_{i_1}, \ldots, \rho_{i_n}$
with an observable corresponding to the basis that diagonalizes all the states in $\{\rho_i\}_i$, they all return the same outcome.
}

Nonreal or negative values { are {blatantly} at odds with this interpretation, and thus immediately rule out a set-incoherent realization}.
{More generally, the specific form in Eq.\eqref{eq: classical Bargmann} constraints the possible values of Bargmann invariants---and of tuples of Bargmann invariants---admitting a set-incoherent realization.}

{Let us remark that Bargmann nonclassicality, as in Def.~\ref{def: relational nonclassicality}, is not equivalent to  set coherence or incompatibility, from Def.~\ref{def: set-coherence}. Bargmann nonclassicality is associated with functions of Bargmann invariants attaining values inaccessible by set-incoherent realizations.
Each departure from such a set incoherent explanation signals a specific form of nonclassicality.
A few examples that we have already encountered are (i) $\text{Re}[\Delta_n(\rho_1,\ldots,\rho_n)]<0$ (negativity), (ii) $\text{Im}[\Delta_n(\rho_1,\ldots,\rho_n))] \neq 0$ (imaginarity), and (iii) violations of overlap inequalities as in Eq.~\eqref{eq: preliminaries 3-cycle inequalities} ({interpretable as coherence from two-state overlaps, as per Ref.~\cite{wagner2022inequalities}}),
each of which expresses an inequivalent form of Bargmann nonclassicality.}
{Imaginarity, in particular, received some attention in Ref.~\cite{oszmaniec2021measuring}.}

{\subsubsection{Nonclassicality in weak values and KD distributions as relational nonclassicality}}

{
Negativity and imaginarity have also been studied as markers of nonclassicality for weak values or KD distributions.
Given our rendering of these quantities in terms of Bargmann invariants, the observations above admit a reading in those terms.
The following two lemmas {are an instantiation of the observations presented in Sec.~\ref{subsubsec: negativity imaginarity witness coherence} and} make explicit this trivial---albeit important---connection between set coherence and negativity or imaginarity of weak values or KD distributions.}

{
\begin{lemma}\label{lemma: weak value negativity and imaginarity are set-coherence}
In a finite-dimensional  Hilbert space $\mathcal{H}$, let $\ket{\phi}$, $\ket{\psi}$ two non-orthogonal vectors,
$A: \mathcal{H} \to \mathcal{H}$ a Hermitian operator,
and $a \in \text{Spec}(A)$ an eigenvalue of $A$ with corresponding eigenstate $\ket{a}$.
Then, negativity or imaginarity of the weak value 
$P^a_w$ of the eigenprojector $P^a = \ket{a}\!\bra{a}$, relative to pre- and post-selected states $\psi$ and $\phi$,
witness the set coherence of $\{\ket{a},\ket{\psi}, \ket{\phi}\}$.
\end{lemma}
}

{
\begin{lemma}\label{lemma: KD negativity and imaginarity are set-coherence}
Let $\{\vert i\rangle\}_{i\in I}, \{\vert f \rangle\}_{f \in F}$ be two (reference) orthonormal bases of a finite-dimensional Hilbert space $\mathcal{H}$,
and $\rho \in \mathcal{D}(\mathcal{H})$ be a state.
Then, negativity or imaginarity of $\xi(\rho \vert i, f)$, the value of the KD distribution associated with $\rho$ at a specific phase-space point $(i,f) \in I \times F$,
witness the set coherence of $\{\rho, \ket{i}, \ket{f}\}$.
An analogous statement also holds for any extended KD distribution.
\end{lemma}
}

{Note that for weak values, given that we assume pure pre- and post-selections, the measurement of $\Delta_2(\ket{\psi},\ket{\phi})$ alone guarantees coherence.} Recently, Ref.~\cite{wagner2023simple} considered in more depth the implications of Lemma \ref{lemma: weak value negativity and imaginarity are set-coherence}, showing that \textit{any} anomalous weak values require coherence, in the more general case where pre- and post-selection states can be mixed.

{\subsubsection{Set coherence beyond negativity and imaginarity}}

The analysis of Bargmann invariants for diagonal states sheds some light on the reason why negativity and imaginarity
require set coherence or noncommutativity, as noticed in Ref.~\cite{arvidsson_Shukur2021conditions}{, while the converse is not the case}.

Indeed, as already discussed, these conditions {on} Bargmann invariants are only two among \textit{many} that such quantities should satisfy in order to corroborate the assumptions connected to the classical probabilistic interpretation provided by Eq.~\eqref{eq: classical Bargmann}.
It is, therefore, {to be expected} that negativity or imaginarity \textit{alone} are incapable of completely characterizing set coherence.
To illustrate this point, consider the {{ following three} states:
{
\begin{equation}
    \begin{aligned}
      \ket{\psi_1} &\defeq \ket{0}, \\
      \ket{\psi_2} &\defeq \frac{\ket{0} + \sqrt{3}\ket{1}}{2}, \\
      \ket{\psi_3} &\defeq \frac{\sqrt{3}\ket{0} + \ket{1}}{2}. 
    \end{aligned}
\end{equation}}%
These states {
satisfy $\Delta_3(\psi_1,\psi_2,\psi_3) = \sfrac{3}{8} > 0$ and similarly for the other orientation.
As such, all Bargmann invariants of third (and therefore of \textit{any}) order among states drawn from the set $\{\ket{\psi_1},\ket{\psi_2},\ket{\psi_3}\}$ are real and positive.}
{However, this set of three states is} set coherent. This} can be seen from the fact that they violate the coherence-witnessing inequalities from Refs.~\cite{galvao2020quantum,giordani2021witnesses}{, reviewed in  Section~\ref{sec: nonclassicality of Bargmann invariants preliminaries}}. { Concretely}, they violate the inequality
\begin{equation}
    -\Delta_2(\psi_1,\psi_2)+\Delta_2(\psi_1,\psi_3)+\Delta_2(\psi_2,\psi_3)\leq 1 ,
\end{equation}
whose left-hand side, for these states, evaluates to 
\begin{equation}
    -\frac{1}{4}+\frac{3}{4}+\frac{3}{4} \;=\; \frac{5}{4} \;>\; 1.
\end{equation}
Later on, we { revisit} these two-state overlap inequalities in a different context{; see Eq.~\eqref{eq: overlap inequalities}}.
For proof that these inequalities are witnesses of coherence, we refer the reader to  Refs.~\cite{galvao2020quantum, giordani2021witnesses} {(see also the brief summary in Sec.~\ref{sec: nonclassicality of Bargmann invariants preliminaries})}.\\

{\subsection{Can negativity and imaginarity be witnessed from overlaps?}}\label{subsec: negativity and imaginarity from overlaps}

{
We have seen that the negativity or imaginarity of (any) Bargmann invariant witness set coherence.
In other situations, knowledge of two-state overlaps is enough to witness set coherence; see Refs.~\cite{galvao2020quantum, wagner2022inequalities}.

A natural question is to understand the extent to which such  witnesses capture distinct aspects of nonclassicality.
Specifically, we ask whether---or to what extent---it is possible to infer negativity and imaginarity of higher-order Bargmann invariants from the knowledge of overlaps (i.e., second-order Bargmann invariants) alone.
}

{\subsubsection{Dimension-specific KD negativity and imaginarity from overlaps in the whole phase space}}

{We first review a known result of this kind.
Refs.~\cite{debievre2021complete,deBivre2023relating, arvidsson_Shukur2021conditions, xu2022classification, xu2022kirkwood}
identified conditions that imply
the presence of \textit{some} negativity or imaginarity in the KD distribution of a pure state $\ket{\psi}$,
using only knowledge of the overlaps between $\ket{\psi}$ and the states in each of the reference bases.
In fact, the condition only requires \textit{possibilistic} information about such overlaps, i.e., whether they are equal to zero (orthogonality) or not.

Concretely, let $\{\ket{i}\}_{i\in I}$ and $\ket{f}_{f \in F}$ be two orthonormal bases of a Hilbert space of dimension $d$,
to be used as reference for the KD distribution,
and make the usual assumption that $\langle i \vert f \rangle \neq 0$ for all $i,f$.
Given a pure state $\ket{\psi}$, define
\begin{equation}
    n_I(\psi) \defeq \left\vert\{i\in I \colon \langle i \vert \psi \rangle \neq 0\}\right\vert
\end{equation}
to be the number of elements of the basis $\{\ket{i}\}_{i\in I}$ not orthogonal to $\ket{\psi}$, or equivalently, with nonzero overlap $\Delta_2(i,\psi)$;
and similarly define $n_F$ with respect to the basis $\ket{f}_{f \in F}$.

Then, one} obtains negativity or imaginarity {at \textit{some} phase-space point $\xi(\psi \vert i,f)$
if }
\begin{equation}\label{eq: Bievre}
    n_I(\psi)+n_F(\psi)>d+1.
\end{equation}
{%

If, moreover, the two reference bases are mutually unbiased, the KD distribution $\xi(\psi \vert \cdot)$ is classical, i.e., it has real and non-negative values at every phase-space point,}
if and only if~\cite{xu2022kirkwood}
\begin{equation}\label{eq: xu}
    n_I(\psi)n_F(\psi) = d.
\end{equation}

Note that, in both cases, { the condition depends on the} dimension of the physical system. Moreover, in order to apply those results, one needs (some) information about all the overlaps $\Delta_2(i,f),\Delta_2(i,\psi),\Delta_2(\psi,f)$, namely whether they are zero or not.

{
We aim to address a related question, still about witnessing negativity or imaginarity from overlaps alone, but adopting a somewhat different perspective.
First, we focus on one phase-space point at a time, i.e., on each triplet $\ket{\psi}$, $\ket{i}$, $\ket{f}$ with fixed $i, f$.
Second, we are interested in dimension-independent witnesses, which avoid assuming a specific Hilbert state dimension.}\\

{\subsubsection{Inferring higher-order Bargmanns from overlaps: the three-state scenario}}\label{subsubsec: three state scenario}

{
The preceding discussion leads us to consider the minimal scenario where the question of witnessing negativity or imaginarity from two-state overlaps is meaningful.
Suppose that there are three (unknown) pure states $\ket{\psi_1}$, $\ket{\psi_2}$, $\ket{\psi_3}$,
and that we are able to obtain knowledge of the three pairwise overlaps
$\Delta_2(\psi_1,\psi_2), \Delta_2(\psi_2,\psi_3)$, and $\Delta_2(\psi_1,\psi_3)$.
We ask the following question:
is such information enough to infer about the negativity or imaginarity of the third-order Bargmann invariant $\Delta_3(\psi_1,\psi_2,\psi_3)$?

In the remainder of this Sec.~\ref{subsec: negativity and imaginarity from overlaps}, we obtain and review some impossibility results and some partial results in this direction.

We set some notation to facilitate precisely stating our results about the three-state overlap scenario.
Recall from Eq.~\eqref{eq:Q3} the definition of the set $\mQt \subseteq [0,1]^3$ 
of overlap triples realizable by quantum states, and from the preceding paragraph in Sec.~\ref{sec: nonclassicality of Bargmann invariants preliminaries} that of the subset $\mCt \subseteq \mQt$ of those overlap triples realizable by set-incoherent states.
We are interested in distinguishing various cases for the (not directly available) value of the third-order invariant among the states.
We write $\mRt$ (resp. $\mIt$, $\mPt$, $\mNt$, $\mZt$)
for the set of overlap triples realizable by three states $\rho_1, \rho_2, \rho_3$ whose third-order invariant $\Delta_3(\rho_1,\rho_2,\rho_3)$ is real (resp. nonreal, positive, negative, zero);
for example,
\begin{equation}
    \mRt = \left\{\left(\begin{matrix}
        \Delta_2(\rho_1,\rho_2)\\\Delta_2(\rho_1,\rho_3)\\ \Delta_2(\rho_2,\rho_3) 
    \end{matrix}\right) \in [0,1]^3 \colon \begin{matrix} \rho_1, \rho_2, \rho_3 \in \mathcal{D}(\mathcal{H}), \\[3pt] \text{Im}[\Delta_3(\rho_1,\rho_2,\rho_3)] \neq 0 \end{matrix} \right\}.
\end{equation}
Many of our results concern realizations with \textit{pure} states.
We write $\pQt$ for the overlap triples realizable by three \textit{pure} states,
\begin{equation}
    \pQt = \left\{\left(\begin{matrix}
        \Delta_2(\psi_1,\psi_2)\\\Delta_2(\psi_1,\psi_3)\\ \Delta_2(\psi_2,\psi_3) 
    \end{matrix}\right) \in [0,1]^3 \colon  \ket{\psi_1}\!, \ket{\psi_2}\!, \ket{\psi_3}\! \in \mathcal{H} \right\}.
\end{equation}
Similarly, $\pCt$, $\pRt$, $\pIt$, $\pPt$, $\pNt, \pZt$
denote the sets of overlap triples admitting a realization with pure states that moreover satisfies the condition corresponding to each letter as above.
Finally, we use the superscript $\wdimd{(\cdot)}$ to restrict the condition to realizations by states in $d$-dimensional Hilbert space.
Trivially, from reasoning at the level of realizations,
$\wdimd{\pXt} \subseteq \wdimd{\mXt}$, $\wdim{\aXt}{d} \subseteq \wdim{\aXt}{d+1}$, and $\aXt = \bigcup_d \wdim{\aXt}{d}$, for all $\cX \in \{ \cQ, \cC, \cR, \cI, \cP, \cN, \cZ\}$ and ${\circbullet} \in \{{\circ},{\bullet}\}$.

By the same token, it is clear  that $\aPt \cup \aNt \cap \aZt = \aRt$ and $\aRt \cup \aIt = \aQt$,
but observe that, crucially, these unions are not \textit{a priori} disjoint.
This is because an element of these sets does not carry all the information about a realization, but only the values of the overlaps among the three realizing states;
a particular overlap triple typically admits more than one realization.
Such limited information might not even be sufficient to distinguish real and imaginary values (or positive and negative values) for the third-order invariant of the realizing states.
In other words, the same tuple of overlaps could be realizable by different sets of three states, with different values for the third-order invariant, including negative and positive, or with and without imaginary part.\footnote{The case $\pZt$ when the third-order Bargmann invariant $\Delta_3(\psi_1,\psi_2,\psi_3)$ is zero is easier to separate because it implies one of the pairwise two-order invariants also vanish.} Our first result below (Theorem~\ref{theorem: impossibility overlaps}) establishes that this indeed happens.

Before stating it, we make a useful remark regarding imaginarity. 
A set of states can be simultaneously represented with only real amplitudes (i.e., with respect to the same basis) if and only if all its Bargmann invariants are real. 
Consequently, in the case of pure states, since a set of three states is fully characterized by Bargmanns up to order $3$, the set
$\pRt$ is equivalently described as the set of overlap values realizable by a choice of three pure states with \textit{real amplitudes}\footnote{For mixed states $\rho_1, \rho_2, \rho_3$, it is conceivable that $\Delta_3(\rho_1\rho_2\rho_3)$ is real but imaginarity still appears in higher-order invariants, such as $\Delta_4(\rho_1,\rho_2,\rho_3,\rho_2)$. As far as we know, it is an open question whether this can happen. For now, only one of the inclusions holds: the set of overlaps realizable by three states with real density matrices is contained in $\mRt$.}.

}

\subsubsection{Overlaps are, in general, not enough for negativity and imaginarity}

We now show that overlaps are, in general, not sufficient to decide negativity or imaginarity of third-order invariants, even in the case of pure states.

\begin{theorem}
    \label{theorem: impossibility overlaps} 
    Let $\ket{\psi_1}, \ket{\psi_2}, \ket{\psi_3}$ be three (unknown) pure states in a finite-dimensional Hilbert space $\mathcal{H}$.
    Suppose that one is given the values of the three pairwise overlaps 
    $\Delta_2(\psi_1,\psi_2), \Delta_2(\psi_2,\psi_3)$, and $\Delta_2(\psi_1,\psi_3)$.
    \begin{enumerate}
        \item In general, it is not possible to discriminate between real and strictly complex third-order Bargmann invariant, i.e., to decide whether $\mathrm{Im}[\Delta_3(\psi_1,\psi_2,\psi_3)]=0$ or $\mathrm{Im}[\Delta_3(\psi_1,\psi_2,\psi_3)]\neq 0$.
    
        \item Further assuming that $\Delta_3(\psi_1,\psi_2,\psi_3)$ is real, it is in general not possible to discriminate between positive and negative third-order Bargmann invariants, i.e., to decide
        whether
        $\Delta_3(\psi_1,\psi_2,\psi_3)>0$ or $\Delta_3(\psi_1,\psi_2,\psi_3)<0$.
    \end{enumerate}
    In summary, $\pRt\cap\pIt \neq \varnothing$ and  $\pPt\cap\pNt \neq \varnothing$.
\end{theorem}
\begin{proof}
    We prove statement 1 by constructing a simple counterexample.
    The idea is to find realizations with equal overlaps
    but different imaginary parts of the third-order invariant.
    
    Considering the triplet of states $(\ket{0}, \ket{+}, \ket{i+})$, where $\ket{i+} = (\ket{0} +i \ket{1})/\sqrt{2}$,
    we have that all pairwise overlaps are equal to $\sfrac{1}{2}$ while their third-order invariants are $\Delta_3 = (1 \pm i)/4$ (with the sign depending on the order), hence $\text{Im}[\Delta_3] \neq 0$.
    
    We now give a different triplet of states with all pairwise overlaps equal to $\sfrac{1}{2}$,
    but whose states can all be described with only real amplitudes.
    The simplest such example would be the triplet consisting of three maximally mixed qubits, i.e., $(\sfrac{\mathbb{1}}{2},\sfrac{\mathbb{1}}{2},\sfrac{\mathbb{1}}{2})$.
    This would suffice to establish the result for mixed states, or even to show
    $\pIt \cap \mRt \neq \varnothing$ (even more specifically, $\wdim{\pIt}{2} \cap \wdim{\mRt}{2} \neq \varnothing$).
    For the stated result, however, we need an example with pure states. The
    following  states (in a Hilbert space of dimension $\geq 3$) have real amplitudes and pairwise overlaps $\sfrac{1}{2}$:
    \begin{equation}
        \frac{\ket{0} + \ket{1}}{\sqrt{2}}, \frac{\ket{1} + \ket{2}}{\sqrt{2}}, \frac{\ket{0} + \ket{2}}{\sqrt{2}}  .
    \end{equation}
    We have thus established that $(\sfrac{1}{2},\sfrac{1}{2},\sfrac{1}{2}) \in \wdim{\pIt}{2} \cap \wdim{\pRt}{3}$.
    
    We now proceed to prove statement 2. We do so, again, by constructing a counterexample where two different realizations return the same two-state overlaps but third-order invariants {that are real and} of opposite sign.
    For the three states
    \begin{equation}
        \ket{0}, \frac{\ket{0}  + \sqrt{3}\ket{1}}{2}, \frac{\ket{0} - \sqrt{3}\ket{1}}{2},
    \end{equation}
    all three overlaps equal $\sfrac{1}{4}$, while $\Delta_3 = -\sfrac{1}{8} < 0$.
    On the other hand, for the three states
    \begin{equation}
        \ket{0}, \frac{\ket{0} + \sqrt{3}\ket{1}}{2}, \frac{3\ket{0} + \sqrt{3}\ket{1} + \sqrt{24}\ket{2}}{6},
    \end{equation}
    all overlaps, once more, equal $\sfrac{1}{4}$, while $\Delta_3 = +\sfrac{1}{8} > 0$.
   This shows that $(\sfrac{1}{4},\sfrac{1}{4},\sfrac{1}{4}) \in \wdim{\pNt}{2} \cap \wdim{\pPt}{3}$.
\end{proof}

The theorem above shows that the knowledge of two-state overlaps for a given triple of states, in general, provides insufficient information about their third-order invariant
to decide its imaginarity or negativity.
The following theorem, pieced together from results in Ref.~\cite{galvao2020quantum}, is a result in a similar direction, for imaginarity.

\begin{theorem}\label{theorem: triplets are real}
    {A triple of overlaps is realizable by three quantum states if and only if it is a convex combination of triples of overlaps realizable by pure states with real amplitudes.
    In other words, $\mQt$ is the convex hull of $\pRt$.}
\end{theorem}

\begin{proof}
    It was shown in Ref.~\cite{galvao2020quantum} that all {the overlap triples that are }extremal points of the set $\mQt$ 
    are realizable using real-amplitude pure states. For completeness, we proceed to sketch the proof of this result.

    First, note that a Bargmann invariant for mixed states can be written as a convex combination of Bargmann invariants of pure states, using a convex decomposition of each mixed state into pure states. Fixing a decomposition for each mixed state in a set $\{\rho_i\}_i$, all Bargmann invariants drawn from the set can be written as the \textit{same} convex mixture of Bargmann invariants among pure states. In particular, this applies to overlap triples, hence $\mQt$ is contained in the convex hull of $\pQt$. It thus suffices to look at realizations by pure states to characterize the border of $\mQt$.
    
    Now, any three pure states can be unitarily sent to a concrete representation of the form
    \begin{align*}
        \ket{\psi_1} &= \ket{0}, \\
        \ket{\psi_2}  &= \cos(\beta)\ket{0} + \sin(\beta)\vert 1\rangle,\\
        \ket{\psi_3} &= \cos(\gamma)\ket{0} + e^{i\phi}\sin(\gamma)\sin(\alpha)\vert 1 \rangle + \sin(\gamma)\cos(\alpha)\vert 2 \rangle,
    \end{align*}
    with $\alpha,\beta,\gamma \in [0,\sfrac{\pi}{2})$ and $\phi \in [0,2\pi)$.
    This fixes $\Delta_2(\psi_1,\psi_2) = \cos^2(\beta)$ and $\Delta_2(\psi_1,\psi_3) = \cos^2(\gamma)$ to be real.
    The existence of imaginary values in the representation of the states may therefore only affect the third overlap $\Delta_2(\psi_2,\psi_3)$.
    To conclude the argument, the maximal values achievable by $\Delta_2(\psi_2,\psi_3)$ for any fixed values of $\beta, \gamma$ can be found as the extreme points of the multivariate function $\Delta_2(\psi_2(\alpha,\phi),\psi_3(\alpha,\phi))$. Ref.~\cite[Appendix A]{galvao2020quantum} performs these calculations explicitly, showing this happens either if $\sin(\alpha)= 0$ or if $\sin(\phi)= 0$, both conditions that imply the three states have real amplitudes.
    We conclude that the {overlap triples in the} boundary of $\pQt$ are realizable by pure states with real amplitudes only. Therefore, $\mQt$ is in the convex hull of $\pRt$.

    For the other direction, it suffices to establish that $\mQt$ is convex. In fact, $\pQt$ (or even $\pRt$) is.
    This was observed numerically in  Ref.~\cite{galvao2020quantum}. Moreover, its bounds were analytically found to correspond to the boundaries of an \textit{elliptope}, discussed in detail in Ref.~\cite{le2023quantum}.\footnote{There is a dual relationship between the scenarios investigated in Ref.~\cite{le2023quantum} and those in Ref.~\cite{galvao2020quantum}. For more information about these aspects, we refer the reader to the discussion in Ref.~\cite{wagner2022inequalities}.}.
    More formally, we establish convexity by showing that if a triple of overlaps $(\Delta_{12},\Delta_{13},\Delta_{23})$ is pure-state realizable,
    then so is any triple in the line segment connecting it to $(0,0,0)$, i.e. any triple of the form $(\lambda \Delta_{12}, \lambda \Delta_{13}, \lambda \Delta_{23})$ for $\lambda \in [0,1]$.
    Starting from $\ket{\psi_1}\!, \ket{\psi_2}\!, \ket{\psi_3}\!$ the pure states in a Hilbert space $\mathcal{H}$ realizing $(\Delta_{12},\Delta_{13},\Delta_{23})$, 
    consider the larger Hilbert space $\mathcal{H} \oplus \mathbb{C}^3$ with $\{\ket{e_i}\}_{i=1}^3$ a basis for the additional summand, and take $\ket{\psi'_i} \defeq \sqrt{\sqrt{\lambda}}\ket{\psi_i} \oplus \sqrt{1-\sqrt{\lambda}}\ket{e_i}$.
    Then, for $i \neq j$, $\langle \psi'_i \vert \psi'_j \rangle = \sqrt{\lambda} \langle \psi_i \vert \psi_j \rangle + (1-\sqrt{\lambda}) \langle e_i \vert e_j \rangle = \sqrt{\lambda} \langle \psi_i \vert \psi_j \rangle $,
    and so $\Delta_2(\psi'_i,\psi'_j) = \lambda \Delta_2(\psi_i,\psi_j)$. Therefore, $\{\ket{\psi'_i}\}_{i=1}^3$ realizes $(\lambda \Delta_{12}, \lambda \Delta_{13}, \lambda \Delta_{23})$, as wanted.
\end{proof}

In other words, this theorem shows that the realizable overlap triples can be described as convex combinations of overlap triples realizable by pure states that are real with respect to some reference basis.
If Theorem~\ref{theorem: impossibility overlaps} showed that an overlap triple does not \textit{always} give enough information to decide imaginarity of the third-order Bargmann, Theorem~\ref{theorem: triplets are real} seems to indicate that it can \textit{never}---or \textit{hardly}---provide useful information in this regard.

These results highlight that, to test imaginarity  using overlaps, one needs to {make} further assumptions on the states, such as on purity or the Hilbert space dimension.
Indeed, a careful look at the proof of Theorem~\ref{theorem: impossibility overlaps} suggests there is scope for such examples if one restricts to \textit{pure} single-qubit states, $\wdim{\pQt}{2}$. 

It turns out that, for instance, three \textit{pure} single-qubit states that have pairwise overlaps of $\sfrac{1}{2}$ must have an imaginary third-order invariant.

\subsubsection{Negativity and overlaps for real amplitude states}

In contrast to imaginarity, for negativity, we are able to obtain some positive results:
restricting to realizations with real-amplitude states (i.e., to $\pRt$), it is possible to draw conclusions about the negativity of third-order Bargmann invariants given only their pairwise overlaps.
In fact, it turns out that negativity and overlap nonclassicality, as witnessed by the overlap inequalities such as Eq.~\ref{eq: preliminaries 3-cycle inequalities} studied in Refs.~\cite{galvao2020quantum, wagner2022inequalities, wagner2022coherence}, witness mutually exclusive forms of set coherence.

We first consider the two-dimensional case, for which we are able to obtain a full discriminating result, partitioning $\wdim{\pRt}{2}$.
Given pure states $\ket{\psi_1}, \ket{\psi_2}, \ket{\psi_3}$, the following inequalities are witnesses of set coherence: the first three are the overlap inequalities defining the polytope $\mCt$, the last is nonnegativity of the third-order Bargmann invariant.
\begin{equation}\label{eq: overlap inequalities}
\begin{aligned}
        +\Delta_2(\psi_1,\psi_2)+\Delta_2(\psi_2,\psi_3)-\Delta_2(\psi_1,\psi_3) &{ \leq } 1,\\
         +\Delta_2(\psi_1,\psi_2)-\Delta_2(\psi_2,\psi_3)+\Delta_2(\psi_1,\psi_3) &{ \leq } 1,\\
          -\Delta_2(\psi_1,\psi_2)+\Delta_2(\psi_2,\psi_3)+\Delta_2(\psi_1,\psi_3) &{ \leq } 1,\\
        \Delta_3(\psi_1,\psi_2,\psi_3) & {\geq} 0.
    \end{aligned}
\end{equation}
If the $\psi_i$ are one-qubit pure states with real amplitudes with respect to some basis, it turns out that no two of these inequalities can be violated.
For any such choice of three states, there exists a unitary that maps them to the rebit subspace.
This implies that we can parametrize the states in the following way;
\begin{equation}
    \begin{aligned}
        { \ket{\psi_1} } &= \ket{0}, \\
        { \ket{\psi_2} }    &= \cos(\theta) \ket{0} + \sin(\theta)\ket{1}, \\
        { \ket{\psi_3} }    &= \cos(\phi) \ket{0} + \sin(\phi)\ket{1}.
    \end{aligned}
\end{equation}
with $\theta, \phi \in [0,\pi]$ being sufficient to capture all possible values. This yields the following values for the Bargmann invariants:
\begin{equation}
    \begin{aligned}
    \Delta_2(\psi_1,\psi_2)&=\cos^2(\theta), \\
    \Delta_2(\psi_1,\psi_3)&=\cos^2(\phi), \\
    \Delta_2(\psi_2,\psi_3)&=\cos^2(\theta-\phi), \\
    \Delta_3(\psi_1,\psi_2,\psi_3) &= \cos(\theta-\psi)\cos(\theta)\cos(\phi).
    \end{aligned}
\end{equation}
The domains in which each of the inequalities from Eq.~\eqref{eq: overlap inequalities} is violated are plotted in Fig.~\ref{fig:real_qubits_negativity}, in terms of the parameters $\theta$ and $\phi$.
There, we can observe that each value of $\theta$ and $\phi$ defining a triplet of states, except when those states are set-incoherent corresponding to borders of the colored regions, violates one and only one inequality. 

We thus conclude that $\wdim{\pRt}{2}$ partitions into two \textit{disjoint} subsets: $\wdim{\pPt}{2}$ and $\wdim{\pCt}{2}$, with the latter being the disjoint union of $\wdim{\pNt}{2}$ and $\wdim{\pZt}{2}$ (the only triples with a set-incoherent realization, where at least one of the pairwise overlaps vanishes).

\begin{figure}
    \centering
    \includegraphics[width=\columnwidth]{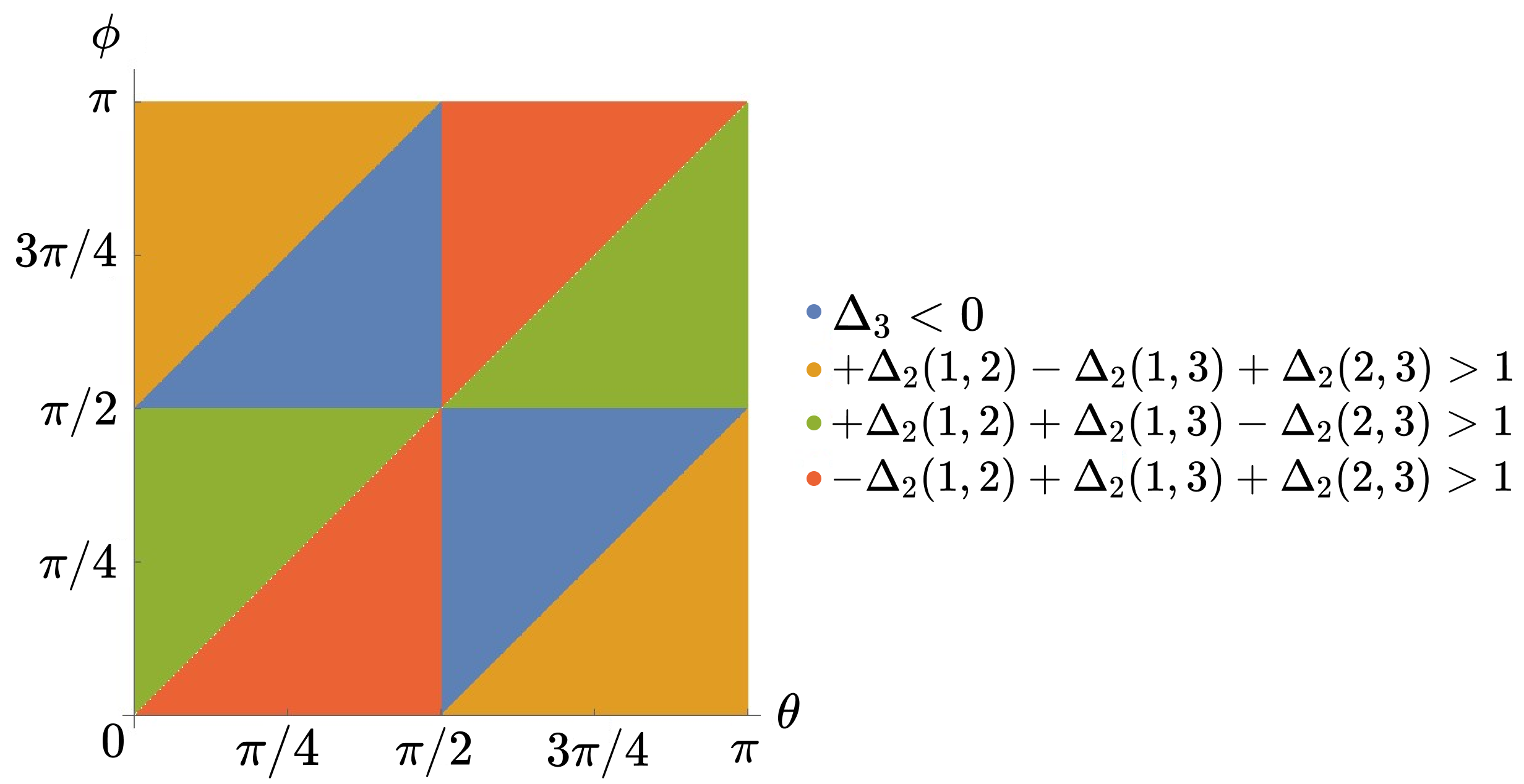}
    \caption{\textbf{Negativity of third-order invariants from measurements of overlaps only.} For real-amplitude one-qubit states, we show a complementarity relation between negativity of third-order invariants and overlap inequalities. Here we use the notation $\Delta_2(i,j) = \vert \langle \psi_i \vert \psi_j \rangle \vert^2$ and $\Delta_3 = \langle \psi_1 \vert \psi_2 \rangle \langle \psi_2 \vert \psi_3 \rangle \langle \psi_3 \vert \psi_1 \rangle$.}
    \label{fig:real_qubits_negativity}
\end{figure}

Turning our attention to arbitrary dimension (i.e., to $\pRt$), this observation leads us to search for similar results, converging on the lemma below.
Even in this more general case, violation of overlap inequalities and negativity of $\Delta_3$ are mutually exclusive, i.e., $\pNt \subseteq \pCt$.

\begin{lemma}\label{lemma: negativity from overlaps}
    Let $\ket{\psi_1}\!, \ket{\psi_2}\!, \ket{\psi_3}$ be pure states {that are} real with respect to some basis{, in} a finite-dimensional Hilbert space $\mathcal{H}$ of dimension $d$.
    Then, violation of any of the inequalities
    \begin{equation}
        \begin{aligned}
            &\Delta_{2}(\psi_1,\psi_2)+\Delta_{2}(\psi_1,\psi_3)-\Delta_{2}(\psi_2,\psi_3) \leq 1 ,\\
            &\Delta_{2}(\psi_1,\psi_2)-\Delta_{2}(\psi_1,\psi_3)+\Delta_{2}(\psi_2,\psi_3) \leq 1 ,\\
            -&\Delta_{2}(\psi_1,\psi_2)+\Delta_{2}(\psi_1,\psi_3)+\Delta_{2}(\psi_2,\psi_3) \leq 1 .
        \end{aligned}
    \end{equation}
    implies that 
    \begin{equation}
        \Delta_{3}(\psi_1,\psi_2,\psi_3) > 0.
    \end{equation}
    In other words, $(\pRt \setminus \pCt) \subseteq \pPt$ and thus $\pNt \subseteq \pCt$.

    If, moreover, $d=2$, the converse also holds: if all overlap inequalities are strictly satisfied (i.e., strictly less than $1$), we obtain $\Delta_{3}<0$.
\end{lemma}

We prove this lemma in Appendix \ref{appendix: proof of the lemma}. The part of the lemma regarding dimension $d=2$ was already shown by our previous argument and is depicted in Fig.~\ref{fig:real_qubits_negativity}. Overlaps provide a simple test for when the third-order invariants are positive, in which case it might simplify the search for negativity.

\section{Discussion and future directions}\label{sec: conclusion}

In this work, we have presented two distinct approaches for the study of central concepts in quantum mechanics, namely, weak values, KD distributions, and state spectra. 

In this work, we {have} explored two complementary approaches to a better understanding of various selected constructions useful in quantum information processing{, including weak values, KD distributions, and state spectra.} The first approach is more practical and application-focused, where we show simple quantum circuits that can be used to estimate various constructions of fundamental interest by measuring multivariate traces known as Bargmann invariants. Specifically, we have described circuits to experimentally estimate the spectrum of any mixed quantum state, weak values, KD quasiprobability distribution, post-selected quantum Fisher information, and OTOCs. The second approach is more foundational and investigates conditions for the nonclassicality of these functions, {in particular,} negativity and imaginarity, relating it to the presence of a recently introduced quantum resource termed set coherence, which acknowledges coherence as a basis-independent, relational property of a set of states.

These results bring proof-of-principle tests of the nonclassicality of quantum theory---a vital resource for information processing---closer to current experimental capabilities. Indirect measurements of unitary invariants up to fourth order have been made in different platforms~\cite{jones2022distinguishability,bitter1987manifestation,suter1988study}. We believe that the unified view of nonclassicality we discussed here may boost the interest in a systematic exploration of nonclassicality witnessed and quantified by unitary invariants.
Moreover, the connection { between} coherence theory {and} all the above constructions motivates not only the { development} of a formal resource theory based on Bargmann invariants but also the experimental investigation of coherence for higher-dimensional systems through the lens of { resource-theoretic} monotones.

The relationship between Bargmann invariants and KD distributions could bring quantitative understanding to other problems of interest. For instance, nonclassicality of Bargmann invariants understood as { witness} of coherence may clarify the connection between the nonclassicality of OTOCs and scrambling of information~\cite{gonzalez2019out} for quantum computation. This new view could lead to more robust and rigorous ways of benchmarking the scrambling of information against decoherence that goes beyond numerical evidence~\cite{harris2022benchmark}. Another interesting point to notice is that the tools described in Ref.~\cite{oszmaniec2021measuring} and further developed here clarify the understanding of when it is indeed \textit{crucial} to experimentally estimate higher-order invariants in extended KD distributions. 
It would be interesting to characterize the minimum cardinality and/or { order} of sets of Bargmann invariants that are necessary for quantifying nonclassicality in particular applications.

We believe that our circuits and the unified picture in terms of Bargmann invariants will motivate further progress in analyzing and quantifying nonclassicality and its relation with quantum information advantage for different tasks. As an example, in Ref.~\cite{wagner2022coherence}, the authors prove a quantum advantage for the task of quantum interrogation relying on a function of overlaps. Our work suggests that higher-order invariants may provide a more complete characterization of advantages for general tasks.

{Conceptually, we expect future work to consolidate---or reject---the view that violations of different constraints over Bargmann invariants, in the form of inequalities, are a helpful guide to identifying different nonclassical resources provided by quantum mechanics.
We view some of our results as preliminary steps in this direction.
A complete answer to these questions, together with a formal framework for addressing the possible consequences for quantum information science, has yet to be investigated.}

As a final technical remark, even though some of the results regarding the use of the univariate traces $\text{Tr}(\rho^n)$ for $n=2,\dots,d$ have previously appeared in the literature, we stress that our novel sample complexity analysis and numerical experiments have shown concrete benefits { in} learning the spectrum using the higher-order Bargmann invariants.
Interesting future work in this direction could involve experimentally probing $\text{Tr}(\rho^n)$ for small $n=2,3,4$ { and} using the { obtained} information to estimate the spectrum of density matrices.
These estimates can serve as subroutines in variational quantum eigensolvers aimed at preparing Gibbs states, such as those in Ref.~\cite{wang2021variational}.

Given the foundational importance of KD distributions and weak values, we expect that our unified framework based on measurements of Bargmann invariants will provide both theoretical insight and practical recipes for experimental implementations of the many applications we analyze here.

\begin{acknowledgments}
We would like to thank Stephan De Bi\`evre, Mafalda Francisco Ram\^oa da Costa Alves, Francesco Buscemi, and Agung Budiyono for providing useful comments on an early version of this work. We would also like to thank an anonymous referee for pointing out the simple yet elegant Eq.~\eqref{eq: KD extended as Bargmann Referee}, which provides a link between the extended KD distribution and higher-order Bargmann invariants.

We thank the International Iberian Nanotechnology Laboratory (INL) and the Bar-Ilan Institute of Nanotechnology and Advanced Materials (BINA) for their support, which was instrumental in kick-starting this collaboration.
RW and ARM acknowledge support from FCT---Fundação para a Ciência e a Tecnologia (Portugal) through PhD Grants SFRH/BD/151199/2021 and SFRH/BD/151453/2021, respectively.
RSB and EFG also acknowledge support from FCT, CEECINST/00062/2018.
ILP acknowledges support from the ERC Advanced Grant FLQuant.
The work of the INL-based authors (RW, ARM, RSB, EFG) was supported by the Digital Horizon Europe project FoQaCiA, GA no. 101070558.
EC was supported by the Israeli Innovation Authority under Project 73795 and the Eureka project, by Elta Systems Ltd., by the Pazy Foundation, by the Israeli Ministry of Science and Technology, and by the Quantum Science and Technology Program of the Israeli Council of Higher Education.
\end{acknowledgments}

\bibliography{bibliography}

\providecommand{\noopsort}[1]{}\providecommand{\singleletter}[1]{#1}%
\begin{thebibliography}{132}%
\makeatletter
\providecommand \@ifxundefined [1]{%
 \@ifx{#1\undefined}
}%
\providecommand \@ifnum [1]{%
 \ifnum #1\expandafter \@firstoftwo
 \else \expandafter \@secondoftwo
 \fi
}%
\providecommand \@ifx [1]{%
 \ifx #1\expandafter \@firstoftwo
 \else \expandafter \@secondoftwo
 \fi
}%
\providecommand \natexlab [1]{#1}%
\providecommand \enquote  [1]{``#1''}%
\providecommand \bibnamefont  [1]{#1}%
\providecommand \bibfnamefont [1]{#1}%
\providecommand \citenamefont [1]{#1}%
\providecommand \href@noop [0]{\@secondoftwo}%
\providecommand \href [0]{\begingroup \@sanitize@url \@href}%
\providecommand \@href[1]{\@@startlink{#1}\@@href}%
\providecommand \@@href[1]{\endgroup#1\@@endlink}%
\providecommand \@sanitize@url [0]{\catcode `\\12\catcode `\$12\catcode
  `\&12\catcode `\#12\catcode `\^12\catcode `\_12\catcode `\%12\relax}%
\providecommand \@@startlink[1]{}%
\providecommand \@@endlink[0]{}%
\providecommand \url  [0]{\begingroup\@sanitize@url \@url }%
\providecommand \@url [1]{\endgroup\@href {#1}{\urlprefix }}%
\providecommand \urlprefix  [0]{URL }%
\providecommand \Eprint [0]{\href }%
\providecommand \doibase [0]{https://doi.org/}%
\providecommand \selectlanguage [0]{\@gobble}%
\providecommand \bibinfo  [0]{\@secondoftwo}%
\providecommand \bibfield  [0]{\@secondoftwo}%
\providecommand \translation [1]{[#1]}%
\providecommand \BibitemOpen [0]{}%
\providecommand \bibitemStop [0]{}%
\providecommand \bibitemNoStop [0]{.\EOS\space}%
\providecommand \EOS [0]{\spacefactor3000\relax}%
\providecommand \BibitemShut  [1]{\csname bibitem#1\endcsname}%
\let\auto@bib@innerbib\@empty
\bibitem [{\citenamefont {Aharonov}\ \emph {et~al.}(1988)\citenamefont
  {Aharonov}, \citenamefont {Albert},\ and\ \citenamefont
  {Vaidman}}]{aharonov1988result}%
  \BibitemOpen
  \bibfield  {author} {\bibinfo {author} {\bibfnamefont {Y.}~\bibnamefont
  {Aharonov}}, \bibinfo {author} {\bibfnamefont {D.~Z.}\ \bibnamefont
  {Albert}},\ and\ \bibinfo {author} {\bibfnamefont {L.}~\bibnamefont
  {Vaidman}},\ }\bibfield  {title} {\bibinfo {title} {How the result of a
  measurement of a component of the spin of a spin-1/2 particle can turn out to
  be 100},\ }\href {https://doi.org/10.1103/PhysRevLett.60.1351} {\bibfield
  {journal} {\bibinfo  {journal} {Phys. Rev. Lett.}\ }\textbf {\bibinfo
  {volume} {60}},\ \bibinfo {pages} {1351} (\bibinfo {year}
  {1988})}\BibitemShut {NoStop}%
\bibitem [{\citenamefont {Kirkwood}(1933)}]{kirkwood1933quantum}%
  \BibitemOpen
  \bibfield  {author} {\bibinfo {author} {\bibfnamefont {J.~G.}\ \bibnamefont
  {Kirkwood}},\ }\bibfield  {title} {\bibinfo {title} {Quantum statistics of
  almost classical assemblies},\ }\href {https://doi.org/10.1103/PhysRev.44.31}
  {\bibfield  {journal} {\bibinfo  {journal} {Phys. Rev.}\ }\textbf {\bibinfo
  {volume} {44}},\ \bibinfo {pages} {31} (\bibinfo {year} {1933})}\BibitemShut
  {NoStop}%
\bibitem [{\citenamefont {Dirac}(1945)}]{dirac1945analogy}%
  \BibitemOpen
  \bibfield  {author} {\bibinfo {author} {\bibfnamefont {P.~A.~M.}\
  \bibnamefont {Dirac}},\ }\bibfield  {title} {\bibinfo {title} {On the analogy
  between classical and quantum mechanics},\ }\href
  {https://doi.org/10.1103/RevModPhys.17.195} {\bibfield  {journal} {\bibinfo
  {journal} {Rev. Mod. Phys.}\ }\textbf {\bibinfo {volume} {17}},\ \bibinfo
  {pages} {195} (\bibinfo {year} {1945})}\BibitemShut {NoStop}%
\bibitem [{\citenamefont {Halpern}\ \emph {et~al.}(2018)\citenamefont
  {Halpern}, \citenamefont {Swingle},\ and\ \citenamefont
  {Dressel}}]{yunger_Halpern2018quasiprobability}%
  \BibitemOpen
  \bibfield  {author} {\bibinfo {author} {\bibfnamefont {N.~Y.}\ \bibnamefont
  {Halpern}}, \bibinfo {author} {\bibfnamefont {B.}~\bibnamefont {Swingle}},\
  and\ \bibinfo {author} {\bibfnamefont {J.}~\bibnamefont {Dressel}},\
  }\bibfield  {title} {\bibinfo {title} {Quasiprobability behind the
  out-of-time-ordered correlator},\ }\href
  {https://doi.org/10.1103/physreva.97.042105} {\bibfield  {journal} {\bibinfo
  {journal} {Phys. Rev. A}\ }\textbf {\bibinfo {volume} {97}},\ \bibinfo
  {pages} {042105} (\bibinfo {year} {2018})}\BibitemShut {NoStop}%
\bibitem [{\citenamefont {Lostaglio}\ \emph {et~al.}(2023)\citenamefont
  {Lostaglio}, \citenamefont {Belenchia}, \citenamefont {Levy}, \citenamefont
  {Hern\'{a}ndez-G\'{o}mez}, \citenamefont {Fabbri},\ and\ \citenamefont
  {Gherardini}}]{lostaglio2022kirkwood}%
  \BibitemOpen
  \bibfield  {author} {\bibinfo {author} {\bibfnamefont {M.}~\bibnamefont
  {Lostaglio}}, \bibinfo {author} {\bibfnamefont {A.}~\bibnamefont
  {Belenchia}}, \bibinfo {author} {\bibfnamefont {A.}~\bibnamefont {Levy}},
  \bibinfo {author} {\bibfnamefont {S.}~\bibnamefont
  {Hern\'{a}ndez-G\'{o}mez}}, \bibinfo {author} {\bibfnamefont
  {N.}~\bibnamefont {Fabbri}},\ and\ \bibinfo {author} {\bibfnamefont
  {S.}~\bibnamefont {Gherardini}},\ }\bibfield  {title} {\bibinfo {title}
  {Kirkwood-{D}irac quasiprobability approach to the statistics of incompatible
  observables},\ }\href {https://doi.org/10.22331/q-2023-10-09-1128} {\bibfield
   {journal} {\bibinfo  {journal} {Quantum}\ }\textbf {\bibinfo {volume} {7}},\
  \bibinfo {pages} {1128} (\bibinfo {year} {2023})}\BibitemShut {NoStop}%
\bibitem [{\citenamefont {Cohen}\ and\ \citenamefont
  {Pollak}(2018)}]{cohen2018determination}%
  \BibitemOpen
  \bibfield  {author} {\bibinfo {author} {\bibfnamefont {E.}~\bibnamefont
  {Cohen}}\ and\ \bibinfo {author} {\bibfnamefont {E.}~\bibnamefont {Pollak}},\
  }\bibfield  {title} {\bibinfo {title} {Determination of weak values of
  quantum operators using only strong measurements},\ }\href
  {https://doi.org/10.1103/PhysRevA.98.042112} {\bibfield  {journal} {\bibinfo
  {journal} {Phys. Rev. A}\ }\textbf {\bibinfo {volume} {98}},\ \bibinfo
  {pages} {042112} (\bibinfo {year} {2018})}\BibitemShut {NoStop}%
\bibitem [{\citenamefont {Horodecki}\ \emph {et~al.}(2009)\citenamefont
  {Horodecki}, \citenamefont {Horodecki}, \citenamefont {Horodecki},\ and\
  \citenamefont {Horodecki}}]{horodecki2009quantum}%
  \BibitemOpen
  \bibfield  {author} {\bibinfo {author} {\bibfnamefont {R.}~\bibnamefont
  {Horodecki}}, \bibinfo {author} {\bibfnamefont {P.}~\bibnamefont
  {Horodecki}}, \bibinfo {author} {\bibfnamefont {M.}~\bibnamefont
  {Horodecki}},\ and\ \bibinfo {author} {\bibfnamefont {K.}~\bibnamefont
  {Horodecki}},\ }\bibfield  {title} {\bibinfo {title} {Quantum entanglement},\
  }\href {https://doi.org/10.1103/RevModPhys.81.865} {\bibfield  {journal}
  {\bibinfo  {journal} {Rev. Mod. Phys.}\ }\textbf {\bibinfo {volume} {81}},\
  \bibinfo {pages} {865} (\bibinfo {year} {2009})}\BibitemShut {NoStop}%
\bibitem [{\citenamefont {Streltsov}\ \emph {et~al.}(2017)\citenamefont
  {Streltsov}, \citenamefont {Adesso},\ and\ \citenamefont
  {Plenio}}]{streltsov2017colloquium}%
  \BibitemOpen
  \bibfield  {author} {\bibinfo {author} {\bibfnamefont {A.}~\bibnamefont
  {Streltsov}}, \bibinfo {author} {\bibfnamefont {G.}~\bibnamefont {Adesso}},\
  and\ \bibinfo {author} {\bibfnamefont {M.~B.}\ \bibnamefont {Plenio}},\
  }\bibfield  {title} {\bibinfo {title} {Colloquium: Quantum coherence as a
  resource},\ }\href {https://doi.org/10.1103/RevModPhys.89.041003} {\bibfield
  {journal} {\bibinfo  {journal} {Rev. Mod. Phys.}\ }\textbf {\bibinfo {volume}
  {89}},\ \bibinfo {pages} {041003} (\bibinfo {year} {2017})}\BibitemShut
  {NoStop}%
\bibitem [{\citenamefont {Tirrito}\ \emph {et~al.}(2023)\citenamefont
  {Tirrito}, \citenamefont {Tarabunga}, \citenamefont {Lami}, \citenamefont
  {Chanda}, \citenamefont {Leone}, \citenamefont {Oliviero}, \citenamefont
  {Dalmonte}, \citenamefont {Collura},\ and\ \citenamefont
  {Hamma}}]{tirrito2023quantifying}%
  \BibitemOpen
  \bibfield  {author} {\bibinfo {author} {\bibfnamefont {E.}~\bibnamefont
  {Tirrito}}, \bibinfo {author} {\bibfnamefont {P.~S.}\ \bibnamefont
  {Tarabunga}}, \bibinfo {author} {\bibfnamefont {G.}~\bibnamefont {Lami}},
  \bibinfo {author} {\bibfnamefont {T.}~\bibnamefont {Chanda}}, \bibinfo
  {author} {\bibfnamefont {L.}~\bibnamefont {Leone}}, \bibinfo {author}
  {\bibfnamefont {S.~F.~E.}\ \bibnamefont {Oliviero}}, \bibinfo {author}
  {\bibfnamefont {M.}~\bibnamefont {Dalmonte}}, \bibinfo {author}
  {\bibfnamefont {M.}~\bibnamefont {Collura}},\ and\ \bibinfo {author}
  {\bibfnamefont {A.}~\bibnamefont {Hamma}},\ }\href
  {https://doi.org/10.48550/arXiv.2304.01175} {\bibinfo {title} {Quantifying
  non-stabilizerness through entanglement spectrum flatness}},\ \bibinfo
  {howpublished} {arXiv:2304.01175 [quant-ph]} (\bibinfo {year}
  {2023})\BibitemShut {NoStop}%
\bibitem [{\citenamefont {Wang}\ \emph {et~al.}(2021)\citenamefont {Wang},
  \citenamefont {Li},\ and\ \citenamefont {Wang}}]{wang2021variational}%
  \BibitemOpen
  \bibfield  {author} {\bibinfo {author} {\bibfnamefont {Y.}~\bibnamefont
  {Wang}}, \bibinfo {author} {\bibfnamefont {G.}~\bibnamefont {Li}},\ and\
  \bibinfo {author} {\bibfnamefont {X.}~\bibnamefont {Wang}},\ }\bibfield
  {title} {\bibinfo {title} {Variational quantum {G}ibbs state preparation with
  a truncated {T}aylor series},\ }\href
  {https://doi.org/10.1103/PhysRevApplied.16.054035} {\bibfield  {journal}
  {\bibinfo  {journal} {Phys. Rev. Appl.}\ }\textbf {\bibinfo {volume} {16}},\
  \bibinfo {pages} {054035} (\bibinfo {year} {2021})}\BibitemShut {NoStop}%
\bibitem [{\citenamefont {Giordani}\ \emph {et~al.}(2021)\citenamefont
  {Giordani}, \citenamefont {Esposito}, \citenamefont {Hoch}, \citenamefont
  {Carvacho}, \citenamefont {Brod}, \citenamefont {Galv\~ao}, \citenamefont
  {Spagnolo},\ and\ \citenamefont {Sciarrino}}]{giordani2021witnesses}%
  \BibitemOpen
  \bibfield  {author} {\bibinfo {author} {\bibfnamefont {T.}~\bibnamefont
  {Giordani}}, \bibinfo {author} {\bibfnamefont {C.}~\bibnamefont {Esposito}},
  \bibinfo {author} {\bibfnamefont {F.}~\bibnamefont {Hoch}}, \bibinfo {author}
  {\bibfnamefont {G.}~\bibnamefont {Carvacho}}, \bibinfo {author}
  {\bibfnamefont {D.~J.}\ \bibnamefont {Brod}}, \bibinfo {author}
  {\bibfnamefont {E.~F.}\ \bibnamefont {Galv\~ao}}, \bibinfo {author}
  {\bibfnamefont {N.}~\bibnamefont {Spagnolo}},\ and\ \bibinfo {author}
  {\bibfnamefont {F.}~\bibnamefont {Sciarrino}},\ }\bibfield  {title} {\bibinfo
  {title} {Witnesses of coherence and dimension from multiphoton
  indistinguishability tests},\ }\href
  {https://doi.org/10.1103/PhysRevResearch.3.023031} {\bibfield  {journal}
  {\bibinfo  {journal} {Phys. Rev. Res.}\ }\textbf {\bibinfo {volume} {3}},\
  \bibinfo {pages} {023031} (\bibinfo {year} {2021})}\BibitemShut {NoStop}%
\bibitem [{\citenamefont {Designolle}\ \emph {et~al.}(2021)\citenamefont
  {Designolle}, \citenamefont {Uola}, \citenamefont {Luoma},\ and\
  \citenamefont {Brunner}}]{designolle2021set}%
  \BibitemOpen
  \bibfield  {author} {\bibinfo {author} {\bibfnamefont {S.}~\bibnamefont
  {Designolle}}, \bibinfo {author} {\bibfnamefont {R.}~\bibnamefont {Uola}},
  \bibinfo {author} {\bibfnamefont {K.}~\bibnamefont {Luoma}},\ and\ \bibinfo
  {author} {\bibfnamefont {N.}~\bibnamefont {Brunner}},\ }\bibfield  {title}
  {\bibinfo {title} {Set coherence: {B}asis-independent quantification of
  quantum coherence},\ }\href {https://doi.org/10.1103/PhysRevLett.126.220404}
  {\bibfield  {journal} {\bibinfo  {journal} {Phys. Rev. Lett.}\ }\textbf
  {\bibinfo {volume} {126}},\ \bibinfo {pages} {220404} (\bibinfo {year}
  {2021})}\BibitemShut {NoStop}%
\bibitem [{\citenamefont {Arvidsson-Shukur}\ \emph {et~al.}(2020)\citenamefont
  {Arvidsson-Shukur}, \citenamefont {Halpern}, \citenamefont {Lepage},
  \citenamefont {Lasek}, \citenamefont {Barnes},\ and\ \citenamefont
  {Lloyd}}]{arvidsson_Shukur2020quantum}%
  \BibitemOpen
  \bibfield  {author} {\bibinfo {author} {\bibfnamefont {D.~R.~M.}\
  \bibnamefont {Arvidsson-Shukur}}, \bibinfo {author} {\bibfnamefont {N.~Y.}\
  \bibnamefont {Halpern}}, \bibinfo {author} {\bibfnamefont {H.~V.}\
  \bibnamefont {Lepage}}, \bibinfo {author} {\bibfnamefont {A.~A.}\
  \bibnamefont {Lasek}}, \bibinfo {author} {\bibfnamefont {C.~H.~W.}\
  \bibnamefont {Barnes}},\ and\ \bibinfo {author} {\bibfnamefont
  {S.}~\bibnamefont {Lloyd}},\ }\bibfield  {title} {\bibinfo {title} {Quantum
  advantage in postselected metrology},\ }\href
  {https://doi.org/10.1038/s41467-020-17559-w} {\bibfield  {journal} {\bibinfo
  {journal} {Nat. Comm.}\ }\textbf {\bibinfo {volume} {11}},\ \bibinfo {pages}
  {3775} (\bibinfo {year} {2020})}\BibitemShut {NoStop}%
\bibitem [{\citenamefont {Gonz\'alez~Alonso}\ \emph {et~al.}(2019)\citenamefont
  {Gonz\'alez~Alonso}, \citenamefont {Yunger~Halpern},\ and\ \citenamefont
  {Dressel}}]{gonzalez2019out}%
  \BibitemOpen
  \bibfield  {author} {\bibinfo {author} {\bibfnamefont {J.~R.}\ \bibnamefont
  {Gonz\'alez~Alonso}}, \bibinfo {author} {\bibfnamefont {N.}~\bibnamefont
  {Yunger~Halpern}},\ and\ \bibinfo {author} {\bibfnamefont {J.}~\bibnamefont
  {Dressel}},\ }\bibfield  {title} {\bibinfo {title}
  {Out-of-time-ordered-correlator quasiprobabilities robustly witness
  scrambling},\ }\href {https://doi.org/10.1103/PhysRevLett.122.040404}
  {\bibfield  {journal} {\bibinfo  {journal} {Phys. Rev. Lett.}\ }\textbf
  {\bibinfo {volume} {122}},\ \bibinfo {pages} {040404} (\bibinfo {year}
  {2019})}\BibitemShut {NoStop}%
\bibitem [{\citenamefont {Alonso}\ \emph {et~al.}(2022)\citenamefont {Alonso},
  \citenamefont {Shammah}, \citenamefont {Ahmed}, \citenamefont {Nori},\ and\
  \citenamefont {Dressel}}]{alonso2022diagnosing}%
  \BibitemOpen
  \bibfield  {author} {\bibinfo {author} {\bibfnamefont {J.~R.~G.}\
  \bibnamefont {Alonso}}, \bibinfo {author} {\bibfnamefont {N.}~\bibnamefont
  {Shammah}}, \bibinfo {author} {\bibfnamefont {S.}~\bibnamefont {Ahmed}},
  \bibinfo {author} {\bibfnamefont {F.}~\bibnamefont {Nori}},\ and\ \bibinfo
  {author} {\bibfnamefont {J.}~\bibnamefont {Dressel}},\ }\href
  {https://doi.org/10.48550/arXiv.2201.08175} {\bibinfo {title} {Diagnosing
  quantum chaos with out-of-time-ordered-correlator quasiprobability in the
  kicked-top model}},\ \bibinfo {howpublished} {arXiv:2201.08175 [quant-ph]}
  (\bibinfo {year} {2022})\BibitemShut {NoStop}%
\bibitem [{\citenamefont {Oszmaniec}\ \emph {et~al.}(2021)\citenamefont
  {Oszmaniec}, \citenamefont {Brod},\ and\ \citenamefont
  {Galv{\~a}o}}]{oszmaniec2021measuring}%
  \BibitemOpen
  \bibfield  {author} {\bibinfo {author} {\bibfnamefont {M.}~\bibnamefont
  {Oszmaniec}}, \bibinfo {author} {\bibfnamefont {D.~J.}\ \bibnamefont
  {Brod}},\ and\ \bibinfo {author} {\bibfnamefont {E.~F.}\ \bibnamefont
  {Galv{\~a}o}},\ }\href {https://doi.org/10.48550/arXiv.2109.10006} {\bibinfo
  {title} {Measuring relational information between quantum states, and
  applications}},\ \bibinfo {howpublished} {arXiv:2109.10006 [quant-ph]}
  (\bibinfo {year} {2021})\BibitemShut {NoStop}%
\bibitem [{\citenamefont {Quek}\ \emph {et~al.}(2022)\citenamefont {Quek},
  \citenamefont {Wilde},\ and\ \citenamefont {Kaur}}]{quek2022multivariate}%
  \BibitemOpen
  \bibfield  {author} {\bibinfo {author} {\bibfnamefont {Y.}~\bibnamefont
  {Quek}}, \bibinfo {author} {\bibfnamefont {M.~M.}\ \bibnamefont {Wilde}},\
  and\ \bibinfo {author} {\bibfnamefont {E.}~\bibnamefont {Kaur}},\ }\href
  {https://doi.org/10.48550/arXiv.2206.15405} {\bibinfo {title} {Multivariate
  trace estimation in constant quantum depth}},\ \bibinfo {howpublished}
  {arXiv:2206.15405 [quant-ph]} (\bibinfo {year} {2022})\BibitemShut {NoStop}%
\bibitem [{\citenamefont {Yirka}\ and\ \citenamefont
  {Suba{\c{s}}ı}(2021)}]{yirka2021qubitefficient}%
  \BibitemOpen
  \bibfield  {author} {\bibinfo {author} {\bibfnamefont {J.}~\bibnamefont
  {Yirka}}\ and\ \bibinfo {author} {\bibfnamefont {Y.}~\bibnamefont
  {Suba{\c{s}}ı}},\ }\bibfield  {title} {\bibinfo {title} {{Qubit-efficient
  entanglement spectroscopy using qubit resets}},\ }\href
  {https://doi.org/10.22331/q-2021-09-02-535} {\bibfield  {journal} {\bibinfo
  {journal} {{Quantum}}\ }\textbf {\bibinfo {volume} {5}},\ \bibinfo {pages}
  {535} (\bibinfo {year} {2021})}\BibitemShut {NoStop}%
\bibitem [{\citenamefont {Escofier}(2001)}]{escofier2012galois}%
  \BibitemOpen
  \bibfield  {author} {\bibinfo {author} {\bibfnamefont {J.-P.}\ \bibnamefont
  {Escofier}},\ }\href {https://link.springer.com/book/9780387987651} {\emph
  {\bibinfo {title} {Galois Theory}}}\ (\bibinfo  {publisher}
  {Springer-Verlag},\ \bibinfo {address} {New York, NY},\ \bibinfo {year}
  {2001})\BibitemShut {NoStop}%
\bibitem [{\citenamefont {Dressel}\ \emph {et~al.}(2014)\citenamefont
  {Dressel}, \citenamefont {Malik}, \citenamefont {Miatto}, \citenamefont
  {Jordan},\ and\ \citenamefont {Boyd}}]{dressel2014colloquium}%
  \BibitemOpen
  \bibfield  {author} {\bibinfo {author} {\bibfnamefont {J.}~\bibnamefont
  {Dressel}}, \bibinfo {author} {\bibfnamefont {M.}~\bibnamefont {Malik}},
  \bibinfo {author} {\bibfnamefont {F.~M.}\ \bibnamefont {Miatto}}, \bibinfo
  {author} {\bibfnamefont {A.~N.}\ \bibnamefont {Jordan}},\ and\ \bibinfo
  {author} {\bibfnamefont {R.~W.}\ \bibnamefont {Boyd}},\ }\bibfield  {title}
  {\bibinfo {title} {Colloquium: {U}nderstanding quantum weak values: {B}asics
  and applications},\ }\href {https://doi.org/10.1103/RevModPhys.86.307}
  {\bibfield  {journal} {\bibinfo  {journal} {Rev. Mod. Phys.}\ }\textbf
  {\bibinfo {volume} {86}},\ \bibinfo {pages} {307} (\bibinfo {year}
  {2014})}\BibitemShut {NoStop}%
\bibitem [{\citenamefont {Vaidman}(2017)}]{vaidman2017weak}%
  \BibitemOpen
  \bibfield  {author} {\bibinfo {author} {\bibfnamefont {L.}~\bibnamefont
  {Vaidman}},\ }\bibfield  {title} {\bibinfo {title} {Weak value controversy},\
  }\href {https://doi.org/10.1098/rsta.2016.0395} {\bibfield  {journal}
  {\bibinfo  {journal} {Philos. Trans. R. Soc. A}\ }\textbf {\bibinfo {volume}
  {375}},\ \bibinfo {pages} {20160395} (\bibinfo {year} {2017})}\BibitemShut
  {NoStop}%
\bibitem [{\citenamefont {Pusey}(2014)}]{pusey2014anomalous}%
  \BibitemOpen
  \bibfield  {author} {\bibinfo {author} {\bibfnamefont {M.~F.}\ \bibnamefont
  {Pusey}},\ }\bibfield  {title} {\bibinfo {title} {Anomalous weak values are
  proofs of contextuality},\ }\href
  {https://doi.org/10.1103/PhysRevLett.113.200401} {\bibfield  {journal}
  {\bibinfo  {journal} {Phys. Rev. Lett.}\ }\textbf {\bibinfo {volume} {113}},\
  \bibinfo {pages} {200401} (\bibinfo {year} {2014})}\BibitemShut {NoStop}%
\bibitem [{\citenamefont {Kunjwal}\ \emph {et~al.}(2019)\citenamefont
  {Kunjwal}, \citenamefont {Lostaglio},\ and\ \citenamefont
  {Pusey}}]{kunjwal2019anomalous}%
  \BibitemOpen
  \bibfield  {author} {\bibinfo {author} {\bibfnamefont {R.}~\bibnamefont
  {Kunjwal}}, \bibinfo {author} {\bibfnamefont {M.}~\bibnamefont {Lostaglio}},\
  and\ \bibinfo {author} {\bibfnamefont {M.~F.}\ \bibnamefont {Pusey}},\
  }\bibfield  {title} {\bibinfo {title} {Anomalous weak values and
  contextuality: robustness, tightness, and imaginary parts},\ }\href
  {https://doi.org/10.1103/PhysRevA.100.042116} {\bibfield  {journal} {\bibinfo
   {journal} {Phys. Rev. A}\ }\textbf {\bibinfo {volume} {100}},\ \bibinfo
  {pages} {042116} (\bibinfo {year} {2019})}\BibitemShut {NoStop}%
\bibitem [{\citenamefont {Rebufello}\ \emph {et~al.}(2021)\citenamefont
  {Rebufello}, \citenamefont {Piacentini}, \citenamefont {Avella},
  \citenamefont {Souza}, \citenamefont {Gramegna}, \citenamefont {Dziewior},
  \citenamefont {Cohen}, \citenamefont {Vaidman}, \citenamefont {Degiovanni},\
  and\ \citenamefont {Genovese}}]{rebufello2021anomalous}%
  \BibitemOpen
  \bibfield  {author} {\bibinfo {author} {\bibfnamefont {E.}~\bibnamefont
  {Rebufello}}, \bibinfo {author} {\bibfnamefont {F.}~\bibnamefont
  {Piacentini}}, \bibinfo {author} {\bibfnamefont {A.}~\bibnamefont {Avella}},
  \bibinfo {author} {\bibfnamefont {M.~A.~d.}\ \bibnamefont {Souza}}, \bibinfo
  {author} {\bibfnamefont {M.}~\bibnamefont {Gramegna}}, \bibinfo {author}
  {\bibfnamefont {J.}~\bibnamefont {Dziewior}}, \bibinfo {author}
  {\bibfnamefont {E.}~\bibnamefont {Cohen}}, \bibinfo {author} {\bibfnamefont
  {L.}~\bibnamefont {Vaidman}}, \bibinfo {author} {\bibfnamefont {I.~P.}\
  \bibnamefont {Degiovanni}},\ and\ \bibinfo {author} {\bibfnamefont
  {M.}~\bibnamefont {Genovese}},\ }\bibfield  {title} {\bibinfo {title}
  {Anomalous weak values via a single photon detection},\ }\href
  {https://doi.org/10.1038/s41377-021-00539-0} {\bibfield  {journal} {\bibinfo
  {journal} {Light Sci. Appl.}\ }\textbf {\bibinfo {volume} {10}},\ \bibinfo
  {pages} {106} (\bibinfo {year} {2021})}\BibitemShut {NoStop}%
\bibitem [{\citenamefont {Dixon}\ \emph {et~al.}(2009)\citenamefont {Dixon},
  \citenamefont {Starling}, \citenamefont {Jordan},\ and\ \citenamefont
  {Howell}}]{dixon2009ultrasensitive}%
  \BibitemOpen
  \bibfield  {author} {\bibinfo {author} {\bibfnamefont {P.~B.}\ \bibnamefont
  {Dixon}}, \bibinfo {author} {\bibfnamefont {D.~J.}\ \bibnamefont {Starling}},
  \bibinfo {author} {\bibfnamefont {A.~N.}\ \bibnamefont {Jordan}},\ and\
  \bibinfo {author} {\bibfnamefont {J.~C.}\ \bibnamefont {Howell}},\ }\bibfield
   {title} {\bibinfo {title} {Ultrasensitive beam deflection measurement via
  interferometric weak value amplification},\ }\href
  {https://doi.org/10.1103/PhysRevLett.102.173601} {\bibfield  {journal}
  {\bibinfo  {journal} {Phys. Rev. Lett.}\ }\textbf {\bibinfo {volume} {102}},\
  \bibinfo {pages} {173601} (\bibinfo {year} {2009})}\BibitemShut {NoStop}%
\bibitem [{\citenamefont {Turner}\ \emph {et~al.}(2011)\citenamefont {Turner},
  \citenamefont {Hagedorn}, \citenamefont {Schlamminger},\ and\ \citenamefont
  {Gundlach}}]{turner11}%
  \BibitemOpen
  \bibfield  {author} {\bibinfo {author} {\bibfnamefont {M.~D.}\ \bibnamefont
  {Turner}}, \bibinfo {author} {\bibfnamefont {C.~A.}\ \bibnamefont
  {Hagedorn}}, \bibinfo {author} {\bibfnamefont {S.}~\bibnamefont
  {Schlamminger}},\ and\ \bibinfo {author} {\bibfnamefont {J.~H.}\ \bibnamefont
  {Gundlach}},\ }\bibfield  {title} {\bibinfo {title} {Picoradian deflection
  measurement with an interferometric quasi-autocollimator using weak value
  amplification},\ }\href {https://doi.org/10.1364/OL.36.001479} {\bibfield
  {journal} {\bibinfo  {journal} {Opt. Lett.}\ }\textbf {\bibinfo {volume}
  {36}},\ \bibinfo {pages} {1479} (\bibinfo {year} {2011})}\BibitemShut
  {NoStop}%
\bibitem [{\citenamefont {Susa}\ \emph {et~al.}(2012)\citenamefont {Susa},
  \citenamefont {Shikano},\ and\ \citenamefont {Hosoya}}]{susa2012optimal}%
  \BibitemOpen
  \bibfield  {author} {\bibinfo {author} {\bibfnamefont {Y.}~\bibnamefont
  {Susa}}, \bibinfo {author} {\bibfnamefont {Y.}~\bibnamefont {Shikano}},\ and\
  \bibinfo {author} {\bibfnamefont {A.}~\bibnamefont {Hosoya}},\ }\bibfield
  {title} {\bibinfo {title} {Optimal probe wave function of weak-value
  amplification},\ }\href {https://doi.org/10.1103/PhysRevA.85.052110}
  {\bibfield  {journal} {\bibinfo  {journal} {Phys. Rev. A}\ }\textbf {\bibinfo
  {volume} {85}},\ \bibinfo {pages} {052110} (\bibinfo {year}
  {2012})}\BibitemShut {NoStop}%
\bibitem [{\citenamefont {Dressel}\ \emph {et~al.}(2013)\citenamefont
  {Dressel}, \citenamefont {Lyons}, \citenamefont {Jordan}, \citenamefont
  {Graham},\ and\ \citenamefont {Kwiat}}]{dressel2013strengthening}%
  \BibitemOpen
  \bibfield  {author} {\bibinfo {author} {\bibfnamefont {J.}~\bibnamefont
  {Dressel}}, \bibinfo {author} {\bibfnamefont {K.}~\bibnamefont {Lyons}},
  \bibinfo {author} {\bibfnamefont {A.~N.}\ \bibnamefont {Jordan}}, \bibinfo
  {author} {\bibfnamefont {T.~M.}\ \bibnamefont {Graham}},\ and\ \bibinfo
  {author} {\bibfnamefont {P.~G.}\ \bibnamefont {Kwiat}},\ }\bibfield  {title}
  {\bibinfo {title} {Strengthening weak-value amplification with recycled
  photons},\ }\href {https://doi.org/10.1103/PhysRevA.88.023821} {\bibfield
  {journal} {\bibinfo  {journal} {Phys. Rev. A}\ }\textbf {\bibinfo {volume}
  {88}},\ \bibinfo {pages} {023821} (\bibinfo {year} {2013})}\BibitemShut
  {NoStop}%
\bibitem [{\citenamefont {Jordan}\ \emph {et~al.}(2014)\citenamefont {Jordan},
  \citenamefont {Mart{\'\i}nez-Rinc{\'o}n},\ and\ \citenamefont
  {Howell}}]{jordan2014technical}%
  \BibitemOpen
  \bibfield  {author} {\bibinfo {author} {\bibfnamefont {A.~N.}\ \bibnamefont
  {Jordan}}, \bibinfo {author} {\bibfnamefont {J.}~\bibnamefont
  {Mart{\'\i}nez-Rinc{\'o}n}},\ and\ \bibinfo {author} {\bibfnamefont {J.~C.}\
  \bibnamefont {Howell}},\ }\bibfield  {title} {\bibinfo {title} {Technical
  advantages for weak-value amplification: when less is more},\ }\href
  {https://doi.org/10.1103/PhysRevX.4.011031} {\bibfield  {journal} {\bibinfo
  {journal} {Phys. Rev. X}\ }\textbf {\bibinfo {volume} {4}},\ \bibinfo {pages}
  {011031} (\bibinfo {year} {2014})}\BibitemShut {NoStop}%
\bibitem [{\citenamefont {Pang}\ \emph {et~al.}(2014)\citenamefont {Pang},
  \citenamefont {Dressel},\ and\ \citenamefont {Brun}}]{pang2014entanglement}%
  \BibitemOpen
  \bibfield  {author} {\bibinfo {author} {\bibfnamefont {S.}~\bibnamefont
  {Pang}}, \bibinfo {author} {\bibfnamefont {J.}~\bibnamefont {Dressel}},\ and\
  \bibinfo {author} {\bibfnamefont {T.~A.}\ \bibnamefont {Brun}},\ }\bibfield
  {title} {\bibinfo {title} {Entanglement-assisted weak value amplification},\
  }\href {https://doi.org/10.1103/PhysRevLett.113.030401} {\bibfield  {journal}
  {\bibinfo  {journal} {Phys. Rev. Lett.}\ }\textbf {\bibinfo {volume} {113}},\
  \bibinfo {pages} {030401} (\bibinfo {year} {2014})}\BibitemShut {NoStop}%
\bibitem [{\citenamefont {Alves}\ \emph {et~al.}(2015)\citenamefont {Alves},
  \citenamefont {Escher}, \citenamefont {de~Matos~Filho}, \citenamefont
  {Zagury},\ and\ \citenamefont {Davidovich}}]{alves2015weak}%
  \BibitemOpen
  \bibfield  {author} {\bibinfo {author} {\bibfnamefont {G.~B.}\ \bibnamefont
  {Alves}}, \bibinfo {author} {\bibfnamefont {B.}~\bibnamefont {Escher}},
  \bibinfo {author} {\bibfnamefont {R.~L.}\ \bibnamefont {de~Matos~Filho}},
  \bibinfo {author} {\bibfnamefont {N.}~\bibnamefont {Zagury}},\ and\ \bibinfo
  {author} {\bibfnamefont {L.}~\bibnamefont {Davidovich}},\ }\bibfield  {title}
  {\bibinfo {title} {Weak-value amplification as an optimal metrological
  protocol},\ }\href {https://doi.org/10.1103/PhysRevA.91.062107} {\bibfield
  {journal} {\bibinfo  {journal} {Phys. Rev. A}\ }\textbf {\bibinfo {volume}
  {91}},\ \bibinfo {pages} {062107} (\bibinfo {year} {2015})}\BibitemShut
  {NoStop}%
\bibitem [{\citenamefont {Harris}\ \emph {et~al.}(2017)\citenamefont {Harris},
  \citenamefont {Boyd},\ and\ \citenamefont {Lundeen}}]{harris2017weak}%
  \BibitemOpen
  \bibfield  {author} {\bibinfo {author} {\bibfnamefont {J.}~\bibnamefont
  {Harris}}, \bibinfo {author} {\bibfnamefont {R.~W.}\ \bibnamefont {Boyd}},\
  and\ \bibinfo {author} {\bibfnamefont {J.~S.}\ \bibnamefont {Lundeen}},\
  }\bibfield  {title} {\bibinfo {title} {Weak value amplification can
  outperform conventional measurement in the presence of detector saturation},\
  }\href {https://doi.org/10.1103/PhysRevLett.118.070802} {\bibfield  {journal}
  {\bibinfo  {journal} {Phys. Rev. Lett.}\ }\textbf {\bibinfo {volume} {118}},\
  \bibinfo {pages} {070802} (\bibinfo {year} {2017})}\BibitemShut {NoStop}%
\bibitem [{\citenamefont {Pfender}\ \emph {et~al.}(2019)\citenamefont
  {Pfender}, \citenamefont {Wang}, \citenamefont {Sumiya}, \citenamefont
  {Onoda}, \citenamefont {Yang}, \citenamefont {Dasari}, \citenamefont
  {Neumann}, \citenamefont {Pan}, \citenamefont {Isoya}, \citenamefont {Liu}
  \emph {et~al.}}]{pfender19}%
  \BibitemOpen
  \bibfield  {author} {\bibinfo {author} {\bibfnamefont {M.}~\bibnamefont
  {Pfender}}, \bibinfo {author} {\bibfnamefont {P.}~\bibnamefont {Wang}},
  \bibinfo {author} {\bibfnamefont {H.}~\bibnamefont {Sumiya}}, \bibinfo
  {author} {\bibfnamefont {S.}~\bibnamefont {Onoda}}, \bibinfo {author}
  {\bibfnamefont {W.}~\bibnamefont {Yang}}, \bibinfo {author} {\bibfnamefont
  {D.~B.~R.}\ \bibnamefont {Dasari}}, \bibinfo {author} {\bibfnamefont
  {P.}~\bibnamefont {Neumann}}, \bibinfo {author} {\bibfnamefont {X.-Y.}\
  \bibnamefont {Pan}}, \bibinfo {author} {\bibfnamefont {J.}~\bibnamefont
  {Isoya}}, \bibinfo {author} {\bibfnamefont {R.-B.}\ \bibnamefont {Liu}},
  \emph {et~al.},\ }\bibfield  {title} {\bibinfo {title} {High-resolution
  spectroscopy of single nuclear spins via sequential weak measurements},\
  }\href {https://doi.org/10.1038/s41467-019-08544-z} {\bibfield  {journal}
  {\bibinfo  {journal} {Nat. Commun.}\ }\textbf {\bibinfo {volume} {10}},\
  \bibinfo {pages} {594} (\bibinfo {year} {2019})}\BibitemShut {NoStop}%
\bibitem [{\citenamefont {Cujia}\ \emph {et~al.}(2019)\citenamefont {Cujia},
  \citenamefont {Boss}, \citenamefont {Herb}, \citenamefont {Zopes},\ and\
  \citenamefont {Degen}}]{cujia19}%
  \BibitemOpen
  \bibfield  {author} {\bibinfo {author} {\bibfnamefont {K.~S.}\ \bibnamefont
  {Cujia}}, \bibinfo {author} {\bibfnamefont {J.~M.}\ \bibnamefont {Boss}},
  \bibinfo {author} {\bibfnamefont {K.}~\bibnamefont {Herb}}, \bibinfo {author}
  {\bibfnamefont {J.}~\bibnamefont {Zopes}},\ and\ \bibinfo {author}
  {\bibfnamefont {C.~L.}\ \bibnamefont {Degen}},\ }\bibfield  {title} {\bibinfo
  {title} {Tracking the precession of single nuclear spins by weak
  measurements},\ }\href {https://doi.org/10.1038/s41586-019-1334-9} {\bibfield
   {journal} {\bibinfo  {journal} {Nature}\ }\textbf {\bibinfo {volume}
  {571}},\ \bibinfo {pages} {230} (\bibinfo {year} {2019})}\BibitemShut
  {NoStop}%
\bibitem [{\citenamefont {Fang}\ \emph {et~al.}(2021)\citenamefont {Fang},
  \citenamefont {Tan}, \citenamefont {Li},\ and\ \citenamefont
  {Wu}}]{fang2021weak}%
  \BibitemOpen
  \bibfield  {author} {\bibinfo {author} {\bibfnamefont {S.-Z.}\ \bibnamefont
  {Fang}}, \bibinfo {author} {\bibfnamefont {H.-T.}\ \bibnamefont {Tan}},
  \bibinfo {author} {\bibfnamefont {G.-X.}\ \bibnamefont {Li}},\ and\ \bibinfo
  {author} {\bibfnamefont {Q.-L.}\ \bibnamefont {Wu}},\ }\bibfield  {title}
  {\bibinfo {title} {Weak value amplification for angular velocity
  measurements},\ }\href {https://doi.org/10.1364/AO.420231} {\bibfield
  {journal} {\bibinfo  {journal} {Appl. Opt.}\ }\textbf {\bibinfo {volume}
  {60}},\ \bibinfo {pages} {4335} (\bibinfo {year} {2021})}\BibitemShut
  {NoStop}%
\bibitem [{\citenamefont {Huang}\ \emph {et~al.}(2021)\citenamefont {Huang},
  \citenamefont {Duan},\ and\ \citenamefont {Hu}}]{huang2021amplification}%
  \BibitemOpen
  \bibfield  {author} {\bibinfo {author} {\bibfnamefont {J.-H.}\ \bibnamefont
  {Huang}}, \bibinfo {author} {\bibfnamefont {X.-Y.}\ \bibnamefont {Duan}},\
  and\ \bibinfo {author} {\bibfnamefont {X.-Y.}\ \bibnamefont {Hu}},\
  }\bibfield  {title} {\bibinfo {title} {Amplification of rotation velocity
  using weak measurements in {S}agnac's interferometer},\ }\href
  {https://doi.org/10.1140/epjd/s10053-021-00117-4} {\bibfield  {journal}
  {\bibinfo  {journal} {Eur. Phys. J. D}\ }\textbf {\bibinfo {volume} {75}},\
  \bibinfo {pages} {114} (\bibinfo {year} {2021})}\BibitemShut {NoStop}%
\bibitem [{\citenamefont {Paiva}\ \emph {et~al.}(2022)\citenamefont {Paiva},
  \citenamefont {Lenny},\ and\ \citenamefont {Cohen}}]{paiva2022geometric}%
  \BibitemOpen
  \bibfield  {author} {\bibinfo {author} {\bibfnamefont {I.~L.}\ \bibnamefont
  {Paiva}}, \bibinfo {author} {\bibfnamefont {R.}~\bibnamefont {Lenny}},\ and\
  \bibinfo {author} {\bibfnamefont {E.}~\bibnamefont {Cohen}},\ }\bibfield
  {title} {\bibinfo {title} {Geometric phases and the {S}agnac effect:
  {F}oundational aspects and sensing applications},\ }\href
  {https://doi.org/10.1002/qute.202100121} {\bibfield  {journal} {\bibinfo
  {journal} {Adv. Quantum Tech.}\ }\textbf {\bibinfo {volume} {5}},\ \bibinfo
  {pages} {2100121} (\bibinfo {year} {2022})}\BibitemShut {NoStop}%
\bibitem [{\citenamefont {Thekkadath}\ \emph {et~al.}(2016)\citenamefont
  {Thekkadath}, \citenamefont {Giner}, \citenamefont {Chalich}, \citenamefont
  {Horton}, \citenamefont {Banker},\ and\ \citenamefont
  {Lundeen}}]{thekkadath2016direct}%
  \BibitemOpen
  \bibfield  {author} {\bibinfo {author} {\bibfnamefont {G.~S.}\ \bibnamefont
  {Thekkadath}}, \bibinfo {author} {\bibfnamefont {L.}~\bibnamefont {Giner}},
  \bibinfo {author} {\bibfnamefont {Y.}~\bibnamefont {Chalich}}, \bibinfo
  {author} {\bibfnamefont {M.~J.}\ \bibnamefont {Horton}}, \bibinfo {author}
  {\bibfnamefont {J.}~\bibnamefont {Banker}},\ and\ \bibinfo {author}
  {\bibfnamefont {J.~S.}\ \bibnamefont {Lundeen}},\ }\bibfield  {title}
  {\bibinfo {title} {{Direct Measurement of the Density Matrix of a Quantum
  System}},\ }\href {https://doi.org/10.1103/PhysRevLett.117.120401} {\bibfield
   {journal} {\bibinfo  {journal} {Phys. Rev. Lett.}\ }\textbf {\bibinfo
  {volume} {117}},\ \bibinfo {pages} {120401} (\bibinfo {year}
  {2016})}\BibitemShut {NoStop}%
\bibitem [{\citenamefont {Lundeen}\ and\ \citenamefont
  {Bamber}(2012)}]{lundeen2012procedure}%
  \BibitemOpen
  \bibfield  {author} {\bibinfo {author} {\bibfnamefont {J.~S.}\ \bibnamefont
  {Lundeen}}\ and\ \bibinfo {author} {\bibfnamefont {C.}~\bibnamefont
  {Bamber}},\ }\bibfield  {title} {\bibinfo {title} {{Procedure for Direct
  Measurement of General Quantum States Using Weak Measurement}},\ }\href
  {https://doi.org/10.1103/PhysRevLett.108.070402} {\bibfield  {journal}
  {\bibinfo  {journal} {Phys. Rev. Lett.}\ }\textbf {\bibinfo {volume} {108}},\
  \bibinfo {pages} {070402} (\bibinfo {year} {2012})}\BibitemShut {NoStop}%
\bibitem [{\citenamefont {Lundeen}\ \emph {et~al.}(2011)\citenamefont
  {Lundeen}, \citenamefont {Sutherland}, \citenamefont {Patel}, \citenamefont
  {Stewart},\ and\ \citenamefont {Bamber}}]{lundeen2011direct}%
  \BibitemOpen
  \bibfield  {author} {\bibinfo {author} {\bibfnamefont {J.~S.}\ \bibnamefont
  {Lundeen}}, \bibinfo {author} {\bibfnamefont {B.}~\bibnamefont {Sutherland}},
  \bibinfo {author} {\bibfnamefont {A.}~\bibnamefont {Patel}}, \bibinfo
  {author} {\bibfnamefont {C.}~\bibnamefont {Stewart}},\ and\ \bibinfo {author}
  {\bibfnamefont {C.}~\bibnamefont {Bamber}},\ }\bibfield  {title} {\bibinfo
  {title} {Direct measurement of the quantum wavefunction},\ }\href
  {https://doi.org/10.1038/nature10120} {\bibfield  {journal} {\bibinfo
  {journal} {Nature}\ }\textbf {\bibinfo {volume} {474}},\ \bibinfo {pages}
  {188} (\bibinfo {year} {2011})}\BibitemShut {NoStop}%
\bibitem [{\citenamefont {Shi}\ \emph {et~al.}(2015)\citenamefont {Shi},
  \citenamefont {Mirhosseini}, \citenamefont {Margiewicz}, \citenamefont
  {Malik}, \citenamefont {Rivera}, \citenamefont {Zhu},\ and\ \citenamefont
  {Boyd}}]{shi2015scanfree}%
  \BibitemOpen
  \bibfield  {author} {\bibinfo {author} {\bibfnamefont {Z.}~\bibnamefont
  {Shi}}, \bibinfo {author} {\bibfnamefont {M.}~\bibnamefont {Mirhosseini}},
  \bibinfo {author} {\bibfnamefont {J.}~\bibnamefont {Margiewicz}}, \bibinfo
  {author} {\bibfnamefont {M.}~\bibnamefont {Malik}}, \bibinfo {author}
  {\bibfnamefont {F.}~\bibnamefont {Rivera}}, \bibinfo {author} {\bibfnamefont
  {Z.}~\bibnamefont {Zhu}},\ and\ \bibinfo {author} {\bibfnamefont {R.~W.}\
  \bibnamefont {Boyd}},\ }\bibfield  {title} {\bibinfo {title} {Scan-free
  direct measurement of an extremely high-dimensional photonic state},\ }\href
  {https://doi.org/10.1364/optica.2.000388} {\bibfield  {journal} {\bibinfo
  {journal} {Optica}\ }\textbf {\bibinfo {volume} {2}},\ \bibinfo {pages} {388}
  (\bibinfo {year} {2015})}\BibitemShut {NoStop}%
\bibitem [{\citenamefont {Malik}\ \emph {et~al.}(2014)\citenamefont {Malik},
  \citenamefont {Mirhosseini}, \citenamefont {Lavery}, \citenamefont {Leach},
  \citenamefont {Padgett},\ and\ \citenamefont {Boyd}}]{malik2014direct}%
  \BibitemOpen
  \bibfield  {author} {\bibinfo {author} {\bibfnamefont {M.}~\bibnamefont
  {Malik}}, \bibinfo {author} {\bibfnamefont {M.}~\bibnamefont {Mirhosseini}},
  \bibinfo {author} {\bibfnamefont {M.~P.~J.}\ \bibnamefont {Lavery}}, \bibinfo
  {author} {\bibfnamefont {J.}~\bibnamefont {Leach}}, \bibinfo {author}
  {\bibfnamefont {M.~J.}\ \bibnamefont {Padgett}},\ and\ \bibinfo {author}
  {\bibfnamefont {R.~W.}\ \bibnamefont {Boyd}},\ }\bibfield  {title} {\bibinfo
  {title} {Direct measurement of a 27-dimensional orbital-angular-momentum
  state vector},\ }\href {https://doi.org/10.1038/ncomms4115} {\bibfield
  {journal} {\bibinfo  {journal} {Nat. Commun.}\ }\textbf {\bibinfo {volume}
  {5}},\ \bibinfo {pages} {3115} (\bibinfo {year} {2014})}\BibitemShut
  {NoStop}%
\bibitem [{\citenamefont {Lundeen}\ and\ \citenamefont
  {Resch}(2005)}]{lundeen2005practical}%
  \BibitemOpen
  \bibfield  {author} {\bibinfo {author} {\bibfnamefont {J.}~\bibnamefont
  {Lundeen}}\ and\ \bibinfo {author} {\bibfnamefont {K.}~\bibnamefont
  {Resch}},\ }\bibfield  {title} {\bibinfo {title} {Practical measurement of
  joint weak values and their connection to the annihilation operator},\ }\href
  {https://doi.org/10.1016/j.physleta.2004.11.037} {\bibfield  {journal}
  {\bibinfo  {journal} {Phys. Lett. A}\ }\textbf {\bibinfo {volume} {334}},\
  \bibinfo {pages} {337} (\bibinfo {year} {2005})}\BibitemShut {NoStop}%
\bibitem [{\citenamefont {Lundeen}\ and\ \citenamefont
  {Steinberg}(2009)}]{lundeen2009experimental}%
  \BibitemOpen
  \bibfield  {author} {\bibinfo {author} {\bibfnamefont {J.~S.}\ \bibnamefont
  {Lundeen}}\ and\ \bibinfo {author} {\bibfnamefont {A.~M.}\ \bibnamefont
  {Steinberg}},\ }\bibfield  {title} {\bibinfo {title} {Experimental joint weak
  measurement on a photon pair as a probe of hardy's paradox},\ }\href
  {https://doi.org/10.1103/PhysRevLett.102.020404} {\bibfield  {journal}
  {\bibinfo  {journal} {Phys. Rev. Lett.}\ }\textbf {\bibinfo {volume} {102}},\
  \bibinfo {pages} {020404} (\bibinfo {year} {2009})}\BibitemShut {NoStop}%
\bibitem [{\citenamefont {Martinez-Becerril}\ \emph {et~al.}(2021)\citenamefont
  {Martinez-Becerril}, \citenamefont {Bussi{\`{e}}res}, \citenamefont {Curic},
  \citenamefont {Giner}, \citenamefont {A.~Abrahao},\ and\ \citenamefont
  {Lundeen}}]{martinezBecerril2021theoryexperiment}%
  \BibitemOpen
  \bibfield  {author} {\bibinfo {author} {\bibfnamefont {A.~C.}\ \bibnamefont
  {Martinez-Becerril}}, \bibinfo {author} {\bibfnamefont {G.}~\bibnamefont
  {Bussi{\`{e}}res}}, \bibinfo {author} {\bibfnamefont {D.}~\bibnamefont
  {Curic}}, \bibinfo {author} {\bibfnamefont {L.}~\bibnamefont {Giner}},
  \bibinfo {author} {\bibfnamefont {R.}~\bibnamefont {A.~Abrahao}},\ and\
  \bibinfo {author} {\bibfnamefont {J.~S.}\ \bibnamefont {Lundeen}},\
  }\bibfield  {title} {\bibinfo {title} {Theory and experiment for
  resource-efficient joint weak-measurement},\ }\href
  {https://doi.org/10.22331/q-2021-12-06-599} {\bibfield  {journal} {\bibinfo
  {journal} {{Quantum}}\ }\textbf {\bibinfo {volume} {5}},\ \bibinfo {pages}
  {599} (\bibinfo {year} {2021})}\BibitemShut {NoStop}%
\bibitem [{\citenamefont {Kim}\ \emph {et~al.}(2018)\citenamefont {Kim},
  \citenamefont {Kim}, \citenamefont {Lee}, \citenamefont {Han}, \citenamefont
  {Moon}, \citenamefont {Kim},\ and\ \citenamefont {Cho}}]{kim2018direct}%
  \BibitemOpen
  \bibfield  {author} {\bibinfo {author} {\bibfnamefont {Y.}~\bibnamefont
  {Kim}}, \bibinfo {author} {\bibfnamefont {Y.-S.}\ \bibnamefont {Kim}},
  \bibinfo {author} {\bibfnamefont {S.-Y.}\ \bibnamefont {Lee}}, \bibinfo
  {author} {\bibfnamefont {S.-W.}\ \bibnamefont {Han}}, \bibinfo {author}
  {\bibfnamefont {S.}~\bibnamefont {Moon}}, \bibinfo {author} {\bibfnamefont
  {Y.-H.}\ \bibnamefont {Kim}},\ and\ \bibinfo {author} {\bibfnamefont {Y.-W.}\
  \bibnamefont {Cho}},\ }\bibfield  {title} {\bibinfo {title} {Direct quantum
  process tomography via measuring sequential weak values of incompatible
  observables},\ }\href {https://doi.org/10.1038/s41467-017-02511-2} {\bibfield
   {journal} {\bibinfo  {journal} {Nat. Commun.}\ }\textbf {\bibinfo {volume}
  {9}},\ \bibinfo {pages} {192} (\bibinfo {year} {2018})}\BibitemShut {NoStop}%
\bibitem [{\citenamefont {Johansen}\ and\ \citenamefont
  {Luis}(2004)}]{johansen2004nonclassicality}%
  \BibitemOpen
  \bibfield  {author} {\bibinfo {author} {\bibfnamefont {L.~M.}\ \bibnamefont
  {Johansen}}\ and\ \bibinfo {author} {\bibfnamefont {A.}~\bibnamefont
  {Luis}},\ }\bibfield  {title} {\bibinfo {title} {Nonclassicality in weak
  measurements},\ }\href {https://doi.org/10.1103/PhysRevA.70.052115}
  {\bibfield  {journal} {\bibinfo  {journal} {Phys. Rev. A}\ }\textbf {\bibinfo
  {volume} {70}},\ \bibinfo {pages} {052115} (\bibinfo {year}
  {2004})}\BibitemShut {NoStop}%
\bibitem [{\citenamefont {Johansen}(2007{\natexlab{a}})}]{johansen2007quantum}%
  \BibitemOpen
  \bibfield  {author} {\bibinfo {author} {\bibfnamefont {L.~M.}\ \bibnamefont
  {Johansen}},\ }\bibfield  {title} {\bibinfo {title} {Quantum theory of
  successive projective measurements},\ }\href
  {https://doi.org/10.1103/PhysRevA.76.012119} {\bibfield  {journal} {\bibinfo
  {journal} {Phys. Rev. A}\ }\textbf {\bibinfo {volume} {76}},\ \bibinfo
  {pages} {012119} (\bibinfo {year} {2007}{\natexlab{a}})}\BibitemShut
  {NoStop}%
\bibitem [{\citenamefont {Johansen}\ and\ \citenamefont
  {Mello}(2008)}]{johansen2008quantum}%
  \BibitemOpen
  \bibfield  {author} {\bibinfo {author} {\bibfnamefont {L.~M.}\ \bibnamefont
  {Johansen}}\ and\ \bibinfo {author} {\bibfnamefont {P.~A.}\ \bibnamefont
  {Mello}},\ }\bibfield  {title} {\bibinfo {title} {Quantum mechanics of
  successive measurements with arbitrary meter coupling},\ }\href
  {https://doi.org/10.1016/j.physleta.2008.07.021} {\bibfield  {journal}
  {\bibinfo  {journal} {Phys. Lett. A}\ }\textbf {\bibinfo {volume} {372}},\
  \bibinfo {pages} {5760} (\bibinfo {year} {2008})}\BibitemShut {NoStop}%
\bibitem [{\citenamefont {Bamber}\ and\ \citenamefont
  {Lundeen}(2014)}]{bamber2014observing}%
  \BibitemOpen
  \bibfield  {author} {\bibinfo {author} {\bibfnamefont {C.}~\bibnamefont
  {Bamber}}\ and\ \bibinfo {author} {\bibfnamefont {J.~S.}\ \bibnamefont
  {Lundeen}},\ }\bibfield  {title} {\bibinfo {title} {{Observing Dirac's
  Classical Phase Space Analog to the Quantum State}},\ }\href
  {https://doi.org/10.1103/PhysRevLett.112.070405} {\bibfield  {journal}
  {\bibinfo  {journal} {Phys. Rev. Lett.}\ }\textbf {\bibinfo {volume} {112}},\
  \bibinfo {pages} {070405} (\bibinfo {year} {2014})}\BibitemShut {NoStop}%
\bibitem [{\citenamefont {Thekkadath}\ \emph {et~al.}(2017)\citenamefont
  {Thekkadath}, \citenamefont {Saaltink}, \citenamefont {Giner},\ and\
  \citenamefont {Lundeen}}]{thekkadath2017determining}%
  \BibitemOpen
  \bibfield  {author} {\bibinfo {author} {\bibfnamefont {G.~S.}\ \bibnamefont
  {Thekkadath}}, \bibinfo {author} {\bibfnamefont {R.~Y.}\ \bibnamefont
  {Saaltink}}, \bibinfo {author} {\bibfnamefont {L.}~\bibnamefont {Giner}},\
  and\ \bibinfo {author} {\bibfnamefont {J.~S.}\ \bibnamefont {Lundeen}},\
  }\bibfield  {title} {\bibinfo {title} {{Determining Complementary Properties
  with Quantum Clones}},\ }\href
  {https://doi.org/10.1103/PhysRevLett.119.050405} {\bibfield  {journal}
  {\bibinfo  {journal} {Phys. Rev. Lett.}\ }\textbf {\bibinfo {volume} {119}},\
  \bibinfo {pages} {050405} (\bibinfo {year} {2017})}\BibitemShut {NoStop}%
\bibitem [{\citenamefont {Hofmann}(2012)}]{hofmann2012how}%
  \BibitemOpen
  \bibfield  {author} {\bibinfo {author} {\bibfnamefont {H.~F.}\ \bibnamefont
  {Hofmann}},\ }\bibfield  {title} {\bibinfo {title} {{How Weak Values Emerge
  in Joint Measurements on Cloned Quantum Systems}},\ }\href
  {https://doi.org/10.1103/PhysRevLett.109.020408} {\bibfield  {journal}
  {\bibinfo  {journal} {Phys. Rev. Lett.}\ }\textbf {\bibinfo {volume} {109}},\
  \bibinfo {pages} {020408} (\bibinfo {year} {2012})}\BibitemShut {NoStop}%
\bibitem [{\citenamefont {Salvail}\ \emph {et~al.}(2013)\citenamefont
  {Salvail}, \citenamefont {Agnew}, \citenamefont {Johnson}, \citenamefont
  {Bolduc}, \citenamefont {Leach},\ and\ \citenamefont
  {Boyd}}]{salvail2013full}%
  \BibitemOpen
  \bibfield  {author} {\bibinfo {author} {\bibfnamefont {J.~Z.}\ \bibnamefont
  {Salvail}}, \bibinfo {author} {\bibfnamefont {M.}~\bibnamefont {Agnew}},
  \bibinfo {author} {\bibfnamefont {A.~S.}\ \bibnamefont {Johnson}}, \bibinfo
  {author} {\bibfnamefont {E.}~\bibnamefont {Bolduc}}, \bibinfo {author}
  {\bibfnamefont {J.}~\bibnamefont {Leach}},\ and\ \bibinfo {author}
  {\bibfnamefont {R.~W.}\ \bibnamefont {Boyd}},\ }\bibfield  {title} {\bibinfo
  {title} {Full characterization of polarization states of light via direct
  measurement},\ }\href {https://doi.org/10.1038/nphoton.2013.24} {\bibfield
  {journal} {\bibinfo  {journal} {Nat. Photonics}\ }\textbf {\bibinfo {volume}
  {7}},\ \bibinfo {pages} {316} (\bibinfo {year} {2013})}\BibitemShut {NoStop}%
\bibitem [{\citenamefont {Hernández-Gómez}\ \emph {et~al.}(2023)\citenamefont
  {Hernández-Gómez}, \citenamefont {Gherardini}, \citenamefont {Belenchia},
  \citenamefont {Lostaglio}, \citenamefont {Levy},\ and\ \citenamefont
  {Fabbri}}]{hernandezgomez2023projective}%
  \BibitemOpen
  \bibfield  {author} {\bibinfo {author} {\bibfnamefont {S.}~\bibnamefont
  {Hernández-Gómez}}, \bibinfo {author} {\bibfnamefont {S.}~\bibnamefont
  {Gherardini}}, \bibinfo {author} {\bibfnamefont {A.}~\bibnamefont
  {Belenchia}}, \bibinfo {author} {\bibfnamefont {M.}~\bibnamefont
  {Lostaglio}}, \bibinfo {author} {\bibfnamefont {A.}~\bibnamefont {Levy}},\
  and\ \bibinfo {author} {\bibfnamefont {N.}~\bibnamefont {Fabbri}},\ }\href
  {https://doi.org/10.48550/arXiv.2207.12960} {\bibinfo {title} {Projective
  measurements can probe non-classical work extraction and
  time-correlations}},\ \bibinfo {howpublished} {arXiv:2207.12960 [quant-ph]}
  (\bibinfo {year} {2023})\BibitemShut {NoStop}%
\bibitem [{\citenamefont {Budiyono}\ and\ \citenamefont
  {Dipojono}(2023)}]{budiyono2023quantifying}%
  \BibitemOpen
  \bibfield  {author} {\bibinfo {author} {\bibfnamefont {A.}~\bibnamefont
  {Budiyono}}\ and\ \bibinfo {author} {\bibfnamefont {H.~K.}\ \bibnamefont
  {Dipojono}},\ }\bibfield  {title} {\bibinfo {title} {{Quantifying quantum
  coherence via Kirkwood-Dirac quasiprobability}},\ }\href
  {https://doi.org/10.1103/PhysRevA.107.022408} {\bibfield  {journal} {\bibinfo
   {journal} {Phys. Rev. A}\ }\textbf {\bibinfo {volume} {107}},\ \bibinfo
  {pages} {022408} (\bibinfo {year} {2023})}\BibitemShut {NoStop}%
\bibitem [{\citenamefont {Cram{\'e}r}(1946)}]{cramer1999mathematical}%
  \BibitemOpen
  \bibfield  {author} {\bibinfo {author} {\bibfnamefont {H.}~\bibnamefont
  {Cram{\'e}r}},\ }\href
  {https://press.princeton.edu/books/paperback/9780691005478/mathematical-methods-of-statistics-pms-9-volume-9}
  {\emph {\bibinfo {title} {Mathematical methods of statistics}}}\ (\bibinfo
  {publisher} {Princeton University Press},\ \bibinfo {address} {Princeton,
  NJ},\ \bibinfo {year} {1946})\BibitemShut {NoStop}%
\bibitem [{\citenamefont {Rao}(1945)}]{rao1945information}%
  \BibitemOpen
  \bibfield  {author} {\bibinfo {author} {\bibfnamefont {C.~R.}\ \bibnamefont
  {Rao}},\ }\bibfield  {title} {\bibinfo {title} {Information and the accuracy
  attainable in the estimation of statistical parameters},\ }\href
  {https://doi.org/10.1007/978-1-4612-0919-5_16} {\bibfield  {journal}
  {\bibinfo  {journal} {Bull. Calcutta Math. Soc.}\ }\textbf {\bibinfo {volume}
  {37}},\ \bibinfo {pages} {81} (\bibinfo {year} {1945})}\BibitemShut {NoStop}%
\bibitem [{\citenamefont {Ly}\ \emph {et~al.}(2017)\citenamefont {Ly},
  \citenamefont {Marsman}, \citenamefont {Verhagen}, \citenamefont {Grasman},\
  and\ \citenamefont {Wagenmakers}}]{ly2017tutorial}%
  \BibitemOpen
  \bibfield  {author} {\bibinfo {author} {\bibfnamefont {A.}~\bibnamefont
  {Ly}}, \bibinfo {author} {\bibfnamefont {M.}~\bibnamefont {Marsman}},
  \bibinfo {author} {\bibfnamefont {J.}~\bibnamefont {Verhagen}}, \bibinfo
  {author} {\bibfnamefont {R.}~\bibnamefont {Grasman}},\ and\ \bibinfo {author}
  {\bibfnamefont {E.-J.}\ \bibnamefont {Wagenmakers}},\ }\bibfield  {title}
  {\bibinfo {title} {A tutorial on {F}isher information},\ }\href
  {https://doi.org/10.1016/j.jmp.2017.05.006} {\bibfield  {journal} {\bibinfo
  {journal} {J. Math. Psychol.}\ }\textbf {\bibinfo {volume} {80}},\ \bibinfo
  {pages} {40} (\bibinfo {year} {2017})}\BibitemShut {NoStop}%
\bibitem [{\citenamefont {Liu}\ \emph {et~al.}(2019)\citenamefont {Liu},
  \citenamefont {Yuan}, \citenamefont {Lu},\ and\ \citenamefont
  {Wang}}]{liu2019quantum}%
  \BibitemOpen
  \bibfield  {author} {\bibinfo {author} {\bibfnamefont {J.}~\bibnamefont
  {Liu}}, \bibinfo {author} {\bibfnamefont {H.}~\bibnamefont {Yuan}}, \bibinfo
  {author} {\bibfnamefont {X.-M.}\ \bibnamefont {Lu}},\ and\ \bibinfo {author}
  {\bibfnamefont {X.}~\bibnamefont {Wang}},\ }\bibfield  {title} {\bibinfo
  {title} {Quantum {F}isher information matrix and multiparameter estimation},\
  }\href {https://doi.org/10.1088/1751-8121/ab5d4d} {\bibfield  {journal}
  {\bibinfo  {journal} {J. Phys. A}\ }\textbf {\bibinfo {volume} {53}},\
  \bibinfo {pages} {023001} (\bibinfo {year} {2019})}\BibitemShut {NoStop}%
\bibitem [{\citenamefont {Paris}(2009)}]{paris2009quantum}%
  \BibitemOpen
  \bibfield  {author} {\bibinfo {author} {\bibfnamefont {M.~G.~A.}\
  \bibnamefont {Paris}},\ }\bibfield  {title} {\bibinfo {title} {Quantum
  estimation for quantum technology},\ }\href
  {https://doi.org/10.1142/S0219749909004839} {\bibfield  {journal} {\bibinfo
  {journal} {Int. J. Quantum Inf.}\ }\textbf {\bibinfo {volume} {7}},\ \bibinfo
  {pages} {125} (\bibinfo {year} {2009})}\BibitemShut {NoStop}%
\bibitem [{\citenamefont {Preskill}(2018)}]{Preskill2018}%
  \BibitemOpen
  \bibfield  {author} {\bibinfo {author} {\bibfnamefont {J.}~\bibnamefont
  {Preskill}},\ }\bibfield  {title} {\bibinfo {title} {Quantum computing in the
  {NISQ} era and beyond},\ }\href {https://doi.org/10.22331/q-2018-08-06-79}
  {\bibfield  {journal} {\bibinfo  {journal} {Quantum}\ }\textbf {\bibinfo
  {volume} {2}},\ \bibinfo {pages} {79} (\bibinfo {year} {2018})}\BibitemShut
  {NoStop}%
\bibitem [{\citenamefont {Bharti}\ \emph {et~al.}(2022)\citenamefont {Bharti},
  \citenamefont {Cervera-Lierta}, \citenamefont {Kyaw}, \citenamefont {Haug},
  \citenamefont {Alperin-Lea}, \citenamefont {Anand}, \citenamefont {Degroote},
  \citenamefont {Heimonen}, \citenamefont {Kottmann}, \citenamefont {Menke},
  \citenamefont {Mok}, \citenamefont {Sim}, \citenamefont {Kwek},\ and\
  \citenamefont {Aspuru-Guzik}}]{Bharti2022nisq}%
  \BibitemOpen
  \bibfield  {author} {\bibinfo {author} {\bibfnamefont {K.}~\bibnamefont
  {Bharti}}, \bibinfo {author} {\bibfnamefont {A.}~\bibnamefont
  {Cervera-Lierta}}, \bibinfo {author} {\bibfnamefont {T.~H.}\ \bibnamefont
  {Kyaw}}, \bibinfo {author} {\bibfnamefont {T.}~\bibnamefont {Haug}}, \bibinfo
  {author} {\bibfnamefont {S.}~\bibnamefont {Alperin-Lea}}, \bibinfo {author}
  {\bibfnamefont {A.}~\bibnamefont {Anand}}, \bibinfo {author} {\bibfnamefont
  {M.}~\bibnamefont {Degroote}}, \bibinfo {author} {\bibfnamefont
  {H.}~\bibnamefont {Heimonen}}, \bibinfo {author} {\bibfnamefont {J.~S.}\
  \bibnamefont {Kottmann}}, \bibinfo {author} {\bibfnamefont {T.}~\bibnamefont
  {Menke}}, \bibinfo {author} {\bibfnamefont {W.-K.}\ \bibnamefont {Mok}},
  \bibinfo {author} {\bibfnamefont {S.}~\bibnamefont {Sim}}, \bibinfo {author}
  {\bibfnamefont {L.-C.}\ \bibnamefont {Kwek}},\ and\ \bibinfo {author}
  {\bibfnamefont {A.}~\bibnamefont {Aspuru-Guzik}},\ }\bibfield  {title}
  {\bibinfo {title} {Noisy intermediate-scale quantum algorithms},\ }\href
  {https://doi.org/10.1103/RevModPhys.94.015004} {\bibfield  {journal}
  {\bibinfo  {journal} {Rev. Mod. Phys.}\ }\textbf {\bibinfo {volume} {94}},\
  \bibinfo {pages} {015004} (\bibinfo {year} {2022})}\BibitemShut {NoStop}%
\bibitem [{\citenamefont {Bharti}(2021)}]{bharti2021fisher}%
  \BibitemOpen
  \bibfield  {author} {\bibinfo {author} {\bibfnamefont {K.}~\bibnamefont
  {Bharti}},\ }\bibfield  {title} {\bibinfo {title} {Fisher information: A
  crucial tool for {NISQ} research},\ }\href
  {https://doi.org/10.22331/qv-2021-10-06-61} {\bibfield  {journal} {\bibinfo
  {journal} {Quantum Views}\ }\textbf {\bibinfo {volume} {5}},\ \bibinfo
  {pages} {61} (\bibinfo {year} {2021})}\BibitemShut {NoStop}%
\bibitem [{\citenamefont {Meyer}(2021)}]{meyer2021fisher}%
  \BibitemOpen
  \bibfield  {author} {\bibinfo {author} {\bibfnamefont {J.~J.}\ \bibnamefont
  {Meyer}},\ }\bibfield  {title} {\bibinfo {title} {Fisher information in noisy
  intermediate-scale quantum applications},\ }\href
  {https://doi.org/10.22331/q-2021-09-09-539} {\bibfield  {journal} {\bibinfo
  {journal} {Quantum}\ }\textbf {\bibinfo {volume} {5}},\ \bibinfo {pages}
  {539} (\bibinfo {year} {2021})}\BibitemShut {NoStop}%
\bibitem [{\citenamefont {Zhang}\ \emph {et~al.}(2023)\citenamefont {Zhang},
  \citenamefont {Lu}, \citenamefont {Liu}, \citenamefont {Ding},\ and\
  \citenamefont {Wang}}]{zhang2022direct}%
  \BibitemOpen
  \bibfield  {author} {\bibinfo {author} {\bibfnamefont {X.}~\bibnamefont
  {Zhang}}, \bibinfo {author} {\bibfnamefont {X.-M.}\ \bibnamefont {Lu}},
  \bibinfo {author} {\bibfnamefont {J.}~\bibnamefont {Liu}}, \bibinfo {author}
  {\bibfnamefont {W.}~\bibnamefont {Ding}},\ and\ \bibinfo {author}
  {\bibfnamefont {X.}~\bibnamefont {Wang}},\ }\bibfield  {title} {\bibinfo
  {title} {Direct measurement of quantum {F}isher information},\ }\href
  {https://doi.org/10.1103/PhysRevA.107.012414} {\bibfield  {journal} {\bibinfo
   {journal} {Phys. Rev. A}\ }\textbf {\bibinfo {volume} {107}},\ \bibinfo
  {pages} {012414} (\bibinfo {year} {2023})}\BibitemShut {NoStop}%
\bibitem [{\citenamefont {Salvati}\ \emph {et~al.}(2023)\citenamefont
  {Salvati}, \citenamefont {Salmon}, \citenamefont {Barnes},\ and\
  \citenamefont {Arvidsson-Shukur}}]{salvati2023compression}%
  \BibitemOpen
  \bibfield  {author} {\bibinfo {author} {\bibfnamefont {F.}~\bibnamefont
  {Salvati}}, \bibinfo {author} {\bibfnamefont {W.}~\bibnamefont {Salmon}},
  \bibinfo {author} {\bibfnamefont {C.~H.~W.}\ \bibnamefont {Barnes}},\ and\
  \bibinfo {author} {\bibfnamefont {D.~R.~M.}\ \bibnamefont
  {Arvidsson-Shukur}},\ }\href@noop {} {\bibinfo {title} {Compression of
  metrological quantum information in the presence of noise}} (\bibinfo {year}
  {2023}),\ \Eprint {https://arxiv.org/abs/2307.08648} {arXiv:2307.08648
  [quant-ph]} \BibitemShut {NoStop}%
\bibitem [{\citenamefont {Jenne}\ and\ \citenamefont
  {Arvidsson-Shukur}(2022)}]{jenne2022unbounded}%
  \BibitemOpen
  \bibfield  {author} {\bibinfo {author} {\bibfnamefont {J.~H.}\ \bibnamefont
  {Jenne}}\ and\ \bibinfo {author} {\bibfnamefont {D.~R.~M.}\ \bibnamefont
  {Arvidsson-Shukur}},\ }\bibfield  {title} {\bibinfo {title} {Unbounded and
  lossless compression of multiparameter quantum information},\ }\href
  {https://doi.org/10.1103/PhysRevA.106.042404} {\bibfield  {journal} {\bibinfo
   {journal} {Phys. Rev. A}\ }\textbf {\bibinfo {volume} {106}},\ \bibinfo
  {pages} {042404} (\bibinfo {year} {2022})}\BibitemShut {NoStop}%
\bibitem [{\citenamefont {Lupu-Gladstein}\ \emph {et~al.}(2022)\citenamefont
  {Lupu-Gladstein}, \citenamefont {Yilmaz}, \citenamefont {Arvidsson-Shukur},
  \citenamefont {Brodutch}, \citenamefont {Pang}, \citenamefont {Steinberg},\
  and\ \citenamefont {Halpern}}]{lupu_Gladstein2022negative}%
  \BibitemOpen
  \bibfield  {author} {\bibinfo {author} {\bibfnamefont {N.}~\bibnamefont
  {Lupu-Gladstein}}, \bibinfo {author} {\bibfnamefont {Y.~B.}\ \bibnamefont
  {Yilmaz}}, \bibinfo {author} {\bibfnamefont {D.~R.}\ \bibnamefont
  {Arvidsson-Shukur}}, \bibinfo {author} {\bibfnamefont {A.}~\bibnamefont
  {Brodutch}}, \bibinfo {author} {\bibfnamefont {A.~O.}\ \bibnamefont {Pang}},
  \bibinfo {author} {\bibfnamefont {A.~M.}\ \bibnamefont {Steinberg}},\ and\
  \bibinfo {author} {\bibfnamefont {N.~Y.}\ \bibnamefont {Halpern}},\
  }\bibfield  {title} {\bibinfo {title} {Negative quasiprobabilities enhance
  phase estimation in quantum-optics experiment},\ }\href
  {https://doi.org/10.1103/physrevlett.128.220504} {\bibfield  {journal}
  {\bibinfo  {journal} {Phys. Rev. Lett.}\ }\textbf {\bibinfo {volume} {128}},\
  \bibinfo {pages} {220504} (\bibinfo {year} {2022})}\BibitemShut {NoStop}%
\bibitem [{\citenamefont {Das}\ \emph {et~al.}(2023)\citenamefont {Das},
  \citenamefont {Modak},\ and\ \citenamefont {Bera}}]{das2023saturating}%
  \BibitemOpen
  \bibfield  {author} {\bibinfo {author} {\bibfnamefont {S.}~\bibnamefont
  {Das}}, \bibinfo {author} {\bibfnamefont {S.}~\bibnamefont {Modak}},\ and\
  \bibinfo {author} {\bibfnamefont {M.~N.}\ \bibnamefont {Bera}},\ }\bibfield
  {title} {\bibinfo {title} {Saturating quantum advantages in postselected
  metrology with the positive {K}irkwood-{D}irac distribution},\ }\href
  {https://doi.org/10.1103/PhysRevA.107.042413} {\bibfield  {journal} {\bibinfo
   {journal} {Phys. Rev. A}\ }\textbf {\bibinfo {volume} {107}},\ \bibinfo
  {pages} {042413} (\bibinfo {year} {2023})}\BibitemShut {NoStop}%
\bibitem [{\citenamefont {Bargmann}(1964)}]{bargmann1964note}%
  \BibitemOpen
  \bibfield  {author} {\bibinfo {author} {\bibfnamefont {V.}~\bibnamefont
  {Bargmann}},\ }\bibfield  {title} {\bibinfo {title} {Note on {W}igner's
  theorem on symmetry operations},\ }\href {https://doi.org/10.1063/1.1704188}
  {\bibfield  {journal} {\bibinfo  {journal} {J. Math. Phys.}\ }\textbf
  {\bibinfo {volume} {5}},\ \bibinfo {pages} {862} (\bibinfo {year}
  {1964})}\BibitemShut {NoStop}%
\bibitem [{\citenamefont {Jones}\ \emph {et~al.}(2022)\citenamefont {Jones},
  \citenamefont {Kumar}, \citenamefont {D'Aurelio}, \citenamefont {Bayerbach},
  \citenamefont {Menssen},\ and\ \citenamefont
  {Barz}}]{jones2022distinguishability}%
  \BibitemOpen
  \bibfield  {author} {\bibinfo {author} {\bibfnamefont {A.~E.}\ \bibnamefont
  {Jones}}, \bibinfo {author} {\bibfnamefont {S.}~\bibnamefont {Kumar}},
  \bibinfo {author} {\bibfnamefont {S.}~\bibnamefont {D'Aurelio}}, \bibinfo
  {author} {\bibfnamefont {M.}~\bibnamefont {Bayerbach}}, \bibinfo {author}
  {\bibfnamefont {A.~J.}\ \bibnamefont {Menssen}},\ and\ \bibinfo {author}
  {\bibfnamefont {S.}~\bibnamefont {Barz}},\ }\href
  {https://doi.org/10.48550/arXiv.2201.04655} {\bibinfo {title}
  {Distinguishability and mixedness in quantum interference}},\ \bibinfo
  {howpublished} {arXiv:2201.04655 [quant-ph]} (\bibinfo {year}
  {2022})\BibitemShut {NoStop}%
\bibitem [{\citenamefont {Wigderson}(2019)}]{wigderson2019mathematics}%
  \BibitemOpen
  \bibfield  {author} {\bibinfo {author} {\bibfnamefont {A.}~\bibnamefont
  {Wigderson}},\ }\bibfield  {title} {\bibinfo {title} {Mathematics and
  computation},\ }in\ \href
  {https://press.princeton.edu/books/hardcover/9780691189130/mathematics-and-computation}
  {\emph {\bibinfo {booktitle} {Mathematics and Computation}}}\ (\bibinfo
  {publisher} {Princeton University Press},\ \bibinfo {address} {Princeton,
  NJ},\ \bibinfo {year} {2019})\BibitemShut {NoStop}%
\bibitem [{\citenamefont {Garcia-Escartin}\ and\ \citenamefont
  {Chamorro-Posada}(2013)}]{garciaescartin2013aSWAPgate}%
  \BibitemOpen
  \bibfield  {author} {\bibinfo {author} {\bibfnamefont {J.~C.}\ \bibnamefont
  {Garcia-Escartin}}\ and\ \bibinfo {author} {\bibfnamefont {P.}~\bibnamefont
  {Chamorro-Posada}},\ }\bibfield  {title} {\bibinfo {title} {A {SWAP} gate for
  qudits},\ }\href {https://doi.org/10.1007/s11128-013-0621-x} {\bibfield
  {journal} {\bibinfo  {journal} {Quantum Inf. Process.}\ }\textbf {\bibinfo
  {volume} {12}},\ \bibinfo {pages} {3625} (\bibinfo {year}
  {2013})}\BibitemShut {NoStop}%
\bibitem [{\citenamefont {Horodecki}(2003)}]{horodecki2003fromlimits}%
  \BibitemOpen
  \bibfield  {author} {\bibinfo {author} {\bibfnamefont {P.}~\bibnamefont
  {Horodecki}},\ }\bibfield  {title} {\bibinfo {title} {From limits of quantum
  operations to multicopy entanglement witnesses and state-spectrum
  estimation},\ }\href {https://doi.org/10.1103/physreva.68.052101} {\bibfield
  {journal} {\bibinfo  {journal} {Phys. Rev. A}\ }\textbf {\bibinfo {volume}
  {68}},\ \bibinfo {pages} {052101} (\bibinfo {year} {2003})}\BibitemShut
  {NoStop}%
\bibitem [{\citenamefont {Vedral}\ \emph {et~al.}(1997)\citenamefont {Vedral},
  \citenamefont {Plenio}, \citenamefont {Rippin},\ and\ \citenamefont
  {Knight}}]{vedral1997quantifying}%
  \BibitemOpen
  \bibfield  {author} {\bibinfo {author} {\bibfnamefont {V.}~\bibnamefont
  {Vedral}}, \bibinfo {author} {\bibfnamefont {M.~B.}\ \bibnamefont {Plenio}},
  \bibinfo {author} {\bibfnamefont {M.~A.}\ \bibnamefont {Rippin}},\ and\
  \bibinfo {author} {\bibfnamefont {P.~L.}\ \bibnamefont {Knight}},\ }\bibfield
   {title} {\bibinfo {title} {Quantifying entanglement},\ }\href
  {https://doi.org/10.1103/PhysRevLett.78.2275} {\bibfield  {journal} {\bibinfo
   {journal} {Phys. Rev. Lett.}\ }\textbf {\bibinfo {volume} {78}},\ \bibinfo
  {pages} {2275} (\bibinfo {year} {1997})}\BibitemShut {NoStop}%
\bibitem [{\citenamefont {Bovino}\ \emph {et~al.}(2005)\citenamefont {Bovino},
  \citenamefont {Castagnoli}, \citenamefont {Ekert}, \citenamefont {Horodecki},
  \citenamefont {Alves},\ and\ \citenamefont {Sergienko}}]{bovino2005direct}%
  \BibitemOpen
  \bibfield  {author} {\bibinfo {author} {\bibfnamefont {F.~A.}\ \bibnamefont
  {Bovino}}, \bibinfo {author} {\bibfnamefont {G.}~\bibnamefont {Castagnoli}},
  \bibinfo {author} {\bibfnamefont {A.}~\bibnamefont {Ekert}}, \bibinfo
  {author} {\bibfnamefont {P.}~\bibnamefont {Horodecki}}, \bibinfo {author}
  {\bibfnamefont {C.~M.}\ \bibnamefont {Alves}},\ and\ \bibinfo {author}
  {\bibfnamefont {A.~V.}\ \bibnamefont {Sergienko}},\ }\bibfield  {title}
  {\bibinfo {title} {{Direct Measurement of Nonlinear Properties of Bipartite
  Quantum States}},\ }\href {https://doi.org/10.1103/PhysRevLett.95.240407}
  {\bibfield  {journal} {\bibinfo  {journal} {Phys. Rev. Lett.}\ }\textbf
  {\bibinfo {volume} {95}},\ \bibinfo {pages} {240407} (\bibinfo {year}
  {2005})}\BibitemShut {NoStop}%
\bibitem [{\citenamefont {Daley}\ \emph {et~al.}(2012)\citenamefont {Daley},
  \citenamefont {Pichler}, \citenamefont {Schachenmayer},\ and\ \citenamefont
  {Zoller}}]{daley2012measuring}%
  \BibitemOpen
  \bibfield  {author} {\bibinfo {author} {\bibfnamefont {A.~J.}\ \bibnamefont
  {Daley}}, \bibinfo {author} {\bibfnamefont {H.}~\bibnamefont {Pichler}},
  \bibinfo {author} {\bibfnamefont {J.}~\bibnamefont {Schachenmayer}},\ and\
  \bibinfo {author} {\bibfnamefont {P.}~\bibnamefont {Zoller}},\ }\bibfield
  {title} {\bibinfo {title} {Measuring entanglement growth in quench dynamics
  of bosons in an optical lattice},\ }\href
  {https://doi.org/10.1103/PhysRevLett.109.020505} {\bibfield  {journal}
  {\bibinfo  {journal} {Phys. Rev. Lett.}\ }\textbf {\bibinfo {volume} {109}},\
  \bibinfo {pages} {020505} (\bibinfo {year} {2012})}\BibitemShut {NoStop}%
\bibitem [{\citenamefont {Suba{\c{s}}{\i}}\ \emph {et~al.}(2019)\citenamefont
  {Suba{\c{s}}{\i}}, \citenamefont {Cincio},\ and\ \citenamefont
  {Coles}}]{subacsi2019entanglement}%
  \BibitemOpen
  \bibfield  {author} {\bibinfo {author} {\bibfnamefont {Y.}~\bibnamefont
  {Suba{\c{s}}{\i}}}, \bibinfo {author} {\bibfnamefont {L.}~\bibnamefont
  {Cincio}},\ and\ \bibinfo {author} {\bibfnamefont {P.~J.}\ \bibnamefont
  {Coles}},\ }\bibfield  {title} {\bibinfo {title} {Entanglement spectroscopy
  with a depth-two quantum circuit},\ }\href
  {https://doi.org/10.1088/1751-8121/aaf54d} {\bibfield  {journal} {\bibinfo
  {journal} {J. Phys. A}\ }\textbf {\bibinfo {volume} {52}},\ \bibinfo {pages}
  {044001} (\bibinfo {year} {2019})}\BibitemShut {NoStop}%
\bibitem [{\citenamefont {Ekert}\ \emph {et~al.}(2002)\citenamefont {Ekert},
  \citenamefont {Alves}, \citenamefont {Oi}, \citenamefont {Horodecki},
  \citenamefont {Horodecki},\ and\ \citenamefont {Kwek}}]{ekert2002direct}%
  \BibitemOpen
  \bibfield  {author} {\bibinfo {author} {\bibfnamefont {A.~K.}\ \bibnamefont
  {Ekert}}, \bibinfo {author} {\bibfnamefont {C.~M.}\ \bibnamefont {Alves}},
  \bibinfo {author} {\bibfnamefont {D.~K.~L.}\ \bibnamefont {Oi}}, \bibinfo
  {author} {\bibfnamefont {M.}~\bibnamefont {Horodecki}}, \bibinfo {author}
  {\bibfnamefont {P.}~\bibnamefont {Horodecki}},\ and\ \bibinfo {author}
  {\bibfnamefont {L.~C.}\ \bibnamefont {Kwek}},\ }\bibfield  {title} {\bibinfo
  {title} {Direct estimations of linear and nonlinear functionals of a quantum
  state},\ }\href {https://doi.org/10.1103/PhysRevLett.88.217901} {\bibfield
  {journal} {\bibinfo  {journal} {Phys. Rev. Lett.}\ }\textbf {\bibinfo
  {volume} {88}},\ \bibinfo {pages} {217901} (\bibinfo {year}
  {2002})}\BibitemShut {NoStop}%
\bibitem [{\citenamefont {Alves}\ \emph {et~al.}(2003)\citenamefont {Alves},
  \citenamefont {Horodecki}, \citenamefont {Oi}, \citenamefont {Kwek},\ and\
  \citenamefont {Ekert}}]{alves2003direct}%
  \BibitemOpen
  \bibfield  {author} {\bibinfo {author} {\bibfnamefont {C.~M.}\ \bibnamefont
  {Alves}}, \bibinfo {author} {\bibfnamefont {P.}~\bibnamefont {Horodecki}},
  \bibinfo {author} {\bibfnamefont {D.~K.~L.}\ \bibnamefont {Oi}}, \bibinfo
  {author} {\bibfnamefont {L.~C.}\ \bibnamefont {Kwek}},\ and\ \bibinfo
  {author} {\bibfnamefont {A.~K.}\ \bibnamefont {Ekert}},\ }\bibfield  {title}
  {\bibinfo {title} {Direct estimation of functionals of density operators by
  local operations and classical communication},\ }\href
  {https://doi.org/10.1103/physreva.68.032306} {\bibfield  {journal} {\bibinfo
  {journal} {Phys. Rev. A}\ }\textbf {\bibinfo {volume} {68}},\ \bibinfo
  {pages} {032306} (\bibinfo {year} {2003})}\BibitemShut {NoStop}%
\bibitem [{\citenamefont {Horodecki}\ and\ \citenamefont
  {Ekert}(2002)}]{horodecki2002method}%
  \BibitemOpen
  \bibfield  {author} {\bibinfo {author} {\bibfnamefont {P.}~\bibnamefont
  {Horodecki}}\ and\ \bibinfo {author} {\bibfnamefont {A.}~\bibnamefont
  {Ekert}},\ }\bibfield  {title} {\bibinfo {title} {Method for direct detection
  of quantum entanglement},\ }\href
  {https://doi.org/10.1103/physrevlett.89.127902} {\bibfield  {journal}
  {\bibinfo  {journal} {Phys. Rev. Lett.}\ }\textbf {\bibinfo {volume} {89}},\
  \bibinfo {pages} {127902} (\bibinfo {year} {2002})}\BibitemShut {NoStop}%
\bibitem [{\citenamefont {Tanaka}\ \emph {et~al.}(2014)\citenamefont {Tanaka},
  \citenamefont {Ota}, \citenamefont {Kanazawa}, \citenamefont {Kimura},
  \citenamefont {Nakazato},\ and\ \citenamefont
  {Nori}}]{tanaka2014determining}%
  \BibitemOpen
  \bibfield  {author} {\bibinfo {author} {\bibfnamefont {T.}~\bibnamefont
  {Tanaka}}, \bibinfo {author} {\bibfnamefont {Y.}~\bibnamefont {Ota}},
  \bibinfo {author} {\bibfnamefont {M.}~\bibnamefont {Kanazawa}}, \bibinfo
  {author} {\bibfnamefont {G.}~\bibnamefont {Kimura}}, \bibinfo {author}
  {\bibfnamefont {H.}~\bibnamefont {Nakazato}},\ and\ \bibinfo {author}
  {\bibfnamefont {F.}~\bibnamefont {Nori}},\ }\bibfield  {title} {\bibinfo
  {title} {Determining eigenvalues of a density matrix with minimal information
  in a single experimental setting},\ }\href
  {https://doi.org/10.1103/PhysRevA.89.012117} {\bibfield  {journal} {\bibinfo
  {journal} {Phys. Rev. A}\ }\textbf {\bibinfo {volume} {89}},\ \bibinfo
  {pages} {012117} (\bibinfo {year} {2014})}\BibitemShut {NoStop}%
\bibitem [{\citenamefont {van Enk}\ and\ \citenamefont
  {Beenakker}(2012)}]{van2012measuring}%
  \BibitemOpen
  \bibfield  {author} {\bibinfo {author} {\bibfnamefont {S.~J.}\ \bibnamefont
  {van Enk}}\ and\ \bibinfo {author} {\bibfnamefont {C.~W.~J.}\ \bibnamefont
  {Beenakker}},\ }\bibfield  {title} {\bibinfo {title} {Measuring
  $\mathrm{Tr}{\ensuremath{\rho}}^{n}$ on single copies of $\ensuremath{\rho}$
  using random measurements},\ }\href
  {https://doi.org/10.1103/PhysRevLett.108.110503} {\bibfield  {journal}
  {\bibinfo  {journal} {Phys. Rev. Lett.}\ }\textbf {\bibinfo {volume} {108}},\
  \bibinfo {pages} {110503} (\bibinfo {year} {2012})}\BibitemShut {NoStop}%
\bibitem [{\citenamefont {Chien}\ and\ \citenamefont
  {Waldron}(2016)}]{chien2016characterization}%
  \BibitemOpen
  \bibfield  {author} {\bibinfo {author} {\bibfnamefont {T.-Y.}\ \bibnamefont
  {Chien}}\ and\ \bibinfo {author} {\bibfnamefont {S.}~\bibnamefont
  {Waldron}},\ }\bibfield  {title} {\bibinfo {title} {A characterization of
  projective unitary equivalence of finite frames and applications},\ }\href
  {https://doi.org/10.1137/15M1042140} {\bibfield  {journal} {\bibinfo
  {journal} {SIAM J. Discrete Math.}\ }\textbf {\bibinfo {volume} {30}},\
  \bibinfo {pages} {976} (\bibinfo {year} {2016})}\BibitemShut {NoStop}%
\bibitem [{\citenamefont {Galv\~ao}\ and\ \citenamefont
  {Brod}(2020)}]{galvao2020quantum}%
  \BibitemOpen
  \bibfield  {author} {\bibinfo {author} {\bibfnamefont {E.~F.}\ \bibnamefont
  {Galv\~ao}}\ and\ \bibinfo {author} {\bibfnamefont {D.~J.}\ \bibnamefont
  {Brod}},\ }\bibfield  {title} {\bibinfo {title} {Quantum and classical bounds
  for two-state overlaps},\ }\href
  {https://doi.org/10.1103/PhysRevA.101.062110} {\bibfield  {journal} {\bibinfo
   {journal} {Phys. Rev. A}\ }\textbf {\bibinfo {volume} {101}},\ \bibinfo
  {pages} {062110} (\bibinfo {year} {2020})}\BibitemShut {NoStop}%
\bibitem [{\citenamefont {Wagner}\ \emph
  {et~al.}(2022{\natexlab{a}})\citenamefont {Wagner}, \citenamefont {Barbosa},\
  and\ \citenamefont {Galv{\~a}o}}]{wagner2022inequalities}%
  \BibitemOpen
  \bibfield  {author} {\bibinfo {author} {\bibfnamefont {R.}~\bibnamefont
  {Wagner}}, \bibinfo {author} {\bibfnamefont {R.~S.}\ \bibnamefont
  {Barbosa}},\ and\ \bibinfo {author} {\bibfnamefont {E.~F.}\ \bibnamefont
  {Galv{\~a}o}},\ }\href {https://doi.org/10.48550/arXiv.2209.02670} {\bibinfo
  {title} {Inequalities witnessing coherence, nonlocality, and
  contextuality}},\ \bibinfo {howpublished} {arXiv:2209.02670 [quant-ph]}
  (\bibinfo {year} {2022}{\natexlab{a}})\BibitemShut {NoStop}%
\bibitem [{\citenamefont {Wagner}\ \emph
  {et~al.}(2022{\natexlab{b}})\citenamefont {Wagner}, \citenamefont
  {Camillini},\ and\ \citenamefont {Galv{\~a}o}}]{wagner2022coherence}%
  \BibitemOpen
  \bibfield  {author} {\bibinfo {author} {\bibfnamefont {R.}~\bibnamefont
  {Wagner}}, \bibinfo {author} {\bibfnamefont {A.}~\bibnamefont {Camillini}},\
  and\ \bibinfo {author} {\bibfnamefont {E.~F.}\ \bibnamefont {Galv{\~a}o}},\
  }\href {https://doi.org/10.48550/arXiv.2210.05624} {\bibinfo {title}
  {Coherence and contextuality in a {M}ach-{Z}ehnder interferometer}},\
  \bibinfo {howpublished} {arXiv:2210.05624 [quant-ph]} (\bibinfo {year}
  {2022}{\natexlab{b}})\BibitemShut {NoStop}%
\bibitem [{\citenamefont {Le}\ \emph {et~al.}(2023)\citenamefont {Le},
  \citenamefont {Meroni}, \citenamefont {Sturmfels}, \citenamefont {Werner},\
  and\ \citenamefont {Ziegler}}]{le2023quantum}%
  \BibitemOpen
  \bibfield  {author} {\bibinfo {author} {\bibfnamefont {T.~P.}\ \bibnamefont
  {Le}}, \bibinfo {author} {\bibfnamefont {C.}~\bibnamefont {Meroni}}, \bibinfo
  {author} {\bibfnamefont {B.}~\bibnamefont {Sturmfels}}, \bibinfo {author}
  {\bibfnamefont {R.~F.}\ \bibnamefont {Werner}},\ and\ \bibinfo {author}
  {\bibfnamefont {T.}~\bibnamefont {Ziegler}},\ }\bibfield  {title} {\bibinfo
  {title} {Quantum correlations in the minimal scenario},\ }\href
  {https://doi.org/10.22331/q-2023-03-16-947} {\bibfield  {journal} {\bibinfo
  {journal} {Quantum}\ }\textbf {\bibinfo {volume} {7}},\ \bibinfo {pages}
  {947} (\bibinfo {year} {2023})}\BibitemShut {NoStop}%
\bibitem [{\citenamefont {Nielsen}\ and\ \citenamefont
  {Chuang}(2000)}]{nielsen2002quantum}%
  \BibitemOpen
  \bibfield  {author} {\bibinfo {author} {\bibfnamefont {M.~A.}\ \bibnamefont
  {Nielsen}}\ and\ \bibinfo {author} {\bibfnamefont {I.}~\bibnamefont
  {Chuang}},\ }\href {https://doi.org/10.1017/CBO9780511976667} {\emph
  {\bibinfo {title} {Quantum computation and quantum information}}}\ (\bibinfo
  {publisher} {Cambridge University Press},\ \bibinfo {address} {Cambridge,
  UK},\ \bibinfo {year} {2000})\BibitemShut {NoStop}%
\bibitem [{\citenamefont {Flammia}\ \emph {et~al.}(2012)\citenamefont
  {Flammia}, \citenamefont {Gross}, \citenamefont {Liu},\ and\ \citenamefont
  {Eisert}}]{flammia2012quantum}%
  \BibitemOpen
  \bibfield  {author} {\bibinfo {author} {\bibfnamefont {S.~T.}\ \bibnamefont
  {Flammia}}, \bibinfo {author} {\bibfnamefont {D.}~\bibnamefont {Gross}},
  \bibinfo {author} {\bibfnamefont {Y.-K.}\ \bibnamefont {Liu}},\ and\ \bibinfo
  {author} {\bibfnamefont {J.}~\bibnamefont {Eisert}},\ }\bibfield  {title}
  {\bibinfo {title} {Quantum tomography via compressed sensing: error bounds,
  sample complexity and efficient estimators},\ }\href
  {https://doi.org/10.1088/1367-2630/14/9/095022} {\bibfield  {journal}
  {\bibinfo  {journal} {New J. Phys.}\ }\textbf {\bibinfo {volume} {14}},\
  \bibinfo {pages} {095022} (\bibinfo {year} {2012})}\BibitemShut {NoStop}%
\bibitem [{\citenamefont {Hoeffding}(1963)}]{hoeffding1994probability}%
  \BibitemOpen
  \bibfield  {author} {\bibinfo {author} {\bibfnamefont {W.}~\bibnamefont
  {Hoeffding}},\ }\bibfield  {title} {\bibinfo {title} {Probability
  inequalities for sums of bounded random variables},\ }\href
  {https://doi.org/10.2307/2282952} {\bibfield  {journal} {\bibinfo  {journal}
  {J. Am. Stat. Assoc.}\ }\textbf {\bibinfo {volume} {58}},\ \bibinfo {pages}
  {13} (\bibinfo {year} {1963})}\BibitemShut {NoStop}%
\bibitem [{\citenamefont {Shalev-Shwartz}\ and\ \citenamefont
  {Ben-David}(2014)}]{shalev2014understanding}%
  \BibitemOpen
  \bibfield  {author} {\bibinfo {author} {\bibfnamefont {S.}~\bibnamefont
  {Shalev-Shwartz}}\ and\ \bibinfo {author} {\bibfnamefont {S.}~\bibnamefont
  {Ben-David}},\ }\href
  {https://www.cambridge.org/gb/academic/subjects/computer-science/pattern-recognition-and-machine-learning/understanding-machine-learning-theory-algorithms?format=HB&isbn=9781107057135}
  {\emph {\bibinfo {title} {Understanding machine learning: From theory to
  algorithms}}}\ (\bibinfo  {publisher} {Cambridge University Press},\ \bibinfo
  {address} {New York, NY},\ \bibinfo {year} {2014})\BibitemShut {NoStop}%
\bibitem [{\citenamefont {Chen}\ \emph {et~al.}(2022)\citenamefont {Chen},
  \citenamefont {Huang}, \citenamefont {Li}, \citenamefont {Liu},\ and\
  \citenamefont {Sellke}}]{chen2022tight}%
  \BibitemOpen
  \bibfield  {author} {\bibinfo {author} {\bibfnamefont {S.}~\bibnamefont
  {Chen}}, \bibinfo {author} {\bibfnamefont {B.}~\bibnamefont {Huang}},
  \bibinfo {author} {\bibfnamefont {J.}~\bibnamefont {Li}}, \bibinfo {author}
  {\bibfnamefont {A.}~\bibnamefont {Liu}},\ and\ \bibinfo {author}
  {\bibfnamefont {M.}~\bibnamefont {Sellke}},\ }\href
  {https://doi.org/10.48550/arXiv.2206.05265} {\bibinfo {title} {When does
  adaptivity help for quantum state learning?}},\ \bibinfo {howpublished}
  {arXiv:2206.05265 [quant-ph]} (\bibinfo {year} {2022})\BibitemShut {NoStop}%
\bibitem [{\citenamefont
  {Johansen}(2007{\natexlab{b}})}]{johansen2007reconstructing}%
  \BibitemOpen
  \bibfield  {author} {\bibinfo {author} {\bibfnamefont {L.~M.}\ \bibnamefont
  {Johansen}},\ }\bibfield  {title} {\bibinfo {title} {Reconstructing weak
  values without weak measurements},\ }\href
  {https://doi.org/10.1016/j.physleta.2007.02.039} {\bibfield  {journal}
  {\bibinfo  {journal} {Phys. Lett. A}\ }\textbf {\bibinfo {volume} {366}},\
  \bibinfo {pages} {374} (\bibinfo {year} {2007}{\natexlab{b}})}\BibitemShut
  {NoStop}%
\bibitem [{\citenamefont {Vallone}\ and\ \citenamefont
  {Dequal}(2016)}]{vallone2016strong}%
  \BibitemOpen
  \bibfield  {author} {\bibinfo {author} {\bibfnamefont {G.}~\bibnamefont
  {Vallone}}\ and\ \bibinfo {author} {\bibfnamefont {D.}~\bibnamefont
  {Dequal}},\ }\bibfield  {title} {\bibinfo {title} {Strong measurements give a
  better direct measurement of the quantum wave function},\ }\href
  {https://doi.org/10.1103/PhysRevLett.116.040502} {\bibfield  {journal}
  {\bibinfo  {journal} {Phys. Rev. Lett.}\ }\textbf {\bibinfo {volume} {116}},\
  \bibinfo {pages} {040502} (\bibinfo {year} {2016})}\BibitemShut {NoStop}%
\bibitem [{\citenamefont {Brunner}\ and\ \citenamefont
  {Simon}(2010)}]{brunner2010measuring}%
  \BibitemOpen
  \bibfield  {author} {\bibinfo {author} {\bibfnamefont {N.}~\bibnamefont
  {Brunner}}\ and\ \bibinfo {author} {\bibfnamefont {C.}~\bibnamefont
  {Simon}},\ }\bibfield  {title} {\bibinfo {title} {Measuring small
  longitudinal phase shifts: weak measurements or standard interferometry?},\
  }\href {https://doi.org/10.1103/PhysRevLett.105.010405} {\bibfield  {journal}
  {\bibinfo  {journal} {Phys. Rev. Lett.}\ }\textbf {\bibinfo {volume} {105}},\
  \bibinfo {pages} {010405} (\bibinfo {year} {2010})}\BibitemShut {NoStop}%
\bibitem [{\citenamefont {Li}\ \emph {et~al.}(2011)\citenamefont {Li},
  \citenamefont {Xu}, \citenamefont {Tang}, \citenamefont {Xu},\ and\
  \citenamefont {Guo}}]{li2011ultrasensitive}%
  \BibitemOpen
  \bibfield  {author} {\bibinfo {author} {\bibfnamefont {C.-F.}\ \bibnamefont
  {Li}}, \bibinfo {author} {\bibfnamefont {X.-Y.}\ \bibnamefont {Xu}}, \bibinfo
  {author} {\bibfnamefont {J.-S.}\ \bibnamefont {Tang}}, \bibinfo {author}
  {\bibfnamefont {J.-S.}\ \bibnamefont {Xu}},\ and\ \bibinfo {author}
  {\bibfnamefont {G.-C.}\ \bibnamefont {Guo}},\ }\bibfield  {title} {\bibinfo
  {title} {Ultrasensitive phase estimation with white light},\ }\href
  {https://doi.org/10.1103/PhysRevA.83.044102} {\bibfield  {journal} {\bibinfo
  {journal} {Phys. Rev. A}\ }\textbf {\bibinfo {volume} {83}},\ \bibinfo
  {pages} {044102} (\bibinfo {year} {2011})}\BibitemShut {NoStop}%
\bibitem [{\citenamefont {Dressel}\ and\ \citenamefont
  {Jordan}(2012)}]{dressel2012significance}%
  \BibitemOpen
  \bibfield  {author} {\bibinfo {author} {\bibfnamefont {J.}~\bibnamefont
  {Dressel}}\ and\ \bibinfo {author} {\bibfnamefont {A.~N.}\ \bibnamefont
  {Jordan}},\ }\bibfield  {title} {\bibinfo {title} {Significance of the
  imaginary part of the weak value},\ }\href
  {https://doi.org/10.1103/PhysRevA.85.012107} {\bibfield  {journal} {\bibinfo
  {journal} {Phys. Rev. A}\ }\textbf {\bibinfo {volume} {85}},\ \bibinfo
  {pages} {012107} (\bibinfo {year} {2012})}\BibitemShut {NoStop}%
\bibitem [{\citenamefont {Mitarai}\ and\ \citenamefont
  {Fujii}(2019)}]{mitarai2019methodology}%
  \BibitemOpen
  \bibfield  {author} {\bibinfo {author} {\bibfnamefont {K.}~\bibnamefont
  {Mitarai}}\ and\ \bibinfo {author} {\bibfnamefont {K.}~\bibnamefont
  {Fujii}},\ }\bibfield  {title} {\bibinfo {title} {Methodology for replacing
  indirect measurements with direct measurements},\ }\href
  {https://doi.org/10.1103/PhysRevResearch.1.013006} {\bibfield  {journal}
  {\bibinfo  {journal} {Phys. Rev. Res.}\ }\textbf {\bibinfo {volume} {1}},\
  \bibinfo {pages} {013006} (\bibinfo {year} {2019})}\BibitemShut {NoStop}%
\bibitem [{\citenamefont {Mazzola}\ \emph {et~al.}(2013)\citenamefont
  {Mazzola}, \citenamefont {De~Chiara},\ and\ \citenamefont
  {Paternostro}}]{mazzola2013measuring}%
  \BibitemOpen
  \bibfield  {author} {\bibinfo {author} {\bibfnamefont {L.}~\bibnamefont
  {Mazzola}}, \bibinfo {author} {\bibfnamefont {G.}~\bibnamefont {De~Chiara}},\
  and\ \bibinfo {author} {\bibfnamefont {M.}~\bibnamefont {Paternostro}},\
  }\bibfield  {title} {\bibinfo {title} {Measuring the characteristic function
  of the work distribution},\ }\href
  {https://doi.org/10.1103/PhysRevLett.110.230602} {\bibfield  {journal}
  {\bibinfo  {journal} {Phys. Rev. Lett.}\ }\textbf {\bibinfo {volume} {110}},\
  \bibinfo {pages} {230602} (\bibinfo {year} {2013})}\BibitemShut {NoStop}%
\bibitem [{\citenamefont {Buscemi}\ \emph {et~al.}(2013)\citenamefont
  {Buscemi}, \citenamefont {Dall'Arno}, \citenamefont {Ozawa},\ and\
  \citenamefont {Vedral}}]{buscemi2013direct}%
  \BibitemOpen
  \bibfield  {author} {\bibinfo {author} {\bibfnamefont {F.}~\bibnamefont
  {Buscemi}}, \bibinfo {author} {\bibfnamefont {M.}~\bibnamefont {Dall'Arno}},
  \bibinfo {author} {\bibfnamefont {M.}~\bibnamefont {Ozawa}},\ and\ \bibinfo
  {author} {\bibfnamefont {V.}~\bibnamefont {Vedral}},\ }\href
  {https://doi.org/10.48550/arXiv.1312.4240} {\bibinfo {title} {Direct
  observation of any two-point quantum correlation function}},\ \bibinfo
  {howpublished} {arXiv:1312.4240 [quant-ph]} (\bibinfo {year}
  {2013})\BibitemShut {NoStop}%
\bibitem [{\citenamefont {Resch}\ and\ \citenamefont
  {Steinberg}(2004)}]{resch2004extracting}%
  \BibitemOpen
  \bibfield  {author} {\bibinfo {author} {\bibfnamefont {K.~J.}\ \bibnamefont
  {Resch}}\ and\ \bibinfo {author} {\bibfnamefont {A.~M.}\ \bibnamefont
  {Steinberg}},\ }\bibfield  {title} {\bibinfo {title} {Extracting joint weak
  values with local, single-particle measurements},\ }\href
  {https://doi.org/10.1103/PhysRevLett.92.130402} {\bibfield  {journal}
  {\bibinfo  {journal} {Phys. Rev. Lett.}\ }\textbf {\bibinfo {volume} {92}},\
  \bibinfo {pages} {130402} (\bibinfo {year} {2004})}\BibitemShut {NoStop}%
\bibitem [{\citenamefont {Mitchison}\ \emph {et~al.}(2007)\citenamefont
  {Mitchison}, \citenamefont {Jozsa},\ and\ \citenamefont
  {Popescu}}]{mitchison2007sequential}%
  \BibitemOpen
  \bibfield  {author} {\bibinfo {author} {\bibfnamefont {G.}~\bibnamefont
  {Mitchison}}, \bibinfo {author} {\bibfnamefont {R.}~\bibnamefont {Jozsa}},\
  and\ \bibinfo {author} {\bibfnamefont {S.}~\bibnamefont {Popescu}},\
  }\bibfield  {title} {\bibinfo {title} {Sequential weak measurement},\ }\href
  {https://doi.org/10.1103/PhysRevA.76.062105} {\bibfield  {journal} {\bibinfo
  {journal} {Phys. Rev. A}\ }\textbf {\bibinfo {volume} {76}},\ \bibinfo
  {pages} {062105} (\bibinfo {year} {2007})}\BibitemShut {NoStop}%
\bibitem [{\citenamefont {Yunger~Halpern}(2017)}]{yunger_Halpern2017jarzynski}%
  \BibitemOpen
  \bibfield  {author} {\bibinfo {author} {\bibfnamefont {N.}~\bibnamefont
  {Yunger~Halpern}},\ }\bibfield  {title} {\bibinfo {title} {Jarzynski-like
  equality for the out-of-time-ordered correlator},\ }\href
  {https://doi.org/10.1103/PhysRevA.95.012120} {\bibfield  {journal} {\bibinfo
  {journal} {Phys. Rev. A}\ }\textbf {\bibinfo {volume} {95}},\ \bibinfo
  {pages} {012120} (\bibinfo {year} {2017})}\BibitemShut {NoStop}%
\bibitem [{\citenamefont {Rall}(2020)}]{rall2020quantum}%
  \BibitemOpen
  \bibfield  {author} {\bibinfo {author} {\bibfnamefont {P.}~\bibnamefont
  {Rall}},\ }\bibfield  {title} {\bibinfo {title} {Quantum algorithms for
  estimating physical quantities using block encodings},\ }\href
  {https://doi.org/10.1103/PhysRevA.102.022408} {\bibfield  {journal} {\bibinfo
   {journal} {Phys. Rev. A}\ }\textbf {\bibinfo {volume} {102}},\ \bibinfo
  {pages} {022408} (\bibinfo {year} {2020})}\BibitemShut {NoStop}%
\bibitem [{\citenamefont {Esposito}\ \emph {et~al.}(2009)\citenamefont
  {Esposito}, \citenamefont {Harbola},\ and\ \citenamefont
  {Mukamel}}]{esposito2009nonequilibrium}%
  \BibitemOpen
  \bibfield  {author} {\bibinfo {author} {\bibfnamefont {M.}~\bibnamefont
  {Esposito}}, \bibinfo {author} {\bibfnamefont {U.}~\bibnamefont {Harbola}},\
  and\ \bibinfo {author} {\bibfnamefont {S.}~\bibnamefont {Mukamel}},\
  }\bibfield  {title} {\bibinfo {title} {Nonequilibrium fluctuations,
  fluctuation theorems, and counting statistics in quantum systems},\ }\href
  {https://doi.org/10.1103/RevModPhys.81.1665} {\bibfield  {journal} {\bibinfo
  {journal} {Rev. Mod. Phys.}\ }\textbf {\bibinfo {volume} {81}},\ \bibinfo
  {pages} {1665} (\bibinfo {year} {2009})}\BibitemShut {NoStop}%
\bibitem [{\citenamefont {Piacentini}\ \emph {et~al.}(2016)\citenamefont
  {Piacentini}, \citenamefont {Avella}, \citenamefont {Levi}, \citenamefont
  {Gramegna}, \citenamefont {Brida}, \citenamefont {Degiovanni}, \citenamefont
  {Cohen}, \citenamefont {Lussana}, \citenamefont {Villa}, \citenamefont
  {Tosi}, \citenamefont {Zappa},\ and\ \citenamefont
  {Genovese}}]{piacentini2016measuring}%
  \BibitemOpen
  \bibfield  {author} {\bibinfo {author} {\bibfnamefont {F.}~\bibnamefont
  {Piacentini}}, \bibinfo {author} {\bibfnamefont {A.}~\bibnamefont {Avella}},
  \bibinfo {author} {\bibfnamefont {M.~P.}\ \bibnamefont {Levi}}, \bibinfo
  {author} {\bibfnamefont {M.}~\bibnamefont {Gramegna}}, \bibinfo {author}
  {\bibfnamefont {G.}~\bibnamefont {Brida}}, \bibinfo {author} {\bibfnamefont
  {I.~P.}\ \bibnamefont {Degiovanni}}, \bibinfo {author} {\bibfnamefont
  {E.}~\bibnamefont {Cohen}}, \bibinfo {author} {\bibfnamefont
  {R.}~\bibnamefont {Lussana}}, \bibinfo {author} {\bibfnamefont
  {F.}~\bibnamefont {Villa}}, \bibinfo {author} {\bibfnamefont
  {A.}~\bibnamefont {Tosi}}, \bibinfo {author} {\bibfnamefont {F.}~\bibnamefont
  {Zappa}},\ and\ \bibinfo {author} {\bibfnamefont {M.}~\bibnamefont
  {Genovese}},\ }\bibfield  {title} {\bibinfo {title} {Measuring incompatible
  observables by exploiting sequential weak values},\ }\href
  {https://doi.org/10.1103/PhysRevLett.117.170402} {\bibfield  {journal}
  {\bibinfo  {journal} {Phys. Rev. Lett.}\ }\textbf {\bibinfo {volume} {117}},\
  \bibinfo {pages} {170402} (\bibinfo {year} {2016})}\BibitemShut {NoStop}%
\bibitem [{\citenamefont {Nielsen}\ and\ \citenamefont
  {Chuang}(1997)}]{NielsenC97}%
  \BibitemOpen
  \bibfield  {author} {\bibinfo {author} {\bibfnamefont {M.~A.}\ \bibnamefont
  {Nielsen}}\ and\ \bibinfo {author} {\bibfnamefont {I.~L.}\ \bibnamefont
  {Chuang}},\ }\bibfield  {title} {\bibinfo {title} {Programmable quantum gate
  arrays},\ }\href {https://doi.org/10.1103/PhysRevLett.79.321} {\bibfield
  {journal} {\bibinfo  {journal} {Phys. Rev. Lett.}\ }\textbf {\bibinfo
  {volume} {79}},\ \bibinfo {pages} {321} (\bibinfo {year} {1997})}\BibitemShut
  {NoStop}%
\bibitem [{\citenamefont {O'Donnell}\ and\ \citenamefont
  {Wright}(2015)}]{odonnel2015quantumspectrum}%
  \BibitemOpen
  \bibfield  {author} {\bibinfo {author} {\bibfnamefont {R.}~\bibnamefont
  {O'Donnell}}\ and\ \bibinfo {author} {\bibfnamefont {J.}~\bibnamefont
  {Wright}},\ }\bibfield  {title} {\bibinfo {title} {Quantum spectrum
  testing},\ }in\ \href {https://doi.org/10.1145/2746539.2746582} {\emph
  {\bibinfo {booktitle} {Proceedings of the forty-seventh annual ACM symposium
  on Theory of computing}}}\ (\bibinfo  {publisher} {ACM},\ \bibinfo {year}
  {2015})\ pp.\ \bibinfo {pages} {529--538}\BibitemShut {NoStop}%
\bibitem [{\citenamefont {Greenberger}\ \emph {et~al.}(1989)\citenamefont
  {Greenberger}, \citenamefont {Horne},\ and\ \citenamefont
  {Zeilinger}}]{greenberger1989going}%
  \BibitemOpen
  \bibfield  {author} {\bibinfo {author} {\bibfnamefont {D.~M.}\ \bibnamefont
  {Greenberger}}, \bibinfo {author} {\bibfnamefont {M.~A.}\ \bibnamefont
  {Horne}},\ and\ \bibinfo {author} {\bibfnamefont {A.}~\bibnamefont
  {Zeilinger}},\ }\bibfield  {title} {\bibinfo {title} {Going beyond {B}ell's
  theorem},\ }in\ \href {https://doi.org/10.1007/978-94-017-0849-4_10} {\emph
  {\bibinfo {booktitle} {Bell's Theorem, Quantum Theory and Conceptions of the
  Universe}}}\ (\bibinfo  {publisher} {Kluwer Academic Publishers},\ \bibinfo
  {address} {Amsterdam, NL},\ \bibinfo {year} {1989})\ pp.\ \bibinfo {pages}
  {69--72}\BibitemShut {NoStop}%
\bibitem [{mol(2022)}]{molero2022github}%
  \BibitemOpen
  \href@noop {} {\bibinfo {title} {Link for the {N}ewton's identity algorithm
  that recovers the spectrum from the polynomials:}},\ \bibinfo {howpublished}
  {\url{https://github.com/Rumoa/example-fl}} (\bibinfo {year}
  {2022})\BibitemShut {NoStop}%
\bibitem [{\citenamefont {Bruzda}\ \emph {et~al.}(2009)\citenamefont {Bruzda},
  \citenamefont {Cappellini}, \citenamefont {Sommers},\ and\ \citenamefont
  {{\.{Z}}yczkowski}}]{bruzda2009random}%
  \BibitemOpen
  \bibfield  {author} {\bibinfo {author} {\bibfnamefont {W.}~\bibnamefont
  {Bruzda}}, \bibinfo {author} {\bibfnamefont {V.}~\bibnamefont {Cappellini}},
  \bibinfo {author} {\bibfnamefont {H.-J.}\ \bibnamefont {Sommers}},\ and\
  \bibinfo {author} {\bibfnamefont {K.}~\bibnamefont {{\.{Z}}yczkowski}},\
  }\bibfield  {title} {\bibinfo {title} {Random quantum operations},\ }\href
  {https://doi.org/10.1016/j.physleta.2008.11.043} {\bibfield  {journal}
  {\bibinfo  {journal} {Phys. Lett. A}\ }\textbf {\bibinfo {volume} {373}},\
  \bibinfo {pages} {320} (\bibinfo {year} {2009})}\BibitemShut {NoStop}%
\bibitem [{\citenamefont {Keyl}\ and\ \citenamefont
  {Werner}(2001)}]{keyl2001estimating}%
  \BibitemOpen
  \bibfield  {author} {\bibinfo {author} {\bibfnamefont {M.}~\bibnamefont
  {Keyl}}\ and\ \bibinfo {author} {\bibfnamefont {R.~F.}\ \bibnamefont
  {Werner}},\ }\bibfield  {title} {\bibinfo {title} {Estimating the spectrum of
  a density operator},\ }\href {https://doi.org/10.1103/PhysRevA.64.052311}
  {\bibfield  {journal} {\bibinfo  {journal} {Phys. Rev. A}\ }\textbf {\bibinfo
  {volume} {64}},\ \bibinfo {pages} {052311} (\bibinfo {year}
  {2001})}\BibitemShut {NoStop}%
\bibitem [{\citenamefont {Alicki}\ \emph {et~al.}(1988)\citenamefont {Alicki},
  \citenamefont {Rudnicki},\ and\ \citenamefont
  {Sadowski}}]{alicki1988symmetry}%
  \BibitemOpen
  \bibfield  {author} {\bibinfo {author} {\bibfnamefont {R.}~\bibnamefont
  {Alicki}}, \bibinfo {author} {\bibfnamefont {S.}~\bibnamefont {Rudnicki}},\
  and\ \bibinfo {author} {\bibfnamefont {S.}~\bibnamefont {Sadowski}},\
  }\bibfield  {title} {\bibinfo {title} {Symmetry properties of product states
  for the system of {$N$} $n$-level atoms},\ }\href
  {https://doi.org/10.1063/1.527958} {\bibfield  {journal} {\bibinfo  {journal}
  {J. Math. Phys.}\ }\textbf {\bibinfo {volume} {29}},\ \bibinfo {pages} {1158}
  (\bibinfo {year} {1988})}\BibitemShut {NoStop}%
\bibitem [{\citenamefont {Wright}(2016)}]{wright2016learn}%
  \BibitemOpen
  \bibfield  {author} {\bibinfo {author} {\bibfnamefont {J.}~\bibnamefont
  {Wright}},\ }\emph {\bibinfo {title} {How to learn a quantum state}},\ \href
  {http://reports-archive.adm.cs.cmu.edu/anon/2016/abstracts/16-108.html}
  {Ph.D. thesis},\ \bibinfo  {school} {Carnegie Mellon University} (\bibinfo
  {year} {2016})\BibitemShut {NoStop}%
\bibitem [{\citenamefont {O'Donnell}\ and\ \citenamefont
  {Wright}(2016)}]{O_Donnell2016}%
  \BibitemOpen
  \bibfield  {author} {\bibinfo {author} {\bibfnamefont {R.}~\bibnamefont
  {O'Donnell}}\ and\ \bibinfo {author} {\bibfnamefont {J.}~\bibnamefont
  {Wright}},\ }\bibfield  {title} {\bibinfo {title} {Efficient quantum
  tomography},\ }in\ \href {https://doi.org/10.1145/2897518.2897544} {\emph
  {\bibinfo {booktitle} {Proceedings of the forty-eighth annual {ACM} symposium
  on Theory of Computing}}}\ (\bibinfo  {publisher} {{ACM}},\ \bibinfo {year}
  {2016})\ pp.\ \bibinfo {pages} {899--912}\BibitemShut {NoStop}%
\bibitem [{\citenamefont {Haah}\ \emph {et~al.}(2016)\citenamefont {Haah},
  \citenamefont {Harrow}, \citenamefont {Ji}, \citenamefont {Wu},\ and\
  \citenamefont {Yu}}]{Haah2016}%
  \BibitemOpen
  \bibfield  {author} {\bibinfo {author} {\bibfnamefont {J.}~\bibnamefont
  {Haah}}, \bibinfo {author} {\bibfnamefont {A.~W.}\ \bibnamefont {Harrow}},
  \bibinfo {author} {\bibfnamefont {Z.}~\bibnamefont {Ji}}, \bibinfo {author}
  {\bibfnamefont {X.}~\bibnamefont {Wu}},\ and\ \bibinfo {author}
  {\bibfnamefont {N.}~\bibnamefont {Yu}},\ }\bibfield  {title} {\bibinfo
  {title} {Sample-optimal tomography of quantum states},\ }in\ \href
  {https://doi.org/10.1145/2897518.2897585} {\emph {\bibinfo {booktitle}
  {Proceedings of the forty-eighth annual {ACM} symposium on Theory of
  Computing}}}\ (\bibinfo  {publisher} {{ACM}},\ \bibinfo {year} {2016})\ pp.\
  \bibinfo {pages} {913--925}\BibitemShut {NoStop}%
\bibitem [{\citenamefont {Yuen}(2023)}]{yuen2023improved}%
  \BibitemOpen
  \bibfield  {author} {\bibinfo {author} {\bibfnamefont {H.}~\bibnamefont
  {Yuen}},\ }\bibfield  {title} {\bibinfo {title} {An improved sample
  complexity lower bound for (fidelity) quantum state tomography},\ }\href
  {https://doi.org/10.22331/q-2023-01-03-890} {\bibfield  {journal} {\bibinfo
  {journal} {Quantum}\ }\textbf {\bibinfo {volume} {7}},\ \bibinfo {pages}
  {890} (\bibinfo {year} {2023})}\BibitemShut {NoStop}%
\bibitem [{\citenamefont {Bacon}\ \emph {et~al.}(2006)\citenamefont {Bacon},
  \citenamefont {Chuang},\ and\ \citenamefont {Harrow}}]{bacon2006efficient}%
  \BibitemOpen
  \bibfield  {author} {\bibinfo {author} {\bibfnamefont {D.}~\bibnamefont
  {Bacon}}, \bibinfo {author} {\bibfnamefont {I.~L.}\ \bibnamefont {Chuang}},\
  and\ \bibinfo {author} {\bibfnamefont {A.~W.}\ \bibnamefont {Harrow}},\
  }\bibfield  {title} {\bibinfo {title} {Efficient quantum circuits for {S}chur
  and {C}lebsch-{G}ordan transforms},\ }\href
  {https://doi.org/10.1103/PhysRevLett.97.170502} {\bibfield  {journal}
  {\bibinfo  {journal} {Phys. Rev. Lett.}\ }\textbf {\bibinfo {volume} {97}},\
  \bibinfo {pages} {170502} (\bibinfo {year} {2006})}\BibitemShut {NoStop}%
\bibitem [{\citenamefont {Beverland}\ \emph {et~al.}(2018)\citenamefont
  {Beverland}, \citenamefont {Haah}, \citenamefont {Alagic}, \citenamefont
  {Campbell}, \citenamefont {Rey},\ and\ \citenamefont
  {Gorshkov}}]{Beverland2018}%
  \BibitemOpen
  \bibfield  {author} {\bibinfo {author} {\bibfnamefont {M.~E.}\ \bibnamefont
  {Beverland}}, \bibinfo {author} {\bibfnamefont {J.}~\bibnamefont {Haah}},
  \bibinfo {author} {\bibfnamefont {G.}~\bibnamefont {Alagic}}, \bibinfo
  {author} {\bibfnamefont {G.~K.}\ \bibnamefont {Campbell}}, \bibinfo {author}
  {\bibfnamefont {A.~M.}\ \bibnamefont {Rey}},\ and\ \bibinfo {author}
  {\bibfnamefont {A.~V.}\ \bibnamefont {Gorshkov}},\ }\bibfield  {title}
  {\bibinfo {title} {Spectrum estimation of density operators with
  {A}lkaline-{E}arth atoms},\ }\href
  {https://doi.org/10.1103/physrevlett.120.025301} {\bibfield  {journal}
  {\bibinfo  {journal} {Phys. Rev. Lett.}\ }\textbf {\bibinfo {volume} {120}},\
  \bibinfo {pages} {025301} (\bibinfo {year} {2018})}\BibitemShut {NoStop}%
\bibitem [{\citenamefont {Brun}(2004)}]{brun2004measuring}%
  \BibitemOpen
  \bibfield  {author} {\bibinfo {author} {\bibfnamefont {T.~A.}\ \bibnamefont
  {Brun}},\ }\bibfield  {title} {\bibinfo {title} {Measuring polynomial
  functions of states},\ }\href {https://doi.org/10.26421/QIC4.5-6} {\bibfield
  {journal} {\bibinfo  {journal} {Quantum Inf. Comput.}\ }\textbf {\bibinfo
  {volume} {4}},\ \bibinfo {pages} {401} (\bibinfo {year} {2004})}\BibitemShut
  {NoStop}%
\bibitem [{\citenamefont {Arnhem}\ \emph {et~al.}(2022)\citenamefont {Arnhem},
  \citenamefont {Griffet},\ and\ \citenamefont {Cerf}}]{arnhem2022multicopy}%
  \BibitemOpen
  \bibfield  {author} {\bibinfo {author} {\bibfnamefont {M.}~\bibnamefont
  {Arnhem}}, \bibinfo {author} {\bibfnamefont {C.}~\bibnamefont {Griffet}},\
  and\ \bibinfo {author} {\bibfnamefont {N.~J.}\ \bibnamefont {Cerf}},\
  }\bibfield  {title} {\bibinfo {title} {Multicopy observables for the
  detection of optically nonclassical states},\ }\href
  {https://doi.org/10.1103/PhysRevA.106.043705} {\bibfield  {journal} {\bibinfo
   {journal} {Phys. Rev. A}\ }\textbf {\bibinfo {volume} {106}},\ \bibinfo
  {pages} {043705} (\bibinfo {year} {2022})}\BibitemShut {NoStop}%
\bibitem [{\citenamefont {Wagner}\ and\ \citenamefont
  {Galv\~ao}(2023)}]{wagner2023simple}%
  \BibitemOpen
  \bibfield  {author} {\bibinfo {author} {\bibfnamefont {R.}~\bibnamefont
  {Wagner}}\ and\ \bibinfo {author} {\bibfnamefont {E.~F.}\ \bibnamefont
  {Galv\~ao}},\ }\bibfield  {title} {\bibinfo {title} {Simple proof that
  anomalous weak values require coherence},\ }\href
  {https://doi.org/10.1103/PhysRevA.108.L040202} {\bibfield  {journal}
  {\bibinfo  {journal} {Phys. Rev. A}\ }\textbf {\bibinfo {volume} {108}},\
  \bibinfo {pages} {L040202} (\bibinfo {year} {2023})}\BibitemShut {NoStop}%
\bibitem [{\citenamefont {Arvidsson-Shukur}\ \emph {et~al.}(2021)\citenamefont
  {Arvidsson-Shukur}, \citenamefont {Drori},\ and\ \citenamefont
  {Halpern}}]{arvidsson_Shukur2021conditions}%
  \BibitemOpen
  \bibfield  {author} {\bibinfo {author} {\bibfnamefont {D.~R.~M.}\
  \bibnamefont {Arvidsson-Shukur}}, \bibinfo {author} {\bibfnamefont {J.~C.}\
  \bibnamefont {Drori}},\ and\ \bibinfo {author} {\bibfnamefont {N.~Y.}\
  \bibnamefont {Halpern}},\ }\bibfield  {title} {\bibinfo {title} {Conditions
  tighter than noncommutation needed for nonclassicality},\ }\href
  {https://doi.org/10.1088/1751-8121/ac0289} {\bibfield  {journal} {\bibinfo
  {journal} {J. Phys. A}\ }\textbf {\bibinfo {volume} {54}},\ \bibinfo {pages}
  {284001} (\bibinfo {year} {2021})}\BibitemShut {NoStop}%
\bibitem [{\citenamefont {De~Bi{\`e}vre}(2021)}]{debievre2021complete}%
  \BibitemOpen
  \bibfield  {author} {\bibinfo {author} {\bibfnamefont {S.}~\bibnamefont
  {De~Bi{\`e}vre}},\ }\bibfield  {title} {\bibinfo {title} {Complete
  incompatibility, support uncertainty, and {K}irkwood-{D}irac
  nonclassicality},\ }\href {https://doi.org/10.1103/PhysRevLett.127.190404}
  {\bibfield  {journal} {\bibinfo  {journal} {Phys. Rev. Lett.}\ }\textbf
  {\bibinfo {volume} {127}},\ \bibinfo {pages} {190404} (\bibinfo {year}
  {2021})}\BibitemShut {NoStop}%
\bibitem [{\citenamefont {Bi{\`{e}}vre}(2023)}]{deBivre2023relating}%
  \BibitemOpen
  \bibfield  {author} {\bibinfo {author} {\bibfnamefont {S.~D.}\ \bibnamefont
  {Bi{\`{e}}vre}},\ }\bibfield  {title} {\bibinfo {title} {{Relating
  incompatibility, noncommutativity, uncertainty, and Kirkwood--Dirac
  nonclassicality}},\ }\href {https://doi.org/10.1063/5.0110267} {\bibfield
  {journal} {\bibinfo  {journal} {J. Math. Phys.}\ }\textbf {\bibinfo {volume}
  {64}},\ \bibinfo {pages} {022202} (\bibinfo {year} {2023})}\BibitemShut
  {NoStop}%
\bibitem [{\citenamefont {Xu}(2022{\natexlab{a}})}]{xu2022classification}%
  \BibitemOpen
  \bibfield  {author} {\bibinfo {author} {\bibfnamefont {J.}~\bibnamefont
  {Xu}},\ }\bibfield  {title} {\bibinfo {title} {Classification of
  incompatibility for two orthonormal bases},\ }\href
  {https://doi.org/10.1103/PhysRevA.106.022217} {\bibfield  {journal} {\bibinfo
   {journal} {Phys. Rev. A}\ }\textbf {\bibinfo {volume} {106}},\ \bibinfo
  {pages} {022217} (\bibinfo {year} {2022}{\natexlab{a}})}\BibitemShut
  {NoStop}%
\bibitem [{\citenamefont {Xu}(2022{\natexlab{b}})}]{xu2022kirkwood}%
  \BibitemOpen
  \bibfield  {author} {\bibinfo {author} {\bibfnamefont {J.}~\bibnamefont
  {Xu}},\ }\href {https://doi.org/10.48550/arXiv.2210.02876} {\bibinfo {title}
  {Kirkwood-{D}irac classical pure states}},\ \bibinfo {howpublished}
  {arXiv:2210.02876 [quant-ph]} (\bibinfo {year}
  {2022}{\natexlab{b}})\BibitemShut {NoStop}%
\bibitem [{\citenamefont {Bitter}\ and\ \citenamefont
  {Dubbers}(1987)}]{bitter1987manifestation}%
  \BibitemOpen
  \bibfield  {author} {\bibinfo {author} {\bibfnamefont {T.}~\bibnamefont
  {Bitter}}\ and\ \bibinfo {author} {\bibfnamefont {D.}~\bibnamefont
  {Dubbers}},\ }\bibfield  {title} {\bibinfo {title} {Manifestation of
  {B}erry's topological phase in neutron spin rotation},\ }\href
  {https://doi.org/10.1103/PhysRevLett.59.251} {\bibfield  {journal} {\bibinfo
  {journal} {Phys. Rev. Lett.}\ }\textbf {\bibinfo {volume} {59}},\ \bibinfo
  {pages} {251} (\bibinfo {year} {1987})}\BibitemShut {NoStop}%
\bibitem [{\citenamefont {Suter}\ \emph {et~al.}(1988)\citenamefont {Suter},
  \citenamefont {Mueller},\ and\ \citenamefont {Pines}}]{suter1988study}%
  \BibitemOpen
  \bibfield  {author} {\bibinfo {author} {\bibfnamefont {D.}~\bibnamefont
  {Suter}}, \bibinfo {author} {\bibfnamefont {K.~T.}\ \bibnamefont {Mueller}},\
  and\ \bibinfo {author} {\bibfnamefont {A.}~\bibnamefont {Pines}},\ }\bibfield
   {title} {\bibinfo {title} {Study of the {A}haronov-{A}nandan quantum phase
  by {NMR} interferometry},\ }\href
  {https://doi.org/10.1103/PhysRevLett.60.1218} {\bibfield  {journal} {\bibinfo
   {journal} {Phys. Rev. Lett.}\ }\textbf {\bibinfo {volume} {60}},\ \bibinfo
  {pages} {1218} (\bibinfo {year} {1988})}\BibitemShut {NoStop}%
\bibitem [{\citenamefont {Harris}\ \emph {et~al.}(2022)\citenamefont {Harris},
  \citenamefont {Yan},\ and\ \citenamefont {Sinitsyn}}]{harris2022benchmark}%
  \BibitemOpen
  \bibfield  {author} {\bibinfo {author} {\bibfnamefont {J.}~\bibnamefont
  {Harris}}, \bibinfo {author} {\bibfnamefont {B.}~\bibnamefont {Yan}},\ and\
  \bibinfo {author} {\bibfnamefont {N.~A.}\ \bibnamefont {Sinitsyn}},\
  }\bibfield  {title} {\bibinfo {title} {Benchmarking information scrambling},\
  }\href {https://doi.org/10.1103/PhysRevLett.129.050602} {\bibfield  {journal}
  {\bibinfo  {journal} {Phys. Rev. Lett.}\ }\textbf {\bibinfo {volume} {129}},\
  \bibinfo {pages} {050602} (\bibinfo {year} {2022})}\BibitemShut {NoStop}%
\bibitem [{\citenamefont {Hariri}\ \emph {et~al.}(2019)\citenamefont {Hariri},
  \citenamefont {Curic}, \citenamefont {Giner},\ and\ \citenamefont
  {Lundeen}}]{hariri2019experimental}%
  \BibitemOpen
  \bibfield  {author} {\bibinfo {author} {\bibfnamefont {A.}~\bibnamefont
  {Hariri}}, \bibinfo {author} {\bibfnamefont {D.}~\bibnamefont {Curic}},
  \bibinfo {author} {\bibfnamefont {L.}~\bibnamefont {Giner}},\ and\ \bibinfo
  {author} {\bibfnamefont {J.~S.}\ \bibnamefont {Lundeen}},\ }\bibfield
  {title} {\bibinfo {title} {Experimental simultaneous readout of the real and
  imaginary parts of the weak value},\ }\href
  {https://doi.org/10.1103/PhysRevA.100.032119} {\bibfield  {journal} {\bibinfo
   {journal} {Phys. Rev. A}\ }\textbf {\bibinfo {volume} {100}},\ \bibinfo
  {pages} {032119} (\bibinfo {year} {2019})}\BibitemShut {NoStop}%
\end{thebibliography}%

\appendix

\section{Formal comparison between standard weak measurement and circuit protocol}\label{appendix: cycle vs weak}

In this section, we estimate the sample complexity of finding weak values using the Bargmann invariants scheme and compare it with the sample complexity of measuring weak values using the standard protocol, known as weak measurement. Operationally, the two procedures are drastically different, which makes the comparison in terms of purely complexity arguments difficult. Therefore, for simplicity of the argument, we consider observables $A = \vert a \rangle \langle a \vert$. Recall that weak values can be written as
\begin{equation}
    A_w = \frac{\Delta_3(\phi, a, \psi)}{\Delta_2(\phi, \psi)}.
\end{equation}

\subsection{Same number of samples for numerator and denominator}\label{subsec: same number of samples}

In order to compute $A_w$ with the cycle test, we must compute the quantities in the numerator and denominator separately: using the SWAP-test for one and the cycle test associated with the operator $C_3$ for the other. We can focus solely on the complexity of estimating the real part $A_w$. The same analysis can be repeated for the imaginary part. Then, since $\Delta_2(\phi,\psi)$ is real, we must determine the sample complexity related to estimating the quantity $\text{Re}[A_w] = \text{Re}[\Delta_3(\phi,a,\psi)]/\Delta_2(\phi,\psi)$, i.e., we can focus on the real part of $\Delta_3(\phi,a,\psi)$ for the estimation of the numerator of $A_w$.

The results provided by runs of such tests can be described by random variables $X_i$ taking values in $\{\pm 1\}$, with $i$ denoting each run. Denoting by $p_+ = (1+\text{Re}[\Delta_3(\phi,a,\psi)])/2$ and $p_- = 1-p_+$ the probabilities of $X_i$ being $+1$ and $-1$, respectively, we write $X_i$'s expectation value as $\mathbb{E}[X_i] = +p_+-p_- = \text{Re}[\Delta_3(\phi,a,\psi)]$.

Assuming we have $N^{(3)}$ runs associated with $X_i$, where superscript $(3)$ refers to the fact that we are estimating third-order invariant, we consider the total random variable $X \defeq \sum_{i=1}^{N^{(3)}} X_i$. By linearity of $\mathbb{E}$, $\mathbb{E}[X] = N^{(3)} \text{Re}[\Delta_3(\phi,a,\psi)]$. We, then, apply Hoeffding's inequality, a standard method in learning theory and statistical analysis~\cite{hoeffding1994probability, shalev2014understanding}. According to this bound, if $a_i$ and $b_i$ are real constants such that $a_i \leq X_i \leq b_i$ for every $i\in\{1,\dots,N\}$ and $X = \sum_i X_i$, we have that, for any $t > 0$,
\begin{equation}
    \mathbb{P}[|X - \mathbb{E}[X]|\geq t] \leq 2 \exp\left(-\frac{2t^2}{\sum_{i=1}^{N^{(3)}}(b_i-a_i)^2}\right).
\end{equation}
In our case, $a_i=-1$ and $b_i=+1$ for every $i$, which allows us to rewrite the above expression as
\begin{equation}
    \mathbb{P}\left[\left|X - N^{(3)} \text{Re}[\Delta_3(\phi,a,\psi)]\right|\geq t\right] \leq 2 e^{-2t^2/4N^{(3)}},
\end{equation}
which implies that
\begin{equation}
    \mathbb{P}\left[\left |\frac{X}{N^{(3)}} -  \text{Re}[\Delta_3(\phi,a,\psi)]\right|\geq \frac{t}{N^{(3)}}\right] \leq 2 e^{-t^2/2N^{(3)}}.
\end{equation}

Defining $\delta^{(3)} \defeq 2e^{-t^2/2N^{(3)}}$, we get that $\delta^{(3)}/2 = e^{-t^2/2N^{(3)}}\implies t = \sqrt{2N^{(3)}\ln(2/\delta^{(3)})}$. Hence, we can ensure that, with probability greater than $1-\delta^{(3)}$, the estimator $X/N^{(3)}$ gives the true value of the real third-order invariant within an error of $\sqrt{2\ln(2/\delta^{(3)})/N^{(3)}}$, i.e.,
\begin{equation}
    \mathbb{P}\left[\left |\frac{X}{N^{(3)}} -  \text{Re}[\Delta_3(\phi,a,\psi)]\right|\geq \sqrt{\frac{2\ln(2/\delta^{(3)})}{N^{(3)}}}\right] \leq \delta^{(3)}.
\end{equation}

Similarly, we can introduce random variables $Y_i$ for $N^{(2)}$ runs of the SWAP test, each with expectation $\mathbb{E}[Y_i] = \Delta_2(\phi,\psi)$, and the respective counterpart of other relevant parameters to write
\begin{equation}
    \mathbb{P}\left[\left |\frac{Y}{N^{(2)}} -  \Delta_2(\phi,\psi)\right|\geq \sqrt{\frac{2\ln(2/\delta^{(2)})}{N^{(2)}}}\right] \leq \delta^{(2)},
\end{equation}
where we have used the fact that $\Delta_2(a,b) \equiv \text{Re}[\Delta_2(a,b)]$ for any pair of states $\vert a \rangle, \vert b \rangle \in \mathcal{H}$.

In our protocol, we have two different procedures: one in which we use $N^{(3)}$ samples to estimate $\text{Re}[\Delta_3(\phi,a,\psi)]$ and another in which we use $N^{(2)}$ samples to estimate $\Delta_2(\phi,\psi)$. Defining $\varepsilon^{(n)} \defeq \sqrt{2\ln(2/\delta^{(2)})/N^{(2)}}$, we can write
\begin{equation}
    \Delta_2(\phi,\psi)-\varepsilon^{(2)} \leq \frac{Y}{N^{(2)}}\leq \Delta_2(\phi,\psi)+\varepsilon^{(2)}
\end{equation}
or simply
\begin{equation}
    \left \vert \frac{Y}{N^{(2)}}\right\vert \geq \vert  \Delta_2(\phi,\psi) - \varepsilon^{(2)}\vert.
    \label{eq:bound-y}
\end{equation}

At first, let us assume that we have equal choices for the precision with a high probability of learning the related second- and third-order invariants, i.e., $\varepsilon^{(2)} = \varepsilon^{(3)} \equiv \varepsilon$ and $\delta^{(2)} = \delta^{(3)} \equiv \delta$. This implies that $N^{(2)} = N^{(3)} \equiv N$. Then, with a probability greater than $1-2\delta$, we have
\begin{equation}
    \begin{aligned}
        \left\vert \frac{X}{Y} \right. & \left.- \frac{\text{Re}[\Delta_3(\phi,a,\psi)]}{\Delta_2(\phi,\psi)}\right\vert = \\
        &= \left \vert \frac{\frac{X}{N}\Delta_2 - \frac{Y}{N}\text{Re}[\Delta_3]}{\frac{Y}{N} \Delta_2}\right \vert\\
        &=\left \vert \frac{\frac{X}{N}\Delta_2 -\Delta_2\text{Re}[\Delta_3]+\Delta_2\text{Re}[\Delta_3]- \frac{Y}{N}\text{Re}[\Delta_3]}{\frac{Y}{N} \Delta_2}\right \vert\\
        &\leq \frac{\Delta_2 \vert \frac{X}{N}-\text{Re}[\Delta_3]\vert + \vert \text{Re}[\Delta_3]\vert \vert \frac{Y}{N}-\Delta_2 \vert }{\Delta_2\frac{Y}{N}}\\
        &\leq \frac{\Delta_2 \varepsilon + \vert \text{Re}[\Delta_3]\vert\varepsilon}{\Delta_2 \vert \Delta_2 - \varepsilon\vert} = \frac{1+\vert \text{Re}[A_w] \vert }{\left\vert \frac{\Delta_2}{\varepsilon}-1\right\vert}.
    \end{aligned}
\end{equation}
In the last inequality, we have used that $\vert Y/N\vert \geq \vert \Delta_2 - \varepsilon \vert$, as seen in Eq.~\eqref{eq:bound-y}. In the above expression, we have also simplified the notation using $\Delta_2 \equiv \Delta_2(\phi,\psi)$ and $\Delta_3 \equiv \Delta_3(\phi,a,\psi)$. As a result, we conclude that
\begin{equation}
    \mathbb{P}\left[\left \vert \frac{X}{Y}-\text{Re}[A_w] \right\vert  \geq \frac{1+\vert \text{Re}[A_w] \vert }{\left\vert \frac{\Delta_2}{\varepsilon}-1\right\vert}\right] \leq 2\delta,
\end{equation}
i.e., the resulting accuracy $\varepsilon_T$ of estimating the real part of the weak value is
\begin{equation}
    \varepsilon_T = \frac{1+\vert \text{Re}[A_w] \vert }{\left\vert \frac{\Delta_2}{\varepsilon}-1\right\vert}.
\end{equation}
Since $\varepsilon = \sqrt{\ln(2/\delta)/N}$, we see that the number of samples needed scales as
\begin{equation}
    N = O\left(\frac{\ln(2/\delta)\vert\text{Re}[A_w] \vert^2}{\Delta_2^2 \varepsilon_T^2}\right).
\end{equation}

This allows us to compare the sample complexity of measuring weak values using the cycle test with the sample complexity associated with standard weak measurement schemes\footnote{It is noteworthy that, differently from the protocol we consider next, there are weak measurement schemes able to estimate both the real and imaginary parts of the weak value simultaneously~\cite{hariri2019experimental}.}.

In the standard weak measurement, one has a successful post-selection given by $N_s = \vert \langle \phi \vert \psi \rangle \vert^2 N$ with $N$ the total number of samples used. For the successful incidents, the measurement outcomes are assumed to be Gaussian distributed with variance satisfying $\sigma_i^2 \gg \gamma^2 \vert A_w \vert^2$. Then, for every $i\in\{1,\dots,N_s\}$, $\mathbb{E}[X_i] = \gamma \text{Re}[A_w]$. Also, $\text{Var}[X] = N_s\sigma_i^2$. Given that $\mathbb{E}[X/N_s \gamma] = \text{Re}[A_w]$ and because the distribution of successful incidents is Gaussian distributed, we have that the probability within a certain region is guaranteed by the error function
\begin{equation}
    \mathbb{P}\left[\left\vert \frac{X}{\gamma N_s} - \text{Re}[A_w] \right\vert \geq \varepsilon\right] \leq 1 - \text{erf}\left(\frac{\varepsilon}{\sqrt{2\text{Var}(X/\gamma N_s)}}\right).
\end{equation}
To simplify the notation, we use that $\text{Var}(X/\gamma N_s) = \text{Var}(X)/\gamma^2N_s^2 = \sigma_i^2/\gamma^2N_s$ and define 
\begin{equation}
    \delta \defeq \text{erfc}\left(\sqrt{\frac{\varepsilon^2 \gamma^2\Delta_2 N}{2\sigma_i^2}}\right),
\end{equation}
where $\text{erfc}(x) = 1-\text{erf}(x)$ is the complement of the error function. Then, the necessary total number of samples $N$ to estimate $\text{Re}[A_w]$ using $X/\gamma N_s$ is
\begin{equation}
    N = {2 (\text{erfc}^{-1}(\delta))^2\frac{\sigma_i^2}{\varepsilon^2 \gamma^2 \Delta_2} \gg } 2(\text{erfc}^{-1}(\delta))^2\frac{\vert A_w \vert^2}{\varepsilon^2 \Delta_2},
\end{equation}
which gives an order of samples better than the case of the cycle test when $\Delta_2 \to 0$ since, in this case,
\begin{equation}
    N^{(\text{weak})} = O\left(\text{erfc}^{-1}(\delta))^2\frac{\vert A_w \vert^2}{\varepsilon^2 \Delta_2}\right).
\end{equation}

This implies that there might be no complex theoretic advantage in using the cycle test protocol to implement weak value amplification, for which $\Delta_2(\phi,\psi)$ must be extremely small. However, this does not exclude the possibility that, for specific setups, there might be benefits in using the cycle test---e.g., in cases where the constant $C$ related to the weak coupling is exponentially big. As pointed out in the main text, for simply witnessing \textit{nonclassical values}, the cycle test \textit{does} provide a sample complexity advantage since one does not need to estimate $\Delta_2(\phi,\psi)$. In particular, this feature might be interesting in specific setups of quantum sensing using the imaginary part of weak values.

\subsection{Different number of samples for numerator and denominator}\label{subsec: different number of samples}

Given that we expect $\Delta_2(\phi,\psi)\ll 1$ in various applications of weak values, e.g., weak value amplification, it is relevant to consider the use of more samples to estimate the overlap such that $N^{(2)}>N^{(3)}$. We can, therefore, analyze this particular situation to observe if (and how) the sample complexity of estimating the weak value changes.

If we do not assume that $N^{(2)} = N^{(3)}$, we have
\begin{widetext}
\begin{equation}
    \begin{aligned}
        &\left\vert \frac{X}{Y}-\frac{N^{(3)}\text{Re}[\Delta_3(\phi,a,\psi)]}{N^{(2)}\Delta_2(\phi,\psi)}\right\vert = \left \vert \frac{N^{(2)}\Delta_2(\phi,\psi)X-N^{(3)}\text{Re}[\Delta_3(\phi,a,\psi)]Y}{YN^{(2)}\Delta_2(\phi,\psi)}\right\vert=\\
        &=\left\vert \frac{N^{(2)}\Delta_2(\phi,\psi)X-N^{(2)}N^{(3)}\Delta_2(\phi,\psi)\text{Re}[\Delta_3(\phi,a,\psi)]+N^{(2)}N^{(3)}\Delta_2(\phi,\psi)\text{Re}[\Delta_3(\phi,a,\psi)]-N^{(3)}\text{Re}[\Delta_3(\phi,a,\psi)]Y}{Y N^{(2)}\Delta_2(\phi,\psi)}\right\vert\\&\leq 
        \frac{\left \vert N^{(2)}\Delta_2(\phi,\psi)X-N^{(2)}N^{(3)}\Delta_2(\phi,\psi)\text{Re}[\Delta_3(\phi,a,\psi)]\right \vert + \left \vert N^{(2)}N^{(3)}\Delta_2(\phi,\psi)\text{Re}[\Delta_3(\phi,a,\psi)]-N^{(3)}\text{Re}[\Delta_3(\phi,a,\psi)]Y\right \vert }{Y N^{(2)}\Delta_2(\phi,\psi)}\\
        &=\frac{N^{(2)}\Delta_2(\phi,\psi)\vert X - N^{(3)}\text{Re}[\Delta_3(\phi,a,\psi)]\vert + N^{(3)}\vert \text{Re}[\Delta_3(\phi,a,\psi)]\vert \vert N^{(2)}\Delta_2(\phi,\psi)-Y\vert}{Y N^{(2)}\Delta_2(\phi,\psi)}.
    \end{aligned}
\end{equation}
Hoeffding's inequality leads to
\begin{equation}
    \mathbb{P}[\vert X - N^{(3)}\text{Re}[\Delta_3(\phi,a,\psi)]\vert \geq \sqrt{2N^{(3)}\ln(2/\delta^{(3)})}]\leq \delta^{(3)}
\end{equation}
and
\begin{equation}
    \mathbb{P}[\vert Y - N^{(2)}\Delta_2(\phi,\psi)\vert \geq \sqrt{2N^{(2)}\ln(2/\delta^{(2)})}]\leq \delta^{(2)},
\end{equation}
implying that, with probability $\geq 1-\delta^{(2)}$,
\begin{equation}
    \frac{1}{Y}\leq \frac{1}{N^{(2)}\Delta_2(\phi,\psi)-\sqrt{2N^{(2)}\ln(2/\delta^{(2)})}} \approx \frac{1}{N^{(2)}\Delta_2(\phi,\psi)}
\end{equation}
for a sufficiently large $N^{(2)}$. In turn, this leads to
\begin{equation}
    \begin{aligned}
        &\left\vert \frac{X}{Y}-\frac{N^{(3)}\text{Re}[\Delta_3(\phi,a,\psi)]}{N^{(2)}\Delta_2(\phi,\psi)}\right\vert \leq \frac{N^{(2)}\Delta_2(\phi,\psi)\sqrt{2N^{(3)}\ln(2/\delta^{(3)})} + N^{(3)}\vert \text{Re}[\Delta_3(\phi,a,\psi)]\vert \sqrt{2N^{(2)}\ln(2/\delta^{(2)})}}{N^{(2)}N^{(3)}\Delta_2(\phi,\psi)^2}\\
        &\left\vert \frac{X}{Y}-\frac{N^{(3)}\text{Re}[\Delta_3(\phi,a,\psi)]}{N^{(2)}\Delta_2(\phi,\psi)}\right\vert \leq \frac{\sqrt{N^{(3)}}}{N^{(2)}\Delta_2(\phi,\psi)}\sqrt{2\ln(2/\delta^{(3)})} + \frac{N^{(3)}}{(N^{(2)})^{3/2}\Delta_2(\phi,\psi)^2}\vert \text{Re}[\Delta_3(\phi,a,\psi)]\vert\sqrt{2\ln(2/\delta^{(2)})},
    \end{aligned}
\end{equation}
from which we conclude that
\begin{equation}
    \left \vert \frac{\frac{X}{N^{(3)}}}{\frac{Y}{N^{(2)}}} - \frac{\text{Re}[\Delta_3(\phi,a,\psi)]}{\Delta_2(\phi,\psi)}\right \vert = \frac{N^{(2)}}{N^{(3)}}\left \vert \frac{X}{Y}-\frac{N^{(3)}\text{Re}[\Delta_3(\phi,a,\psi)]}{N^{(2)}\Delta_2(\phi,\psi)}\right \vert \leq \underbrace{\frac{\sqrt{2\ln(2/\delta^{(3)})}}{\sqrt{N^{(3)}}\Delta_2(\phi,\psi)} + \frac{\vert \text{Re}[\Delta_3(\phi,a,\psi)]\vert\sqrt{2\ln(2/\delta^{(2)})}}{\sqrt{N^{(2)}}\Delta_2(\phi,\psi)^2}}_{=: \varepsilon}.
\end{equation}
Defining $\delta \defeq \delta^{(2)}+\delta^{(3)}$, we write
\begin{equation}
    \mathbb{P}\left[\left \vert \frac{\frac{X}{N^{(3)}}}{\frac{Y}{N^{(2)}}}-\frac{\text{Re}[\Delta_3(\phi,a,\psi)]}{\Delta_2(\phi,\psi)}\right \vert \geq \varepsilon\right]\leq \delta.
\end{equation}

Since we have already analyzed the case in which $N^{(2)}=N^{(3)}$, we now focus on the case in which they differ. Let us capture this difference with a parameter $c$, letting $N^{(2)} \sim N^{(3)}/\Delta_2(\phi,\psi)^c$. Since our main interest lies in scenarios $\Delta_2(\phi,\psi)$ is assumed to be small, we want to study the situation $N^{(2)}\gg N^{(3)}$ to see if we have some improvement in the sample complexity. Then, we assume $c$ is positive and $\Delta_2(\phi,\psi)^2\ll 1$. As a result, the total number of samples satisfies $N \defeq N^{(2)} + N^{(3)} \simeq N^{(3)}/\Delta_2(\phi,\psi)^c$. This implies that
\begin{equation}
    \varepsilon = \frac{\sqrt{2\ln(2/\delta^{(3)})}}{\sqrt{N^{(3)}}\Delta_2(\phi,\psi)} + \frac{\vert \text{Re}[\Delta_3(\phi,a,\psi)]\vert\sqrt{2\ln(2/\delta^{(2)})}}{\sqrt{N^{(2)}}\Delta_2(\phi,\psi)^2} = \frac{\sqrt{2(\ln(4/\delta))}}{\sqrt{N}}\left(\frac{1}{\Delta_2(\phi,\psi)^{1+c/2}} + \frac{\vert \text{Re}[\Delta_3(\phi,a,\psi)]\vert}{\Delta_2(\phi,\psi)^2}\right),
\end{equation}
from which we can extract the dependency in the number of samples with precision $\varepsilon$ in learning the real part of the weak value:
\begin{equation}
    N \simeq \frac{2\ln(4/\delta)}{\varepsilon^2}\left(\frac{1}{\Delta_2(\phi,\psi)^{1+c/2}} + \frac{\vert \text{Re}[\Delta_3(\phi,a,\psi)]\vert}{\Delta_2(\phi,\psi)^2}\right)^2.
\end{equation}
\end{widetext}

Finally, let us study the number of samples by considering the possible values of the parameter $c$:
\begin{enumerate}
    \item Case $c=2$: We have $1+c/2=2$. Then, we are left with an order of number of samples equal to the one when we had $N^{(2)}=N^{(3)}$ since $N=\frac{2(\ln(4/\delta))}{\varepsilon^2\Delta_2(\phi,\psi)^4}(1+\vert \text{Re}[\Delta_3(\phi,a,\psi)]\vert)^2$. Note that, in the results presented in the previous subsection, the number of samples refers to the estimation of the real part of the weak value while, here, it refers to the estimation of the third-order invariant. This is the reason for the difference seen in the order of the overlap.
    \item Case $0<c<2$: Assuming that $\vert \text{Re}[\Delta_3(\phi,a,\psi)]\vert = O(1)$ and that it does not scale with $\Delta_2(\phi,\psi)$, it follows that $\vert \text{Re}[\Delta_3(\phi,a,\psi)]\vert/\Delta_2(\phi,\psi)\gg 1/\Delta_2(\phi,\psi)^{1 + c/2}$ for small $\Delta_2(\phi,\psi)$. After substituting these approximations into the equation for $N$, we are left with the same sample complexity as before.
    \item Case $c>2$: Direct calculation leads to $N \simeq \frac{2(\ln(4/\delta))}{\varepsilon^2\Delta_2(\phi,\psi)^{2+c}}$, which is worse in sample complexity than the previous results, and worsens the larger the value the parameter $c$ takes.
\end{enumerate}

In summary, these calculations show that the sample complexity does not improve statistically if one estimates the overlaps with a larger, or even much larger, number of samples than the number of samples necessary to estimate the third-order invariant in the measurement of weak values.

\section{Sample and measurement complexity of estimating the spectrum using the cycle test}\label{appendix: samples spectrum}

\subsection{Proof of Theorem~\ref{theorem: sample complexity learning spectrum}}\label{subsec: proof of theorem sample spectrum}

To estimate $\{ \Tr(\rho^2), \hdots, \Tr(\rho^d) \}$, we want to learn with high probability $1-\delta$ the entire set of estimators within an error $\varepsilon$. Using the union bound,
\begin{equation}
    \begin{aligned}
        &\mathbb{P}\left[ \bigcup_{i=2}^d \left\{ \hat{\rho^i} \; \colon | \Tr(\hat{\rho^i})- \Tr(\rho^i)| \geq \varepsilon  \right\}  \right]  \\
        &\leq \sum_{i=2}^{d}  \mathbb{P} \left[  | \Tr(\hat{\rho^i})- \Tr(\rho^i)| \geq \varepsilon  \right] \\
        &= (d-1)   \mathbb{P} \left[  | \Tr(\hat{\rho^i})- \Tr(\rho^i)| \geq \varepsilon  \right].
    \end{aligned}
\end{equation}
We fix $\mathbb{P} \left[|\Tr(\hat{\rho^i})- \Tr(\rho^i)| \geq \varepsilon\right] = \delta/(d-1)$ such that the total error probability is $\delta$.

\begin{figure}
    \centering
    \includegraphics[width=\columnwidth]{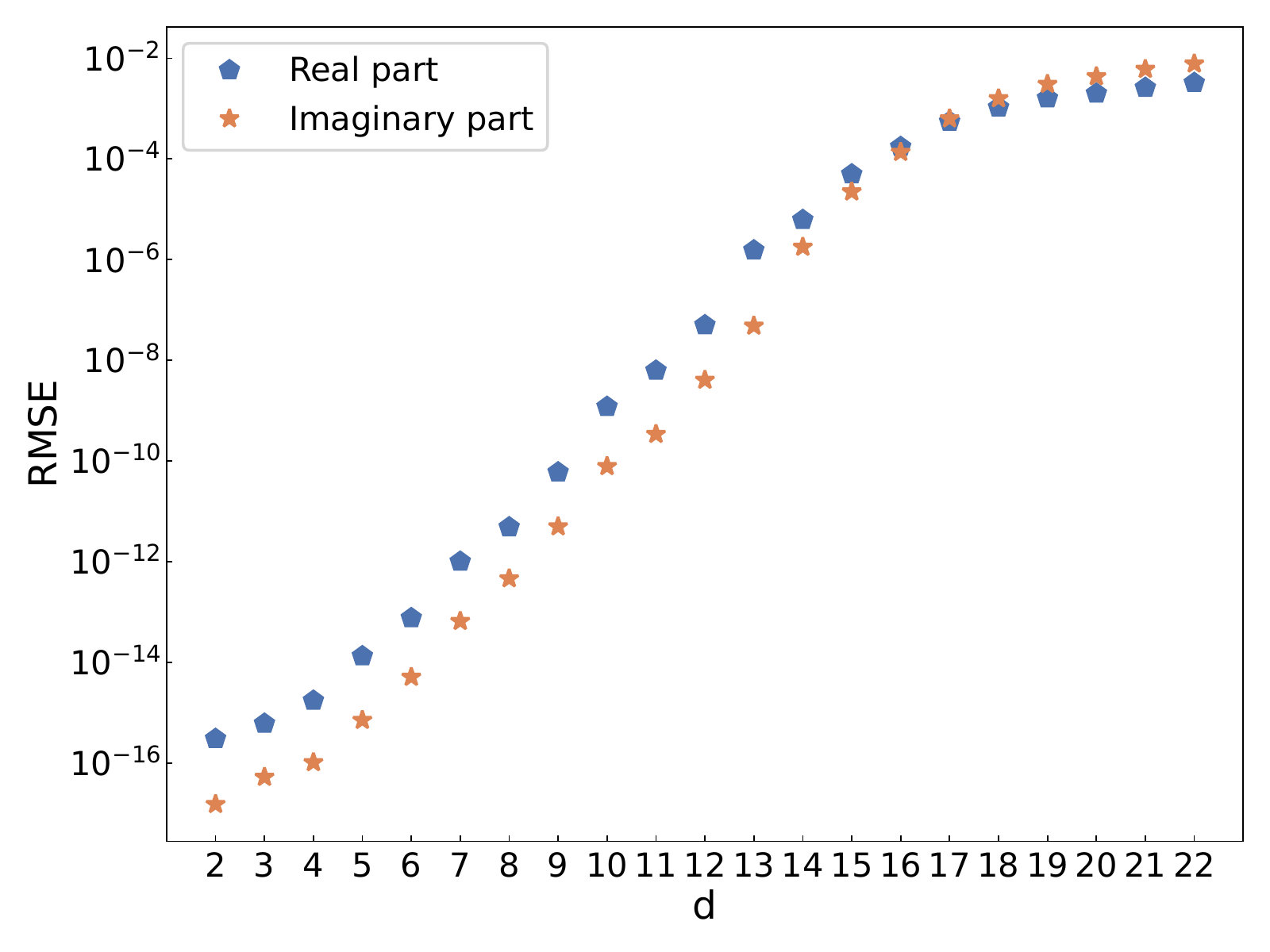}
    \caption{\textbf{Aggregated root-mean-squared error (RMSE) spectrum estimation using the Faddeev-LeVerrier Algorithm.} Eigenvalues of mixed states computed from multivariate traces $\{\text{Tr}(\rho^n)\}_{n=1}^d$ are compared with the ones obtained by diagonalization of the density matrix using \textit{numpy}. For each dimension, a data set of random 1000 density matrices $\rho$ is used. We sample the states  using the Ginibre ensemble.}
    \label{fig:dimension Fa-Le re_im}
\end{figure}

From Hoeffding's inequality, for $N$ independent samples $X_i \in [a, b]$ and mean estimator, the  number of \textit{statistical} samples needed to estimate each trace is, as already mentioned in Appendix \ref{appendix: cycle vs weak}, $O(\ln(2/\delta)/\varepsilon^2)$. But each statistical sample requires $O(d^2)$ states to compute $\{\Tr(\rho^2), \hdots, \Tr(\rho^d)\}$ using the cycle test. The total number of states is, therefore, $Nd^2 = \frac{d^2}{\varepsilon^2} \ln\sqrt{2 (d-1)/\delta}$ and, as a consequence, $N = O\left(\frac{d^2}{\varepsilon^2}\ln(d/\delta)\right)$ samples estimate the chosen set of traces to precision $\varepsilon$ with probability $1-\delta$. Since the number of measurements equals the number of statistical samples times $d$, which accounts for one measurement for each element in the set of traces, we have that the number of measurements is of order $O\left(\frac{d}{\varepsilon^2} \ln\left(d/\delta\right)\right)$. This concludes the proof.

In the main text, we reported the sample complexity $\tilde{O}(d^2/\varepsilon^2)$ that hides $\log(d)$ factors and fixed $\delta$ to be a reference value to prescribe what is meant by ``high probability,'' usually taken to be equal to $1-\delta = 2/3$.

\subsection{Numerical analysis for the  Faddeev--LeVerrier algorithm}\label{subsec: numerical analysis FV}

In this section, we conduct a numerical study of the behavior of the Faddeev--LeVerrier algorithm~\cite[Chapter 3]{escofier2012galois}, which returns the estimates of the eigenvalues of $\rho$, given $\{\Tr(\rho^i)\}_{i=2}^{d}$ as an input.

\begin{figure}
    \centering
    \includegraphics[width=\columnwidth]{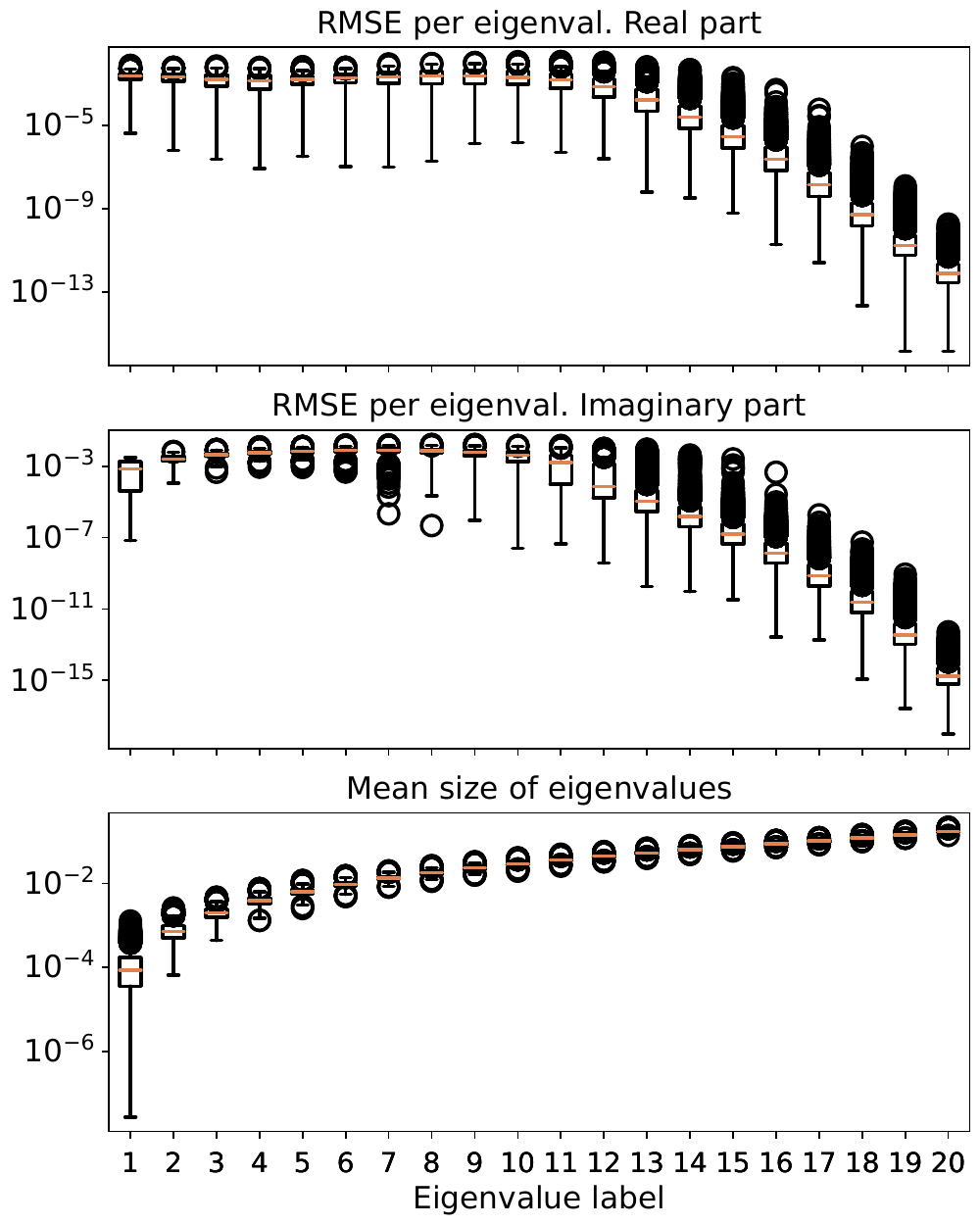}
    \caption{\textbf{Mean size of eigenvalues and RMSE per eigenvalue}. For 1000 random mixed density matrices of fixed Hilbert space dimension 20, we study the real and imaginary parts of the spectrum given from the Faddeev--LeVerrier algorithm.}
    \label{fig:spectrum size}
\end{figure}

In Fig.~\ref{fig:dimension Fa-Le re_im}, using noiseless experimental data, i.e., the traces $\{\Tr(\rho^i)\}_i$, we show the aggregated RMSE of the real and imaginary parts of the spectrum of $\rho$. We compare the predicted spectrum with the one computed via diagonalization with \textit{numpy}.

Ideally, the spectrum should be real-valued. However, as mentioned in the main text, numerical instabilities generate imaginary parts in the outputs of the algorithm. At a certain dimension, the RMSE of the imaginary part of all eigenvalues \textit{exceeds} the one for the real part. This is expected, as numerical inaccuracies occur when very small roots need to be estimated. As a result, the aggregated RMSE for high dimensions can be orders of magnitude bigger than the actual size of some eigenvalues. In what follows, we study if the appearance of these imaginary values, as well as their size, significantly influences the comparison between the real part of the spectrum and the ideal spectrum.

Let us take as a case study the error \textit{per eigenvalue} for a sufficiently large Hilbert space dimension. In Fig.~\ref{fig:spectrum size}, we show the average size of the eigenvalues using 1000 random density matrices of dimension 20. The biggest eigenvalues have sizes of order $10^{-2}$ while the associated error between the real and imaginary part is smaller than $10^{-13}$. On the other hand, the smallest eigenvalues have an error approximately equal to the size of the eigenvalue, which explains the aggregated RMSE from Fig.~\ref{fig:dimension Fa-Le re_im} in higher dimensions.

\begin{figure}
    \centering
    \includegraphics[width=\columnwidth]{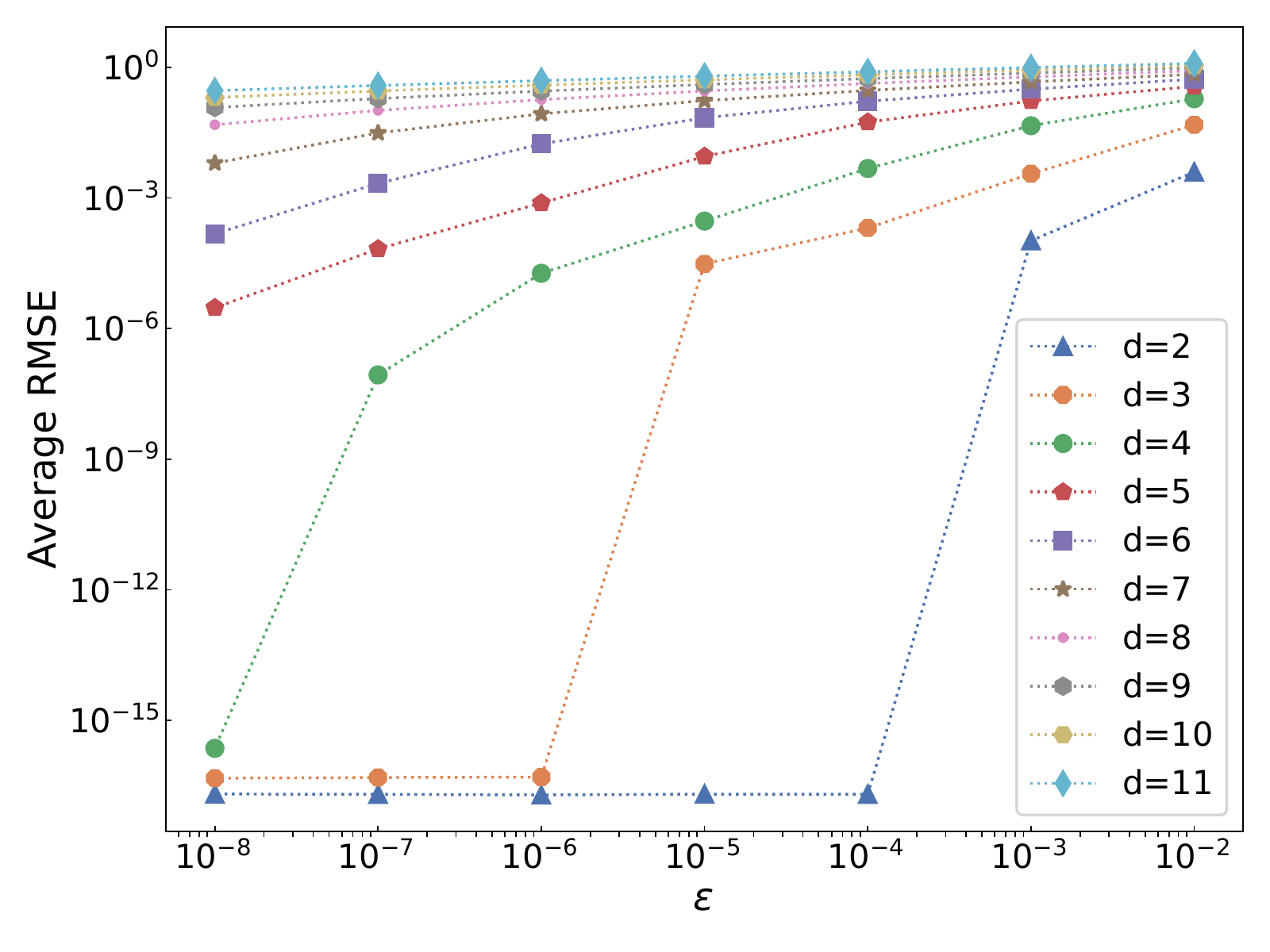}
    \caption{
    \textbf{Average RMSE of the imaginary parts of the spectrum estimation of 5000 random $\rho$.} This figure represents the imaginary counterpart of Fig.~\ref{fig:spectrum MSE}. Each matrix generates 1000 noisy samples to compute the RMSE.}
    \label{fig:spectrum imaginary_noisy}
\end{figure}

Although the predicted eigenvalues can contain non-negligible imaginary parts, their presence does not affect the prediction of the real parts. These imaginary values come from the presence of small coefficients in the characteristic polynomial. In Fig.~\ref{fig:spectrum imaginary_noisy}, we show the RMSE of the imaginary part in the estimation of the spectrum using noisy data. The sudden increase for a certain dimension is closely related to the amount of noise introduced. For larger dimensions, smaller amounts of noise modify the coefficients enough to produce the appearance of imaginary values. However, by comparing with Fig.~\ref{fig:spectrum MSE}, the trend of the RMSE of real parts is not affected by these ``bumps'' in the imaginary parts, allowing us to safely discard them from the predictions.

\subsubsection{Learning the largest eigenvalue}

As smaller eigenvalues produce bigger errors, the Faddeev--LeVerrier algorithm can be used to estimate only the largest eigenvalues. In Fig.~\ref{fig:spectrum RMSE largest}, we show the RMSE of the real part in the prediction of the largest eigenvalue using the same noisy dataset as the one used for the prediction of the entire spectrum. Since the largest eigenvalue depends mostly on the coefficients of the highest order, the size of these makes them less sensitive to the noise and allows the obtaining of a RMSE of $10^{-3}$ in dimension 6 with an error $\varepsilon=10^{-4}$ for the traces. However, as stated in the main text, computing the entire set of traces is not recommended in this task as higher powers in the traces of $\rho$ are more difficult to estimate and contribute less to the largest eigenvalue.

\begin{figure}
    \centering
    \includegraphics[width=\columnwidth]{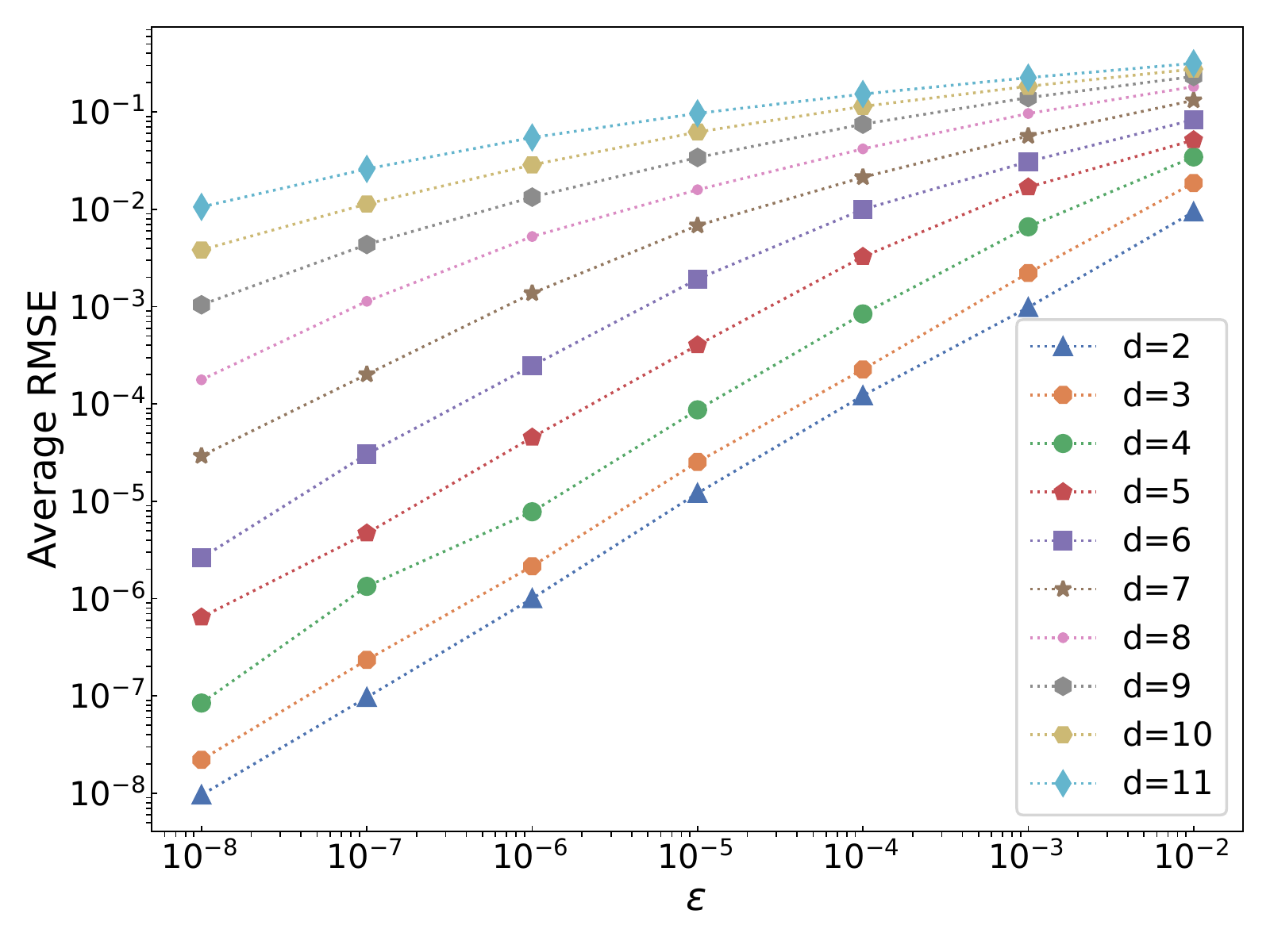}
    \caption{\textbf{Average RMSE of the real part of largest eigenvalue for 5000 random density matrices.} Each matrix generates 1000 noisy samples to compute the RMSE. We see that, when focusing on the largest eigenvalue, the accuracy of the estimation improves when compared to the results for the entire spectrum estimation from Fig.~\ref{fig:spectrum MSE} in the main text.}
    \label{fig:spectrum RMSE largest}
\end{figure}

\section{Proof of Lemma \ref{lemma: negativity from overlaps}}\label{appendix: proof of the lemma}

Consider a general triplet of pure states $\{\vert \psi_1 \rangle, \vert \psi_2 \rangle, \vert \psi_3 \rangle\}$ real with respect to some basis and choose the following parametrization for it:
\begin{equation}
    \begin{aligned}
        &\vert \psi_1\rangle = \vert 0 \rangle \\
        &\vert \psi_2 \rangle = \cos(\beta)\vert 0 \rangle + \sin(\beta)\vert 1 \rangle \\
        &\vert \psi_3 \rangle = \cos(\gamma)\vert 0 \rangle + \sin(\gamma)\sin(\alpha)\vert 1 \rangle + \sin(\gamma)\cos(\alpha)\vert 2 \rangle,
    \end{aligned}
\end{equation}
where $\beta,\gamma \in [0,\pi]$ and $\alpha \in [0,2\pi)$. Under this choice, made without loss of generality, we can define the function
\begin{figure*}
    \centering
    \includegraphics[width=\textwidth]{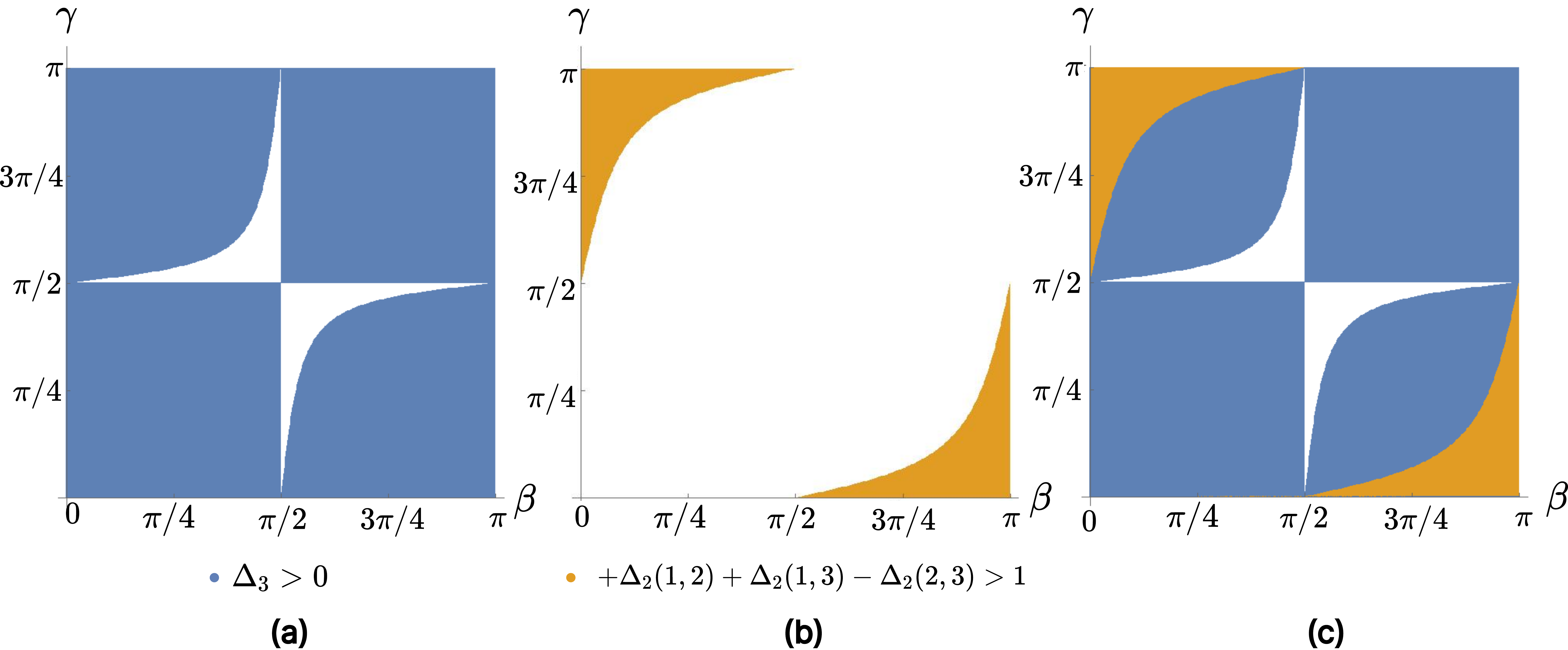}
    \caption{\textbf{Witnesses of third-order positivity.} (a) The white regions in the center are the area corresponding to $\Delta_{3}\leq 0$. The blue region is given by $\Delta_{3}(\alpha,\beta,\gamma)>0$ for $\alpha = 0.11$. (b) The light-brown  region overlaps with the blue region completely, and it is given by $h_3(\alpha,\beta,\gamma)-1>0$, showing that this condition implies $\Delta_{3}(\alpha,\beta,\gamma)>0$. However, the converse does not necessarily hold. In fact, it is true only for $\alpha \in [0,\pi]$, as can be seen in (c), where we combine both configurations (a) and (b).}
    \label{fig: lemma figure}
\end{figure*}
\begin{equation}
    \begin{aligned}
        h_{1}(\alpha,\beta,\gamma) = -\Delta_{2}(\psi_1,\psi_2)+\Delta_{2}(\psi_1,\psi_3)+\Delta_{2}(\psi_2,\psi_3),\\
        h_2(\alpha,\beta,\gamma) = +\Delta_{2}(\psi_1,\psi_2)-\Delta_{2}(\psi_1,\psi_3)+\Delta_{2}(\psi_2,\psi_3),\\
        h_3(\alpha,\beta,\gamma) = +\Delta_{2}(\psi_1,\psi_2)+\Delta_{2}(\psi_1,\psi_3)-\Delta_{2}(\psi_2,\psi_3),
    \end{aligned}
\end{equation}
where each overlap depends on the variables $\alpha,\beta,\gamma$. We can, then, study what the constraints $h_{i}>1$ represent, i.e.,
\begin{equation}
    \begin{aligned}
        &h_1 > 1 \implies \\
        &0<-1-\cos^2(\beta)+\cos^2(\gamma)+\cos^2(\beta)\cos^2(\gamma)+\\
        &+\sin^2(\beta)\sin^2(\gamma)\underbrace{\sin^2(\alpha)}_{\in [0,1]}+\frac{1}{2}\sin(2\beta)\sin(2\gamma)\sin(\alpha)\\
        &\leq-1-\cos^2(\beta)+\cos^2(\gamma)+\cos^2(\beta)\cos^2(\gamma)+\\
        &\hspace{0.3cm}+\sin^2(\beta)\sin^2(\gamma)+\frac{1}{2}\sin(2\beta)\sin(2\gamma)\sin(\alpha)\\
        &=-2\cos^2(\beta)+2\cos^2(\beta)\sin^2(\gamma)\\
        &\hspace{1.6cm}+\frac{1}{2}\sin(2\beta)\sin(2\gamma)\sin(\alpha)\\
        &=-2\cos^2(\beta)+2\Delta_{3}(\alpha,\beta,\gamma)\leq 2\Delta_{3}(\alpha,\beta,\gamma),
    \end{aligned}
\end{equation}
which implies that $\Delta_{3}(\alpha,\beta,\gamma) \equiv \Delta_3(\psi_1,\psi_2,\psi_3) > 0$, where we have used that
\begin{equation}
    \begin{aligned}
        \Delta_{3}(\alpha,\beta,\gamma) &= \langle \psi_1 \vert \psi_2 \rangle \langle  \psi_2 \vert \psi_3 \rangle \langle \psi_3 \vert \psi_1 \rangle \\
        &= \cos^2(\beta)\cos^2(\gamma) + \frac{1}{4}\sin(2\beta)\sin(2\gamma)\sin(\alpha).
    \end{aligned}
\end{equation}
The same reasoning applies for $h_2 > 1$. For $h_3>1$, the situation is less straightforward since the two functions $h_3(\alpha,\beta,\gamma)-1$ and $\Delta_{3}(\alpha,\beta,\gamma)$ are not totally ordered in the domain $[0,2\pi) \times [0,\pi]^2$. However, we can analyze what $\Delta_{3}(\alpha,\beta,\gamma)>0$ corresponds to in terms of the parameters $\beta$ and $\gamma$ when $\alpha$ lies in some domains. From the above, it follows that
\begin{equation}\label{eq: bargmann positivity implication}
    \tan(\beta)\tan(\gamma)>-\frac{1}{\sin(\alpha)} \implies \Delta_{3}(\alpha,\beta,\gamma)>0
\end{equation}
if $\alpha \in [0,\pi]$. The case in which $\alpha \in [\pi,2\pi]$ can be treated similarly. The function $h_3(\alpha,\beta,\gamma)-1$ can be written as
\begin{equation}
    \begin{aligned}
        &h_3(\alpha,\beta,\gamma)-1 =\\
        &=-1 +\cos^2(\beta) + \cos^2(\gamma)-\cos^2(\beta)\cos^2(\gamma)\\
        &-\sin^2(\beta)\sin^2(\gamma)\sin^2(\alpha)-\frac{1}{2}\sin(2\beta)\sin(2\gamma)\sin(\alpha)\\
        &=-\sin^2(\beta)\sin^2(\gamma)-\sin^2(\beta)\sin^2(\gamma)\sin^2(\alpha)\\
        &\hspace{1.5cm}-\frac{1}{2}\sin(2\beta)\sin(2\gamma)\sin(\alpha)\\
        &=-\sin^2(\beta)\sin^2(\gamma)(1+\sin^2(\alpha))\\
        &\hspace{1.5cm}-\frac{1}{2}\sin(2\beta)\sin(2\gamma)\sin(\alpha).
    \end{aligned}
\end{equation}
The condition that $h_3(\alpha,\beta,\gamma)-1>0$ implies that $\sin^2(\beta)\sin^2(\gamma)(1+\sin^2(\alpha))<-\sin(2\beta)\sin(2\gamma)\sin(\alpha)/2$.

\begin{lemma}
    For any $\alpha,\beta,\gamma \in [0,\pi]$, if
    \begin{equation}
        \sin^2(\beta)\sin^2(\gamma)(1+\sin^2(\alpha))<-\frac{1}{2}\sin(2\beta)\sin(2\gamma)\sin(\alpha),
    \end{equation}
    then
    \begin{equation}
        \tan(\beta)\tan(\gamma)>-\frac{1}{\sin(\alpha)}.
    \end{equation}
\end{lemma}

\begin{proof}
    This lemma can be easily seen to hold numerically. In the range of $\alpha,\beta,\gamma \in [0,\pi]$, we see that there is always an overlap between the region for which the first inequality is satisfied simultaneously with the second inequality. The converse is not necessarily true. 
\end{proof}

The above lemma, together with Eq.~\eqref{eq: bargmann positivity implication}, implies that $h_3-1>0 \implies \Delta_{3}>0$, as we wanted to show. The case $\alpha \in [\pi,2\pi]$ is treated similarly, as already mentioned.

Fixing $\alpha =0.11$, Fig.~\ref{fig: lemma figure} shows the region of values $(\beta,\gamma) \in [0,\pi]$ for which the third-order invariant positivity can be witnessed { with} the overlap inequality. { This provides a visual way to see that} $\Delta_{2}(\psi_1,\psi_2)+\Delta_{2}(\psi_1,\psi_3)-\Delta_{2}(\psi_2,\psi_3)>1$ implies { positivity of the third-order Bargmann invariant,} while the {converse} does not always hold.

\end{document}